\documentclass[11pt]{article} 
 
\title{Modelling Extremal Dependence for Operational Risk by a Bipartite Graph} 
 
\author{Oliver Kley\thanks{Center for Mathematical Sciences, Technische Universit\"at M\"unchen,  Boltzmannstr. 3, 85748 Garching, Germany, 
email: oliver.kley@tum.de,\,cklu@tum.de}
\and Claudia Kl\"uppelberg\footnotemark[1]
\and Sandra Paterlini\thanks{Department of Management and Economics, University of Trento, Italy, 
email: sandra.paterlini@ebs.edu, sandra.paterlini@unitn.it}}

\usepackage{rotating}
\usepackage{amssymb}
\usepackage{amsfonts}
\usepackage{amsmath}
\usepackage{amsthm}
\usepackage[authoryear]{natbib}
\usepackage{graphicx}
\usepackage{caption}
\usepackage{subfigure}
\usepackage[usenames, dvipsnames]{xcolor}
\usepackage{verbatim}
\usepackage{dsfont}
\usepackage{xcolor}
\usepackage[utf8]{inputenc}
\usepackage{hyperref}
\usepackage{comment}
 \usepackage{tikz}
\usepackage{chngcntr}
\usepackage{csvsimple}
 \usepackage{graphicx}
 \usepackage{mathptmx}      

\usetikzlibrary{positioning,chains,fit,shapes,calc}

\numberwithin{equation}{section}

\newtheorem{theorem}{Theorem}[section]
\newtheorem{lemma}[theorem]{Lemma}
\newtheorem{remark}[theorem]{Remark}
\newtheorem{example}[theorem]{Example}
\newtheorem{proposition}[theorem]{Proposition}
\newtheorem{definition}[theorem]{Definition}

\newtheorem{corollary}[theorem]{Corollary}
\newtheorem{fig}[theorem]{Figure}

\newcommand{\bthe}{\begin{theorem}}
\newcommand{\ethe}{\end{theorem}}

\newcommand{\ben}{\begin{enumerate}}
\newcommand{\een}{\end{enumerate}}

\newcommand{\bit}{\begin{itemize}}
\newcommand{\eit}{\end{itemize}}

\newcommand{\beq}{\begin{equation}}
\newcommand{\eeq}{\end{equation}}

\newcommand{\ble}{\begin{lemma}}
\newcommand{\ele}{\end{lemma}}

\newcommand{\bde}{\begin{definition}\rm}
\newcommand{\ede}{\halmos\end{definition}}

\newcommand{\bco}{\begin{corollary}}
\newcommand{\eco}{\end{corollary}}

\newcommand{\bpr}{\begin{proposition}}
\newcommand{\epr}{\end{proposition}}

\newcommand{\brem}{\begin{remark}\rm}
\newcommand{\erem}{\end{remark}}

\newcommand{\bproof}{\begin{proof}}
\newcommand{\eproof}{\end{proof}}

\newcommand{\bexam}{\begin{example}\rm}
\newcommand{\eexam}{\end{example}}

\newcommand{\bfi}{\begin{fig}}
\newcommand{\efi}{\end{fig}}

\newcommand{\btab}{\begin{tab}}
\newcommand{\etab}{\end{tab}}

\newcommand{\beao}{\begin{eqnarray*}}
\newcommand{\eeao}{\end{eqnarray*}\noindent}

\newcommand{\beam}{\begin{eqnarray}}
\newcommand{\eeam}{\end{eqnarray}\noindent}

\newcommand{\barr}{\begin{array}}
\newcommand{\earr}{\end{array}}

\newcommand{\bdis}{\begin{displaymath}}
\newcommand{\edis}{\end{displaymath}\noindent}

\def\S{{\mathbb S}}
\def\P{{\mathbb P}}
\def\E{{\mathbb E}}
\def\R{{\mathbb R}}

\newcommand{\pr}[1]{\P\left[#1\right]}

\def\1{\mathds{1}}

\newcommand{\stv}{\stackrel{v}{\rightarrow}}

\newcommand{\al}{{\alpha}}

\newcommand{\ga}{{\gamma}}

\newcommand{\MRV}{{\mbox{\rm MRV}}}

\newcommand{\VaR}{{\rm VaR}}

\newcommand{\ov}{\overline}
\newcommand{\wh}{\widehat}

\newcommand{\CoTE}{{\rm CoTE}}

\newcommand{\ET}{{$ET$}}
\newcommand{\BL}{{$BL$}}

\newcommand{\halmos}{\quad\hfill\mbox{$\Box$}}  

\newcommand{\red}{\textcolor{red}}

\allowdisplaybreaks[4]

\begin{document}


\maketitle

\begin{abstract}
We introduce a  statistical model for operational losses based on heavy-tailed distributions and bipartite graphs, which  captures the event type and business line structure of operational risk data. The model explicitly takes into account the Pareto tails of losses and the heterogeneous dependence structures between them. 
We then derive estimators for individual as well as aggregated tail risk, measured in terms of  Value-at-Risk and Conditional-Tail-Expectation for very high confidence levels, and provide also an asymptotically full capital allocation method.
Estimation methods for such  tail risk measures and capital allocations are also proposed and tested on simulated data. 
Finally, by having access to real-world operational risk losses from the Italian banking system,  we show that even with a small number of observations, the proposed estimation methods produce reliable estimates, and that quantifying dependence by means of the empirical network  has a big  impact on estimates at both individual and aggregate level, as well as for capital allocations.
\end{abstract}



\section{Introduction}\label{introduction}

Operational risk is one of the core risks for the banking system with potentially dangerous spillover effects for the entire financial and economic sector. The LIBOR scandal, the rogue trader from Soci\'{e}t\'{e} G\'{e}n\'{e}rale, the 9/11 terrorist attack, the Madoff fraud are all by now classical examples of operational risk losses of extremely large magnitude. 
In 2017, the ten largest operational losses worldwide exceeded \$11.6 billion as reported in \citet{topten}. 
The largest loss was caused by fraudulent transactions at the Brazilian development bank totalling \$2.52 billion. 
Secondly, employees at the Shoko Chukin Bank in Japan improperly granted \$2.39 billion of loans by falsifying approval
documents. In the third place, the U.S. Securities and Exchange Commission brought charges against the Woodbridge Group of Companies with running a \$1.22 billion Ponzi scheme.
The awareness of their potential threat to the financial system and the expectation of potentially even larger ones, such as from IT failures or cyber attacks, has prompted regulators to introduce further regulations as well as to improve the existing ones in an attempt to set up  tools to better monitor and control losses, while setting aside adequate capital provisions. 
For example, in 2015 the European Banking Authority (EBA)  released the  final draft regulatory standard for the institutions that are allowed to use Advanced Measurement Approaches (AMA) for operational risk in \citet{EBA2015}, which typically are the largest banks. 
The document \citet{EBA2014}  on common procedures and methodologies for the Supervisory Review and Evaluation Process (SREP) 
better clarifies many aspects related to the definition and measurements of operational losses, while setting up common principles for  regulation. There is a clear intent to promote further comparability across banks, avoid spurious differences due to modelling and measurement choices as well as to balance-off the need to have  models that are understandable, while still capable of providing a realistic picture of the operational risk within the banking system, thereby discarding too simplistic assumptions (e.g. Gaussian copulas to capture dependence).

In March 2016,  the Basel Commitee  provided  guidelines for the new Standardized Approach (SA) \citet{BCBS2013} and \citet{BCBS2014},  with the final intent to find a better balance between simplicity, comparability and sensitivity across banks. Such an approach moves away from sophisticated statistical modelling, and will become the regulatory standard from January 2022 onwards; cf. \citet{Basel2017} and \citet[Section 4]{EBA2017}.

In spite of the regulatory discussion in recent years, the largest banks have invested resources in setting up operational risk departments to better analyze the sources of risks, define their risk appetite, and improve their risk management systems as well as developing possible hedging or insurance strategies for the often rare but  extremely large losses that could destabilize the banks and the entire system. This strategy is also supported by \citet{KPMG}, who criticise the new SA approach as it is ``likely to reduce the incentive for banks to strengthen their operational risk management''. They continue with ``Banks should ... focus on models that support decision-making and the running of the business.''

Moreover,  some institutions,  such as  Credit Suisse, trade securities related to the extreme tail-risk of their operational risk losses, requiring them to be able to adequately price such instruments. Most likely, even if the regulatory framework will change, banks will keep investing resources into adequate operational risk quantification, as such losses might be a thread for the survival of the institutions with potential catastrophic effects on the entire system. Furthermore, there is a clear understanding of the need of setting up appropriate pricing, hedging and insurance strategies to adequately sell and share such risks.    

This motivates, in spite of the simple SMA required for regulatory purposes, the development of realistic models for operational risk for better understanding and quantitative assessment. Although each company may now develop its own operational risk model, independent of principles of the former AMA approach, we remain with its basic structure. This is for comparison with previous studies, but also since the model has been developed for good reasons. The obvious features of operational losses data are the following.
They are typically characterized by high frequency/low severity and   low frequency/high severity losses\red{;}Extreme Value Theory and Pareto-like tail distributions are possibly the most widely used tools  to model the marginal distribution of  operational risk losses and to provide reliable estimates for tail-related risk measures, such as Value-at-Risk and Conditional-Tail-Expectation at very high confidence level; a (non-exhaustive) list of references includes books like \citet{Shev2, EKM1997, McNeilFreyEmbrechts2005, Shev}, and the articles \citet{BockerKluppelberg2010, ChavezDemoulinEmbrechtsNeslehova2006, ChavezDemoulinetal2015}). 
 Since the beginning of operational risk research, the topic of measuring dependence among operational losses and explictly taking it into account for risk-capital estimates has received large attention, also because the regulatory assumption of perfect correlation is considered too stringent  (\citet{BrechmannCzadoPaterlini2014, CopeAntonini2008, DallaValleFantazziniGiudici2008, Frachotetal2004b, MittnikPaterliniYener2011, RachediFantazzini2009}). 
 Finally, being capable of  setting up monitoring tools to detect, which business lines are typically most affected by operational risk losses, is high on the agenda, especially for the industry in order to develop effective \red{risk} management strategies within an institution.

Operational losses are typically classified in a matrix of 56 risk classes (seven event types ($ET$) $\times$ eight business lines ($BL$)\footnote{Business lines: 1. Corporate Finance, 2. Trading and Sales, 3. Retail Banking, 4. Commercial Banking, 5. Payment and Settlement, 6. Agency and Custody, 7. Asset Management, 8. Retail Brokerage. Event types: 1. Internal Fraud, 2. External Fraud, 3. Employment Practices \& Workplace Safety, 4. Clients, Products \& Business Practices, 5. Damage to Physical Assets, 6. Business Disruption \& System Failures, 7. Execution, Delivery \& Process Management.}; see \citet{Basel2006}).  

Data are typically scarce, with large numbers of zero losses for some cells 
and short time series,  characterized by few heavy or extreme tails. 
Losses in different $BL$s or $ET$s can also show heterogeneous tail dependence (i.e. large losses occur jointly and bivariate extreme dependencies occur with different probabilities). 
According to \citet{Basel2009}, before 2010 most banks have focused on using bivariate copulas to capture dependencies among losses, with a special focus on Gaussian copulas. 
However, due to the regulation \citet{EBA2015} in 2015, and in recognition of the inability of Gaussian copulas to capture tail dependence,  which  might have a strong direct impact on the capital estimates, different copula families, such as Archimedian and Student-$t$, have become popular. 
For instance, \citet{BrechmannCzadoPaterlini2014} have introduced a 7- and 8-dimensional flexible multivariate model that relies on different copula families, pointing out empirically  the diversification benefit in terms of risk capital reduction, when explicitly modelling dependence. 
Such empirical work  provides insights along the lines of  the 2015 introduced regulation for AMA banks (\citet{EBA2015}), which explicitly requires to replace  
Gaussian copulas in favor of alternative dependence modelling tools such as Student-$t$ copulas, which can detect tail dependence.
 
In fact, \citet{BrechmannCzadoPaterlini2014}  have shown that a flexible model of vine copulas with individual Student-$t$ copulas, which imply different tail dependencies in upper and lower tails, is often preferred to alternative elliptical or Archimedian specification. 
However, their approach does not explictly take into account the possible presence of extremes in the marginal distributions of $ET$s and $BL$s. 
Furthermore, \citet{BockerKluppelberg2010} present a structural approach to multivariate operational risk modelling using L\'{e}vy copulas building on a previously developed single cell model for heavy tailed operational risk losses, see \citet{BKrisk}. 
Work like \citet{BernardVanDuffel2016} and \citet{Embrechtsetal2015}  (as well as references therein) also show how attractive the possibility of deriving asymptotic and ideally ``tight" bounds for risk-capital under different dependence specifications as well as measuring model risk and uncertainty is. 
Still, most of the work so far builds on the use of correlation or copulas to capture the dependence structure between operational risk classes.

In this work, we  introduce a new tractable model, which relies on extreme value theory and bipartite graphs to model the  dependence structure of the 56-dimensional distribution of the $ET-BL-$matrix. 
Such an approach has the advantage of allowing us to derive asymptotic closed-form formulas, as well as bounds for the tail-related risk measures. The model  not only quantifies the risk for the system of risk sources, as captured by the $ET-BL-$matrix, but also at individual business line/event type level, providing an important tool to gain insights into the risk allocation within the financial institution itself. Moreover, as the $ET-BL-$matrix is modelled as a network, the presence of dependence between the different cells can immediately be depicted, pointing out the relevant linkages between the different $BL$s and $ET$s and capturing their dependence, while naturally interpreting the lack of observations as missing links in the bipartite graph. Here, we also provide a practical implementation of the estimator by introducing two semi-parametric estimation procedures and show their effectiveness on simulated data under different scenarios. 
Finally, by using real-world operational risk data, provided by the DIPO consortium (Database Italiano Perdite Operative---Italian Banking Association Consortium), we test our new model and can provide relevant information on the impact that dependence modelling plays with respect to the risk-capital estimates at both system and individual level.  Capturing dependence explicitly will then allow to set up more accurate risk-management tools to better monitor and hedge potential extreme risks.     

{The paper is structured as follow.} Section~\ref{intuitive} defines the model and provides an intuitive explanation.
Here, also a detailed procedure can be found of how the estimation of the tail risk measures and the risk contributions is performed. This section also reports the main theoretical results based on multivariate regular variation for Pareto-like tail distributions and derives the risk-measures. Some mathematical details and all proofs are postponed to the Appendix.
 Section~\ref{opRiskData}  describes the DIPO operational risk data and provides the results of the estimation procedure.
Section~\ref{Sec:SimulationStudy}  reports results on simulated data from two different scenarios and compares them with the real data. Finally, Section~\ref{sect:conclusion} concludes.


\section{The Model: Bipartite Graph and Distributional Assumptions}\label{intuitive}

 According to the Basel regulation \citet{Basel2006}, each operational risk loss is typically classified according to both the event type of the loss and the business line to which the loss belongs. Here, we assume data are operational risk losses aggregated over a fixed period of time (which could be a week or a month). 
These  operational losses are partitioned into  seven event types $ET=(ET_{1},\dots,ET_{7})^\top$ and eight business lines $BL=(BL_1,\dots,BL_8)^\top$ categories. 
For example, $ET_{1}$ is the sum of all losses, no matter which business line they belong to, due to Internal Fraud during this period, while $BL_{1}$ is the sum of all losses in the business line Corporate Finance during the same period, no matter which event type they belong to. 
Consequently,  after classifying losses, we need to take into account 56 loss variables $L_{ij}$  $i=1,\dots,8$ (corresponding to business lines) and $j=1,\dots,7$ (corresponding to event types), which are the components of the $ET-BL-$matrix:
\begin{center}
\begin{tabular}{c|cccc}
  
                                     & $ET_{1}$     &   $\ldots$       & $ET_{7}$ \\\hline
                    &                  & \\
                 $BL_{1}$      &  $L_{11}$   &  $\ldots$        & $L_{17}$       \\
                                       &                    &                  & \\
                     $\vdots$          &  $\vdots$          &   $\ddots$               &  $\vdots$    \\
                                       &                    &                  & \\
                 $BL_{8}$   & $L_{81}$    &  $\ldots$        &$L_{87}$ 
\label{tabular}        
\end{tabular}
\vspace{0.3cm}
\end{center}
 
$L_{11}$ are then all losses due to Internal Fraud in the business unit Corporate Finance. 
Here, we think of the loss event types $ET=(ET_{1},\dots,ET_{7})^\top$ as being observed random variables, which are as a rule heavy tailed, possessing a particular extremal depedence structure to be explored later on. 
These event type losses are distributed to the relevant business lines, where
we observe that each business line $BL_i$ has to carry a certain \textit{fraction} $A_{ij}$ of the loss of an event type $ET_j$ with $\sum_{i=1}^{8} A_{ij}=1$.
The total loss can then be computed  as 
\beam\label{Sum}
\sum_{i=1}^8\sum_{j=1}^{7}L_{ij}=BL_{1}+\dots+BL_{8}=\sum_{i=1}^{8}\sum_{j=1}^{7}A_{ij}ET_{j}=ET_{1}+\dots+ET_{7}.
\eeam

Each single cell loss $L_{ij}$  can therefore be expressed as $L_{ij}=A_{ij}ET_{j}$ for some $i,j$ with $A_{ij}$ being the fraction of the loss in  $ET_j$ that is assigned to  $BL_i$; i.e., $A_{ij} = L_{ij}/(L_{1j} + \dots + L_{8j}).$ Therefore, after defining the so-called \textit{network fraction matrix} $A=(A_{ij})$, the business line losses are the result of the matrix-vector-multiplication
\begin{gather}\label{BPhiT}
BL=A \times ET.
\end{gather}
The matrix $A$ can be interpreted as a weighted adjacency matrix of a {\em bipartite random graph} (see e.g. \citet[Section~6.6]{Newman}).
It consists of two groups of nodes: the event types and the business lines, where a zero at position $(i,j)$  in the matrix $A$ means that there is no edge in the  graph between $BL_{i}$ and $ET_{j}$, implying that there is no loss in the single cell of the $ET-BL-$matrix. 
Figure~\ref{bipartite} shows a possible realization of a bipartite random graph, where for example the losses of $ET_4$ have occurred for a proportion equal to $A_{34}$ in $BL_3$ and for  $A_{74}$ in $BL_7$ with $A_{34}+A_{74}=1$. 
The representation in \eqref{BPhiT}  then incorporates the full information originally given by the $ET-BL-$matrix. 

{\small
\begin{figure}[t]
\begin{center}
\begin{tikzpicture}[thin,
  fsnode/.style={},
  ssnode/.style={},
  every fit/.style={ellipse,draw,inner sep=2pt,text width=5cm},
]
 \tikzstyle{every node}=[font=\small]
\begin{scope}[start chain=going right,node distance=5mm]
\foreach \i /\xcoord in {1/1,2/2,3/3,4/4,5/5,6/6,7/7,8/8}
  \node[fsnode,on chain, fill=orange,draw=none, circle,text=white] (f\i) {$BL_{\xcoord}$};
\end{scope}
\begin{scope}[xshift=1.5 cm,yshift=-2.5cm,start chain=going right,node distance=5mm]
\foreach \i/\xcoord in {9/1,10/2,11/3,12/4,13/5,14/6,15/7}
  \node[ssnode,on chain,fill=blue,draw=none, circle,text=white] (s\i) {$ET_{\xcoord}$};
\end{scope}
\path
    (f1) edge node [left] {} (s9);
              
\path    (s9) edge node [right] {} (f2);
        
\path    (f3) edge node [right] {} (s10);
    
\path		(f3) edge node [right] {$A_{34}$} (s12)	;	
				
\path    (s12) edge node [left] {$A_{74}$} (f7);
		
 \path   (s13) edge node [left] {} (f4);
		
	\path	(s13) edge node [left] {} (f5);
		
\path		(s13) edge node [left] {} (f8);
		
	\path	  (f6) edge node [left] {} (s15);
\end{tikzpicture}
\end{center}
\caption{\label{bipartite} Realization of the bipartite random graph, where a link  between business line $BL_{i}$ and event type $ET_j$ indicates a loss in the cell $(i,j)$ of size $A_{ij} ET_j$, here shown with exemplary weights $A_{34}$ and $A_{74}$. } 
\end{figure}
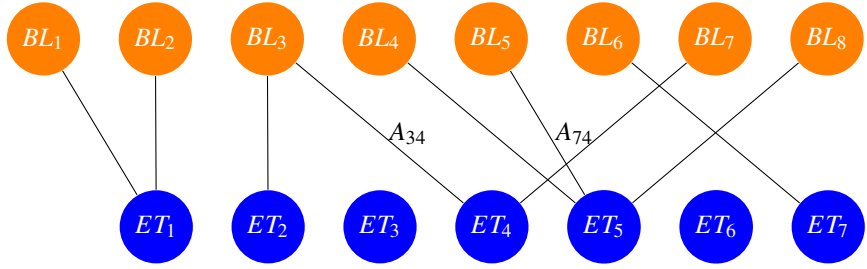 
}

The model relies on two assumptions. First, we consider the distribution of the random vector $ET=(ET_1,\dots, ET_7)$ to be {\em multivariate regularly varying}. 
The following definition (cf. \citet[Theorem~6.1]{Resnick2007}) gives an intuition of the features of the distribution.
An equivalent definition and more details are reported in  Appendix~\ref{Appendixmodel}. 

We start with a random vector\footnote{In the operational risk context $X$ corresponds to the vector of event types and $d=7$.} $X=(X_1,\dots,X_d)^\top$, whose one-dimensional marginals are all positive and heavy-tailed in the sense of regular variation, which means that for all $j=1,\dots,d$,
\begin{align}\label{marginrv}
\lim_{x\to\infty} \frac{\P(X>tx)}{\P(X>x)} = t^{-\al}
\end{align}
for all $t>0$ and some {\em tail index} $\al>0$.

Multivariate regular variation of $X$ models the marginals $X_j$ for $j=1,\dots,d$ as heavy-tailed random variables as in \eqref{marginrv}, and the dependence structure between the components separately by means of a random vector $\Theta$ with values in the positive unit sphere $\S_+^{d-1}$ of $\R^d$. 
The sphere can be taken with respect to any norm, examples are 
$\S_+^{d-1} = \{x\in\R^d_+ : \max_{1\le j\le d} x_j =1\}$,
$\S_+^{d-1} = \{x\in\R^d_+ : \sum_{j=1}^d x_j =1\}$, or $\S_+^{d-1} = \{x\in\R_+^d : \sqrt{\sum_{j=1}^d x_j^2}=1\}$.
More precisely, for every $t>0$ and every  set $S\subset\S_+^{d-1}$ satisfying the weak regularity condition that it is a Borel set and its boundary has probability zero (e.g. an interval on the sphere),
the norm $\|X\|$ and the normalized vector $X/\|X\|$ (which has values on the sphere) have the property that 
\begin{align}\label{multrv}
\lim_{x\to\infty}\frac{\P(\|X\|>tx, X/\|X\|\in S)}{\P(\|X\|>x)} = t^{-\al} \P(\Theta\in S) =: t^{-\al} \Gamma(S).
\end{align}
We want to emphasize that in the limit the norm of the vector $\|X\|$, which is one-dimensional regularly varying as in \eqref{marginrv}, and the dependence structure given by the distribution of $\Theta=X/\|X\|$ with values on the sphere, become independent.
The probability measure $\Gamma$ is called {\em spectral measure} or {\em angular measure}.

As a statistically senseful model we choose the event types to have Pareto-like tails, which satisfy \eqref{marginrv}, more precisely, we assume that 
$\lim_{t\to\infty} t^\al \pr{ET_j>t } = K_j>0$ for $j=1,\dots,d$, which we write as
\begin{align}\label{Pareto}
\pr{ET_j>t } & \sim K_j t^{-\alpha}, \quad  t\to \infty,
\end{align}
with {\em scaling parameters} $K_j>0$ and a common {\em tail index} $\alpha>0$.
{Before we discuss the implications and limitations of this assumption for the operational risk application, we explain its basic properties.}

For this regularly varying operational risk model none of the event types can be regarded  unimportant---or vice versa---none of the event types will ultimatively dominate the others. 
This is expressed in \eqref{Pareto} by having the same tail index  $\alpha$ for every event type.
However, the model still allows for variability between the event types through different scaling parameters $K_j$. 

In additition, and in contrast to the classical concept of correlation, multivariate regular variation models the dependence structure between the event types $ET_1,\dots, ET_7$ conditional on the occurrence of large losses. 

As a second assumption, the random network fraction matrix $A$ and the multivariate regularly varying vector $ET$ are stochastically independent. 

In Section~\ref{opRiskData} we discuss and validate these assumptions by testing them for a real-world database of operational losses, provided by DIPO from the Italian  Banking Association.  
We perform the following steps:
\begin{enumerate}
\item[(1)]
We fit a generalized Pareto distribution (gpd) with shape parameter $\xi$ and scaling parameter $\beta$ to each event type, which we use for estimating the tail index $\alpha_j$ and the scaling parameter $K_j$ for $ET_j$ for $j=1,\dots,d$.
As expected, the estimated tail indices differ, but all lie in the interval $[1.29,2.27]$; i.e., all have finite first moment, and one event type has even finite variance.
\item[(2)]
Aiming for a common value $\alpha$, we take the mean over all seven event type estimates $\wh\xi_j$ of the shape parameters  for $j=1,\dots,7$, which is supported by computing the confidence intervals. The common estimator $\wh\alpha$ will be the reciprocal of that value. 
Note that the estimators for the shape and scale parameters of a GPD are dependent, hence we account for having taken the mean over the shape parameters by re-estimating the scale parameters given that mean shape parameter.
In a second step, given this common tail index $\wh\al$, we re-estimate the scale parameters $K_j$.
\item[(3)]
By relying on statistical tests of independence, we also show that the assumption of independence between the network fraction matrix $A$ and the event type vector $ET$ is justified.
\item[(4)]
Our goal is to assess risk quantities for the business lines for large values; i.e., of the vector $BL$ resulting from large values of $ET$ and the network fraction matrix $A$.
Here we use the fact that multivariate regular variation of $ET$ with index $\alpha$ entails multivariate regular variation of $BL = A\times ET$ with the same index $\alpha$ for every network fraction matrix $A$ (Proposition A.1 in \citet{Basrak200295}).
Indeed, the linear transformation of \ET\ by $A$ transforms the marginals in an obvious way, and also the transformation of the spectral measure is well understood. 
Figure~\ref{simplePic}  provides a simple bivariate example on how the original dependence structure is changed under a linear matrix transformation. 
\item[(5)]
As we are interested in the dependence structure of the vector $BL$ of business lines, we first apply the transformation $A$ to $ET$ and estimate the spectral measure of $BL$.
\item[(6)]
Finally, we estimate the quantities of interest like Value-at-Risk, Conditional-Tail-Expectation, as well as risk contributions.
\end{enumerate}

We describe Figure~\ref{simplePic} in detail.
The left hand plot shows 5\,000 heavy-tailed simulated bivariate losses where the first and the second component are independent standard Pareto(2)-distributed. We observe that the large losses are present either in one or the other of the two components $(X_1,X_2)$, but never simultaneously. This is a special property of multivariate heavy-tailed distributions with (asymptotically) independent components. The spectral measure is then concentrated on the coordinate axes. {The plot} changes if a linear matrix transformation {(by a deterministic fraction matrix $A$)} is applied as shown in the right plot of bivariate losses $( \frac13 X_1+\frac23 X_2, \frac12 X_1 +\frac12 X_2 )$, where the first and second component are obviously dependent. The corresponding spectral measure is then concentrated on the lines from the origin crossing the points $(\frac13,\frac12)$ and $(\frac23,\frac12)$.
When starting with $(X_1,X_2)$ having a spectral measure with a density on the positive unit sphere of $\R^2$, then also the linearly transformed vector has, by the density transformation theorem, a density.

\vspace{-0.3cm}
\begin{figure}[ht!]
\begin{center}
\includegraphics[scale=0.7]{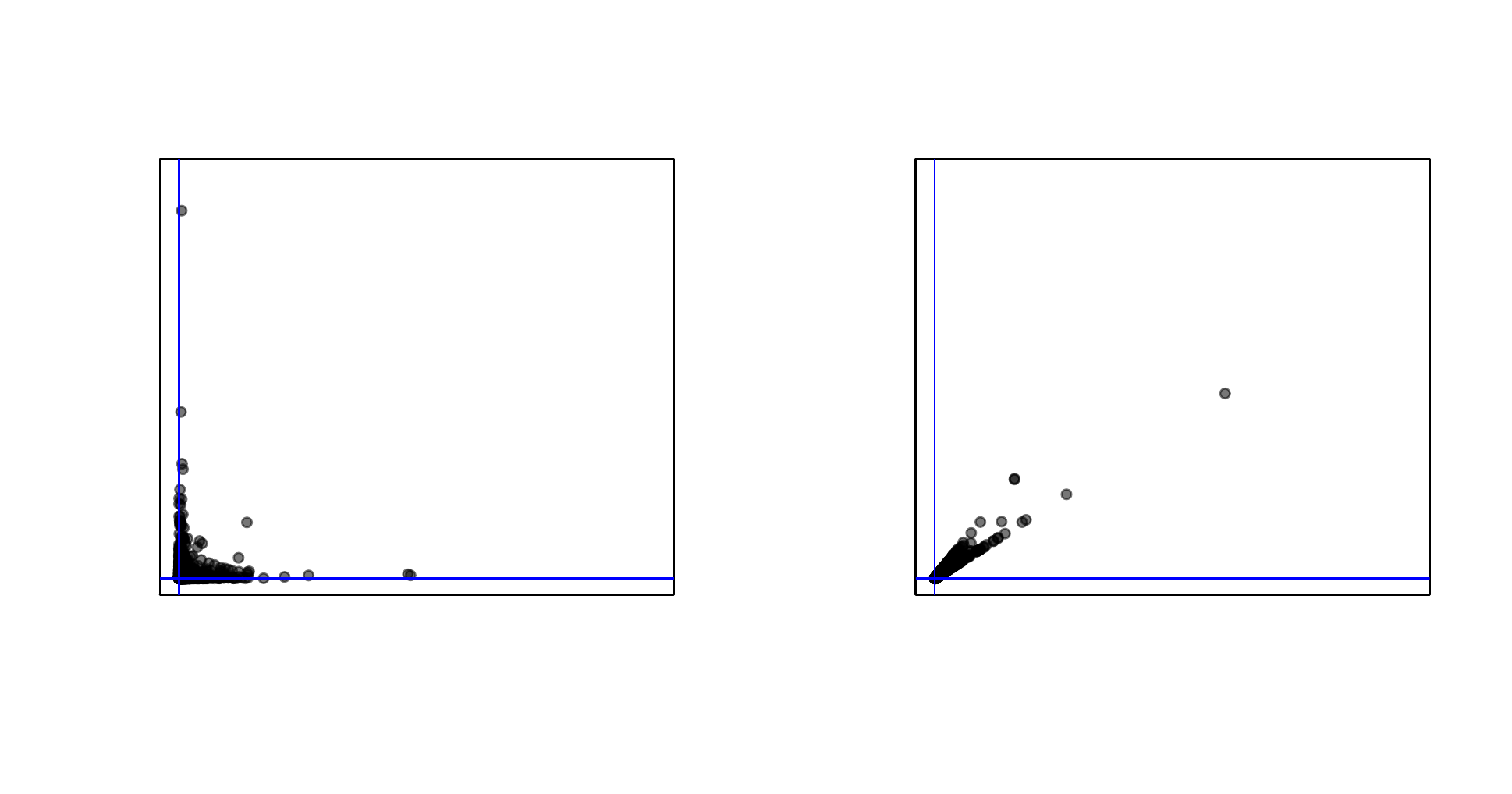}
\vspace{-1cm}
\caption{\label{simplePic} Dependence structure under linear matrix transformation: the left hand plot shows 5000 simulated bivariate event type data $(ET_1,ET_2)$ whose components are independent, which are linearly transformed in the right hand plot, representing business lines $(BL_1,BL_2)=( \frac13 ET_1+\frac23 ET_2, \frac12 ET_1 +\frac12 ET_2 )$.}
\end{center}
\end{figure}

\vspace{-0.5cm}
In the remainder of this section we {address} three topics most relevant to every operational risk manager. 
In Section~\ref{model} we quantify risk on both a single business line as well as on a company wide scale through tail risk measures. 
Section~\ref{allocation} deals with the problem of sensefully  allocating risk capital to  business lines. 
Both regulators and risk managers need to have methods for the allocation of risk to business lines according to their firm-wide importance with the goal to better quantify their risk appetite and risk tolerance, and also to set up adequate strategies for monitoring, insurance and hedging.
In the last Section~\ref{uncertainty}, we shed light on how to deal with dependence uncertainty by providing bounds for the tail risk measures; i.e., best and worst cases whichever type of asymptotic dependence the joint distribution of the event types  $ET_1,\dots,ET_7$ exhibit. 

\subsection{Tail risk measures} \label{model}

We start by recalling that, if the vector \ET\ of event types is multivariate regularly varying, then the vector \BL\ of business lines is again multivariate regularly varying, satisfying a version of \eqref{multrv}. 
Moreover, the tail index $\al$ is the same for \ET\ and \BL, only the spectral measures of both vectors differ, as illustrated in the above bivariate example. 
We aim at translating the information about the tail index $\alpha$, the extremal dependence structure between the event types given by the spectral measure $\Gamma$ as defined in \eqref{multrv}, and the network fraction matrix $A$ into risk estimates for the business lines.

Recall that the Value-at-Risk and the Conditional-Tail-Expectation of some one-dimensional random risk variable $X$ at confidence level $1-\gamma$ with $\gamma \in (0,1)$ are defined as
\begin{eqnarray*}
		\VaR_{1-\gamma}(X)&=&\inf \{t \ge 0 : \P[X>t]  \leq \gamma \}, \\
	      \CoTE_{1-\gamma}(X)&=&E[X|X>\VaR_{1-\gamma}(X) ] .
\end{eqnarray*}
Whereas the $\VaR$ is defined for all tail indices $\al>0$, the $\CoTE$ requires $\al>1$.
As Table~\ref{tab:parest} shows, all event types have estimated tail index larger than 1, such that
for the data at hand both risk measures can be estimated, which we do in Section~\ref{opRiskData}.

According to par.667 of the Basel II Framework for operational risk, the VaR has to be estimated at the extreme confidence level of 0.999; see \cite{Basel2011}.
By lack of data in this range the task obviously does not allow for non-parametric quantile estimation, only parametric or semi-parametric estimation makes sense. 
Whereas one-dimensional tail and quantile estimation based on extreme value theory or regular variation are by now classics, the challenge lies in quantifying extreme risk of the whole system, here represented by the large values of every business line and their extreme dependencies.

As a consequence, our foremost goal is to find asymptotic expressions for the tail risk measures in the bipartite model involving the different building blocks: the tail index $\alpha$, the asymptotic dependence structure {between the event type losses} given by a spectral measure, as well as the network fraction matrix which represents the process of distributing the event type losses to the different business lines. 
After providing estimates of these quantities, we develop estimators for the tail risk measures defined above. 
Moreover, relying on the linear transformation of the event type losses by the network fraction matrix yields a robust way to capture dependence between the business lines, which has to be considered as the result of two ingredients: first the intrinsic dependence structure between the event types which is altered  in a second step by the linear transformation; cf. Figure~\ref{simplePic}. 
The following result states the asymptotics for $\VaR$ and $\CoTE$ for growing confidence levels, and is the basis for  estimation procedures to follow.

\begin{theorem}\label{tail-risk-measures}
Let the event type vector $ET=(ET_1,\dots,ET_7)$ be multivariate regularly varying with tail index $\alpha>0$ and  spectral measure $\Gamma$ (see \eqref{multrv}) as well as stochastically  independent of the (random) network fraction matrix $A$. Denoting by 
$\mathbb{S}_{+}^{6}$ a positive unit sphere in $\R^{7}$,
we define risk constants for $i=1,\dots,8$,
\begin{gather}\label{constantsIS}
C^i=\E \int_{\mathbb{S}_{+}^{6}}(As)_i^{\alpha}d\Gamma(s)\quad \mbox{and}\quad C^S= \E \int_{\mathbb{S}_{+}^{6}}\|As\|^{\alpha}d\Gamma(s),
\end{gather}
for some norm $\|\cdot\|$.
Then the asymptotic behaviour of the VaR and the CoTE---in the latter case only for $\alpha >1$---of a single business line $BL_i$ is given by
\begin{gather}\label{varconstants}
\VaR_{1-\gamma}(BL_i)\sim (C^i)^{1/\alpha}\gamma^{-1/\alpha} \mbox{ and }\CoTE_{1-\gamma}(BL_i)\sim \frac{\alpha}{\alpha-1}(C^i)^{1/\alpha}\gamma^{-1/\alpha},\quad\ga\to 0.
\end{gather}
Moreover, the asymptotic behaviour of the VaR and CoTE---in the latter case only for $\alpha>1$---
on a company-wide scale modeled by the sum $\|BL\|=BL_1+\dots+BL_8$ over the business lines  
is given by
\begin{gather}\label{coteconstants}
\VaR_{1-\gamma}(\|BL\|)\sim (C^S)^{1/\alpha}\gamma^{-1/\alpha} \mbox{ and }\CoTE_{1-\gamma}(\|BL\|)\sim \frac{\alpha}{\alpha-1}(C^S)^{1/\alpha}\gamma^{-1/\alpha},\quad\ga\to 0.
\end{gather}
\end{theorem}

Certainly, understanding the risk of the single business lines as well as the company-wide risk is most important. 
In Theorem~\ref{prop:varasym} we provide and prove a more general version of this result, which encompasses besides subsets of business lines also other aggregation functions than the sum for the company-wide risk.

\subsection{Capital Allocation}\label{allocation}

The allocation of total capital to different entities in an optimal way is an important topic in finance; see e.g. \citet{Embrechtsetal2017}.
Applications include risk adjusted performance measurement for  portfolios as in \citet{kalkbrener,Tasche2000}, capital allocation in credit risk \citet{Tasche2004b}, or allocating systemic risk \citet{GourMon}.
In all these examples, allocated risk contributions  should sum up to the full  risk of a portfolio or a system.
We adapt this concept to the allocation of company-wide operational risk to business lines, which will then allow us to possibly identify which critical areas require setting up risk monitoring and mitigating strategies to better prevent and hedge operational losses. 
The following property is essential.

For a risk vector\footnote{In the operational risk context $Y$ corresponds to the vector of business lines and $q=8$.}   
$Y=(Y_1,\dots,Y_q)$ and a risk measure $R$ (like $\VaR$ or $\CoTE$)
we define the firm-wide risk by $R(\|Y\|)$.
 If there exist risk contributions $R(Y_i\mid \|Y\|)$ such that
 \begin{gather}\label{fullAllocation}
R(\|Y\|) = \sum_{i=1}^q R(Y_i\mid \|Y\|),
\end{gather}
 then the firm-wide risk satisfies the {\em full allocation property}.

Since tail risk measures are in general is not $additive$,  $\VaR(\|Y\|) = \sum_{i=1}^q  \VaR(Y_i)$, does not hold. 
Applying Euler's principle to the asymptotic formulae of total economic capital as derived in Theorem~\ref{tail-risk-measures}, we obtain an asymptotically full allocation property for the risk measures $\VaR$ and $\CoTE$.
A proof is given in Appendix~\ref{Appendixmodel}.

\begin{theorem}\label{pr:riskcontvar}
{Let the assumptions of Theorem~\ref{tail-risk-measures} 
hold with $\alpha > 1 $ and take 
 $R=\VaR_{1-\ga}$ or $R=\CoTE_{1-\gamma}$  
as well as the sum norm $\| \cdot \|$  in  \eqref{fullAllocation}. 
We define for $i=1,\dots,8$ the capital allocation constants
\beam\label{capall}
CA^i := (C^{S})^{1/\al-1} \Big( \E  \int_{\mathbb{S}_+^{d-1}}   \|As\|^{\alpha -1}(As)_{i}  d\Gamma(s)\Big)
\eeam
Then we obtain the asymptotically full allocation property} for the vector of business lines \linebreak 
$BL=(BL_1,\dots, BL_8)^\top$:
\beam\label{alloc}
\VaR_{1-\ga}(\|BL\|) &\sim &\sum_{i=1}^8 \VaR_{1-\ga}(BL_i\mid \|BL\|),\quad\ga\to 0,
\eeam
with risk contributions 
\beam\label{riskcontvari}
\VaR_{1-\ga}(BL_i \mid \|BL\|) &:=&    \gamma^{-1/\al} \, CA^i
\eeam as well as 
\beao
\CoTE_{1-\ga}(\|BL\|) &\sim &\sum_{i=1}^q \CoTE_{1-\ga}(BL_i\mid \|BL\|),\quad\ga\to 0,
\eeao
with risk contributions
\beam\label{riskcontcotei}
\CoTE_{1-\ga}(BL_i \mid \|BL\|) &:=&   \frac{\alpha}{\alpha-1} \gamma^{-1/\al} \, CA^i
\eeam
\end{theorem}

\subsection{Handling dependence uncertainty}\label{uncertainty}

As the estimation of the spectral measure in \eqref{constantsIS}, which models the dependence structure of a multivariate regularly varying vector, is not a simple task, we provide bounds given by the worst and best case dependence scenarios for the risk constants $C^{i}$ at individual level  and $C^S$ at company-wide level (cf. \eqref{constantsIS}, both in terms of the asymptotic dependence structure of the event type losses given by the spectral measure $\Gamma$ of \ET, see \citet{KK} and also \citet{MainikOrder}. 
Writing 
$C^{i}_{\Gamma}$ and $C^S_{\Gamma}$ for the constants in case of an arbitrary spectral measure $\Gamma$ we give bounds below. 
Denote by $C^{i}_{ind}, C^S_{ind}$ the risk constants corresponding to asymptotic independence of the event types, where the spectral measure $\Gamma_{ind}$ of $ET$ is concentrated on the axes, and by $C^i_{dep}, C^S_{dep}$ the risk constants corresponding to full asymptotic dependence,  where the spectral measure $\Gamma_{dep}$ of the event types is concentrated on the line through {$K^{1/\alpha} 1$} with $K^{1/\alpha} =diag(K_1^{1/\alpha} ,\dots,K_d^{1/\alpha})$ and $1=(1,\dots,1)^{\top}$.
As stated in  Theorem 3.2 in \citet{OKCKLUGR}, denoting the canonical unit vectors by $e_j$ for  $j =1,\dots,d$,  we find the risk constants for the corresponding business lines 
\beam \label{VaRconst}
 C^i_{ind} =  C^i_{ind} (A) = \sum_{j=1}^{d} K_j\E A_{ij}^{\alpha},  \,\, i=1, \ldots, q,  \, \mbox{ and } \, C_{ind}^S = C_{ind}^S (A) = \sum_{j=1}^d K_j \E \| A e_j \|^\alpha,\\
\label{VaRconst_dep}
C_{dep}^i = C_{dep}^i (A) :=  \E (AK^{1/\al}1)_{i}^{\alpha }, \,\, i=1, \ldots, q,   \,  \mbox{ and }\quad  C_{dep}^S = C_{dep}^S (A) = \E \|AK^{1/\al}1\|^\alpha.
\eeam
The following statement summarizes Theorems 3.1, 3.2 and 3.3 in \citet{KK}.
These inequalities for the risk constants entail asymptotic inequalities for the tail risk measures in \eqref{varconstants} and \eqref{coteconstants}.

\bpr
Consider the three event type vectors $ET_{ind}, ET_\Gamma, and ET_{dep}$, each with the same components $ET_1,\dots,ET_d$, but different dependence structures as listed above. Then for the risk constants satisfy the following bounds:
\begin{gather}
C_{ind}^{i}\leq C_{\Gamma}^i \leq C_{dep}^i\quad \text{and} \quad  C_{ind}^{S}\leq C_{\Gamma}^S \leq C_{dep}^S\quad \text{for}\quad        \alpha \geq 1,   \label{individualbounds} \\ 
\label{systembounds}
C_{ind}^{i}\geq C_{\Gamma}^i \geq C_{dep}^i \quad \text{and} \quad  C_{ind}^{S}\geq C_{\Gamma}^S \geq C_{dep}^S\quad \text{for}\quad        0< \alpha \leq 1. 
\end{gather}
\epr

Note that an upper bound in the case of $\alpha >1$ (where the expectation of the loss variable is finite) is a lower bound in the case of  $\alpha <1$ (where the expectation of the loss variables is infinite) and vice versa. 
As a consequence, the assumption of full asymptotic dependence between the components of the risk vector gives an upper bound only for $\alpha>1$, which implies in particular that the sum of the $\VaR$s of each risk component is not only not an (asymptotic) upper bound in general, but can even be the lower bound. 
Note also that multivariate regular variation naturally models association; i.e., positive (or no) dependence between the event types and positive (or no) dependence between the business lines.

Having such bounds available has the useful practical implication to gain understanding of the uncertainty we are facing in the most extreme situations and the impact that dependence may have on the tail-risk; cf. Figure~\ref{depindfigure} for a data example.
Moreover, the possibility to capture the dependence between the business lines by the matrix $A$ provides a simple tool to test the potential consequences of alternative dependence configurations between the business lines so to say as a way to enable scenario analyses as well as developing stress testing systems that can capture joint extreme losses.

\section{Operational Risk Data}\label{opRiskData}

The operational risk data have been provided by the DIPO consortium (Database Italiano Perdite Oper\-ative---Italian Banking Association Consortium).
For confidentiality reasons, we can provide only some descriptive statistics, and cannot reveal numerical values of the risk constants and the capital allocation constants. 
This is, however, not needed for an illustration of our new method.

The database contains 145\,473 losses reported daily from 33 Italian banking groups, with about 180 entities, in a period of 10 years during 01/01/2003-30/12/2013. The reporting threshold is 5\,000 Euro, below which no loss is reported. To avoid the presence of a large number of zeros, we {aggregate losses weekly} for the total of 575 weeks and the resulting zeros for event type and business line data.
As displayed in column $\#0s$ in Table \ref{tab:ETBLdesc},  $ET_2$, $ET_4$ and $ET_7$ as well as $BL_3$, $BL_4$ and $BL_8$ have complete time series of weekly aggregated losses. 
Even when {aggregating losses  weekly}, there is only a small number of  zeros for most event types, with a maximum of 15.13\% for  Business Distruption \& System Failures ($ET_6$). 
In contrast, this is not the case when considering business lines, where 4 out of 8 business lines have more than 27\% of zeros in the entire period, with a maximum of 81\% for Corporate Finance ($BL_1$). 

When examining the marginal distributions of the magnitude, also called {\em severity}, of weekly losses, we notice that the distributions of losses to both $ET$s and $BL$s are not evenly spread, with $ET_5$ and $ET_6$ and $BL_1$, $BL_5$, $BL_6$, and $BL_7$ having less than 5\% of the total severity of losses. 
Figure \ref{fig:boxplot} displays the boxplots of weekly losses (on log-scale) grouped by $ET$s (left) and  $BL$s (right). Median losses are much lower for $ET_5$ and $ET_6$ as well as for $BL_1$, $BL_5$, $BL_6$ and $BL_7$, which are also characterized by a larger number of missing data as well as a smaller fraction of total severity than the remaining ones. 

In order to build a better model for assessing the risk tolerance and risk appetite of a company, which we will do at the level of single business lines, we begin with a model  for event types.

\begin{table}[]
\centering
\begin{tabular}{lccl|l|r|r|}
\cline{1-3} \cline{5-7}
\multicolumn{1}{|l|}{\textbf{Category}} & \multicolumn{1}{c|}{\textbf{ \#0s}} & \multicolumn{1}{c|}{\textbf{Severity}} & \hspace*{1cm} & \textbf{Category} & \textbf{ \#0s} & \textbf{Severity} \\ 
\cline{1-3} \cline{5-7} 
\multicolumn{1}{|l|}{{$ET_1$}}   & \multicolumn{1}{r|}{1,9\%}                 & \multicolumn{1}{r|}{13,1\%}            &  & {$BL_1$}   & 80,5\%                & 3,3\%             \\ 
\cline{1-3} \cline{5-7} 
\multicolumn{1}{|l|}{{$ET_2$}}   & \multicolumn{1}{r|}{-}                 & \multicolumn{1}{r|}{12,8\%}            &  & {$BL_2$}   & 1,7\%                 & 9,5\%             \\ 
\cline{1-3} \cline{5-7} 
\multicolumn{1}{|l|}{{$ET_3$}}   & \multicolumn{1}{r|}{0,3\%}                 & \multicolumn{1}{r|}{5,4\%}             &  & {$BL_3$}   & -                 & 43,1\%            \\ 
\cline{1-3} \cline{5-7} 
\multicolumn{1}{|l|}{{$ET_4$}}   & \multicolumn{1}{r|}{-}                 & \multicolumn{1}{r|}{43,8\%}            &  & {$BL_4$}   & -                 & 18,4\%            \\ 
\cline{1-3} \cline{5-7} 
\multicolumn{1}{|l|}{{$ET_5$}}   & \multicolumn{1}{r|}{8,3\%}                 & \multicolumn{1}{r|}{0,6\%}             &  & {$BL_5$}   & 27,1\%                & 0,5\%             \\ \cline{1-3} \cline{5-7} 
\multicolumn{1}{|l|}{{$ET_6$}}   & \multicolumn{1}{r|}{15,1\%}                & \multicolumn{1}{r|}{1,2\%}             &  & {$BL_6$}   & 52,5\%                & 0,4\%             \\ 
\cline{1-3} \cline{5-7} 
\multicolumn{1}{|l|}{{$ET_7$}}   & \multicolumn{1}{r|}{-}                 & \multicolumn{1}{r|}{23,3\%}            &  & {$BL_7$}   & 45,0\%                & 0,5\%             \\ 
\cline{1-3} \cline{5-7} 
                                     & \multicolumn{1}{l}{}                       & \multicolumn{1}{l}{}                   &  & {$BL_8$}   & -                 & 24,3\%            \\ \cline{5-7} 
\end{tabular}
\caption{Percentage of frequencies of numbers of zero weekly aggregated losses (column: $\#0s$) and percentage of severity of losses (column: Severity), when grouped by Event Types (left table) and Business Lines (right table).\\ 
\footnotesize{Event types: 1. Internal Fraud, 2. External Fraud, 3. Employment Practices \& Workplace Safety, 4. Clients, Products \& Business Practices, 5. Damage to Physical Assets, 6. Business Disruption \& System Failures, 7. Execution, Delivery \& Process Management. Business lines: 1. Corporate Finance, 2. Trading and Sales, 3. Retail Banking, 4. Commercial Banking, 5. Payment and Settlement, 6. Agency and Custody, 7. Asset Management, 8. Retail Brokerage. }}
\label{tab:ETBLdesc}
\end{table}

\begin{figure}[ht!]
\centering
\begin{tabular}{cc}
\includegraphics[scale=0.4]{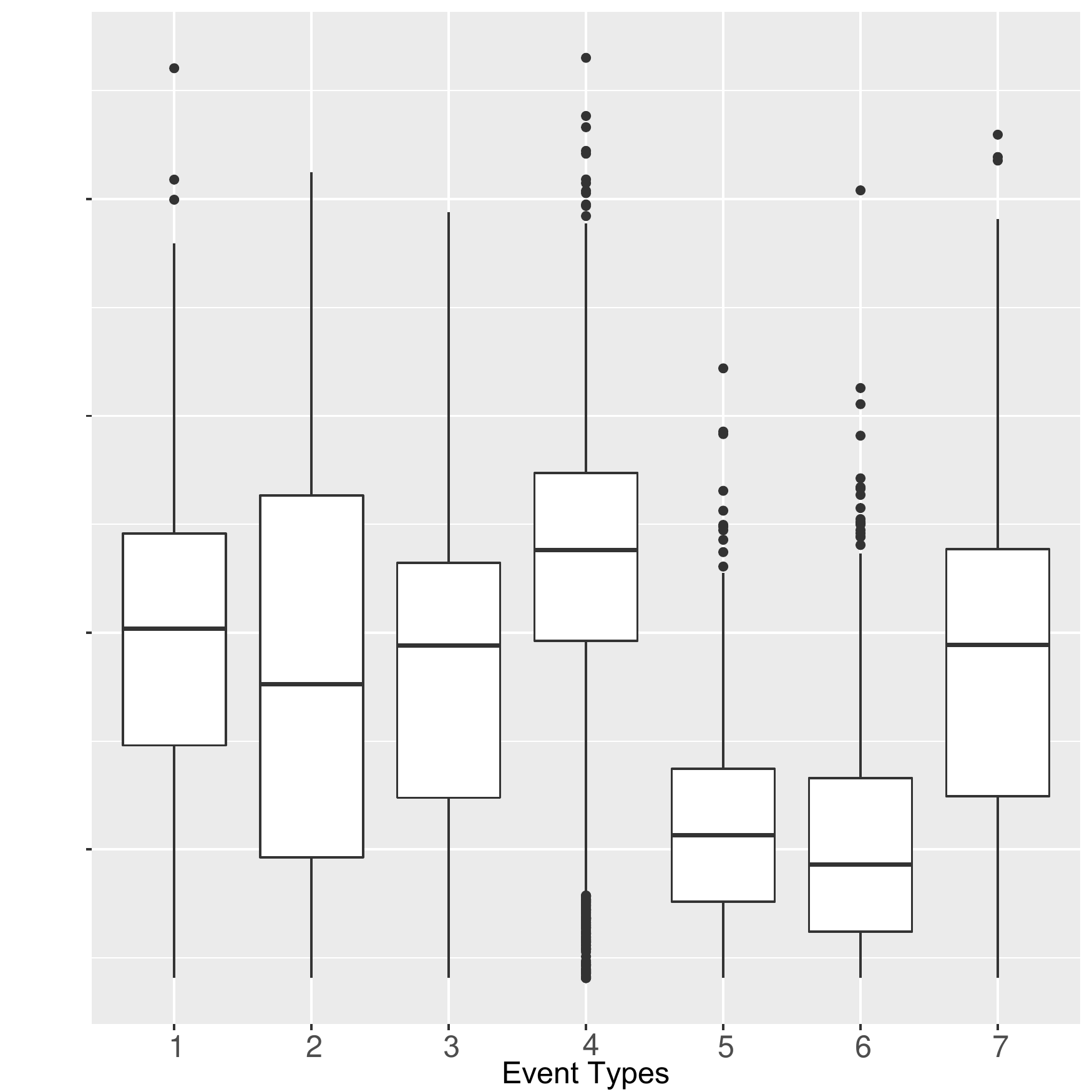}
&
\includegraphics[scale=0.4]{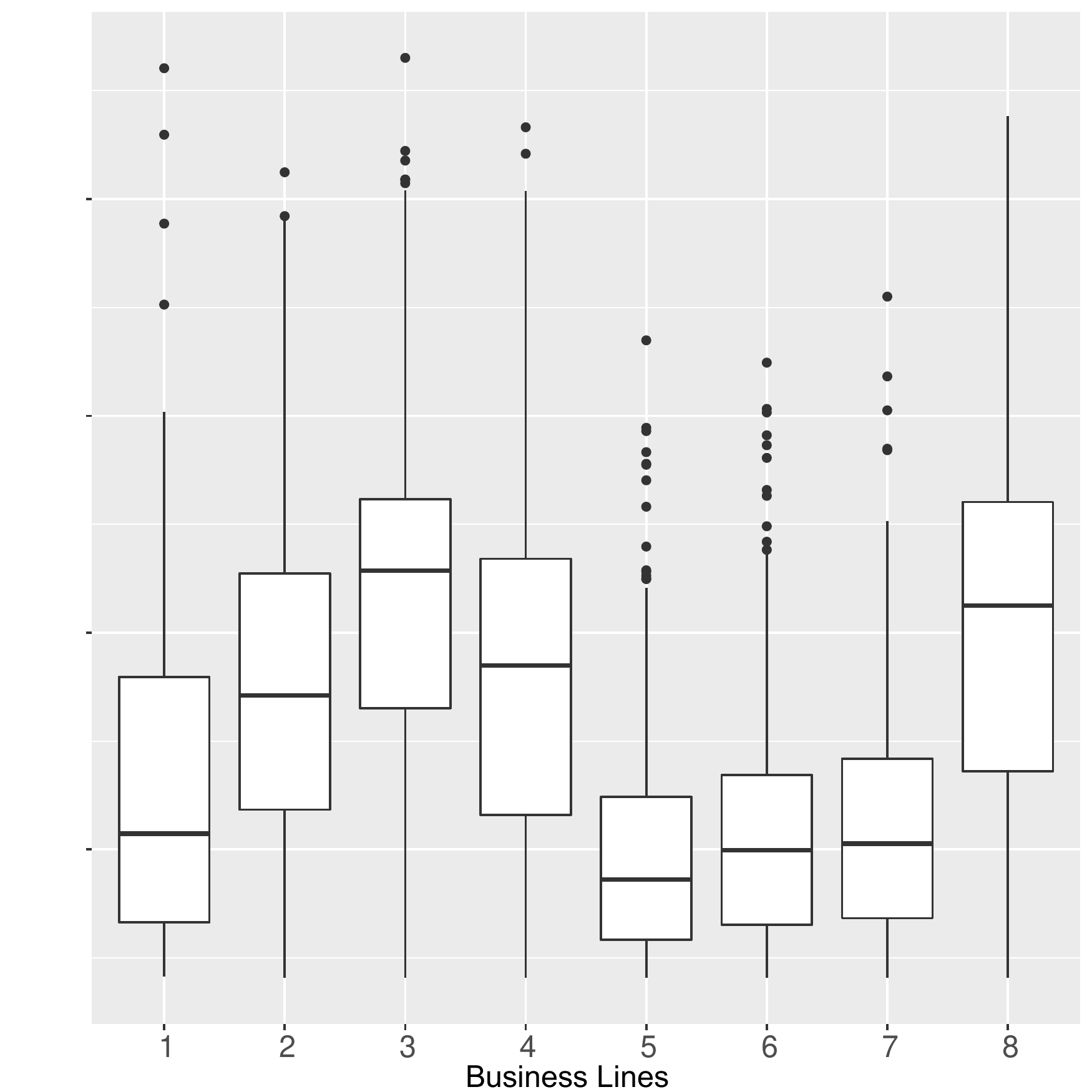}
 \end{tabular}
\caption{Boxplots of weekly aggregated severity of  operational risk losses (in log-scale)  grouped by $ET$s (left plot) and $BL$s (right plot).}
\label{fig:boxplot}
\end{figure}

\subsection{Marginal modelling and estimation}\label{s31}

Extreme value theory and regular variation is one of the basic tools in operational risk; see e.g. the books \citet{Shev2, EKM1997, McNeilFreyEmbrechts2005, Shev} or the articles \citet{BockerKluppelberg2010, ChavezDemoulinEmbrechtsNeslehova2006, ChavezDemoulinetal2015}. 
After aggregating the losses weekly, as a model for marginal regular variation, we take for every $j\in\{1,\dots,7\}$ the form \eqref{Pareto} in its statistical version given by
\begin{gather}
\pr{ET_j>t}\approx K_j t^{-\alpha}, \quad t>u_j,
\end{gather}
meaning we assume that from a fixed large $u_j$ on that the right-hand side provides a good approximation to the left-hand tail probability. 

Marginal estimation of regularly varying distribution tails, in particular the tail index $\al$, has been studied since Hill's seminal paper; see e.g. \citet[Section~4]{Resnick2007}, or  \citet[Section~6]{EKM1997}, where a number of different methods are discussed.

We rely on the common agreement of being the best method and estimate $\alpha$ and $K_j$ by the POT--method (\textit{peaks--over--threshold}); see \citet[Section~6.5.1]{EKM1997}.
As detailed in Section~\ref{subsec:estimation_procedure}, we proceed as follows.
For each marginal we estimate the shape parameter $\xi$, the scale parameter $\beta$ and the location parameter $u$ of the gpd in \eqref{gpd} by the POT method and then rely on the approximations in \eqref{gpd_app1} and \eqref{gpd_app2}. 

Recall from \eqref{multrv} that multivariate regular variation requires the same tail index $\alpha=1 / \xi$ for every marginal distribution.
For real data, the parameter $\alpha$  will never be estimated as being equal for different marginals, even in the case of  truly multivariate regularly varying data, where the tail indices for different component of the risk vector are theoretically equal.
Estimation in practice will not lead to the same tail index estimate over the components, simply because of limited sample size and estimation error.  

\begin{table}[]
\centering
\begin{tabular}{|c|c|r|r|}
\hline
\textbf{Class} & \textbf{$\wh{\xi_j}$} & \textbf{$\wh{\beta_j}$} & \textbf{$\wh u_j$} \\ \hline
$ET_1$   & 0.61       & 2\,508\,366       & 3\,320\,328    \\ \hline
$ET_2$   & 0.56       & 1\,736\,738       & 4\,150\,559    \\ \hline
$ET_3$   & 0.44       & 571\,195        & 633\,702     \\ \hline
$ET_4$  & 0.59        & 9\,091\,252       & 10\,463\,268   \\ \hline
$ET_5$   & 0.64       & 190\,547        & 187\,932     \\ \hline
$ET_6$   & 0.77       & 105\,049        & 49\,243      \\ \hline
$ET_7$   & 0.70       & 2\,401\,487       & 3\,417\,419    \\ \hline
\end{tabular}
\caption{Estimates for the gpd parameters for weekly aggregated operational losses from the DIPO database.  Column 2:  shape parameter $\xi$, column 3: scale parameter $\beta$, column 4: location (threshold) parameter $u$.}
\label{tab:parest}
\end{table}

As a remedy, for each event type, we take the {mean over the estimates $\wh\xi_j$} in Table~\ref{tab:parest}, such that  $\wh\xi_{mean}= \frac17 \sum_{j=1}^7 \wh\xi_j=0.61471$, resulting in a common tail index estimate $\wh{\alpha}=1/\wh\xi_{mean} =1.62678$.  The standard deviation of the estimate $\wh\xi_{mean}$ is small and the single estimates $\wh\xi_j$ for $ET_j$ for $j=1,\dots,7$ are quite close to the mean, with three values slightly above it and four below.

When computing the asymptotic confidence intervals for each parameter estimate $\wh\xi_j$, we have that all $99\%$ confidence intervals around $\wh\xi_j$ include the value $\wh\xi_{mean} = 0.6147$, while at $95\%$ all but the confidence interval for $\wh\xi_3$ does not include this value; however, the upper bound of its confidence interval is still very close, being equal to 0.6117.

Under multivariate regular variation, taking the mean ofer all $\xi_j$ allows us to compensate for potential under- or over-estimation. 
The range of estimates $\wh\xi_j$ for $j=1,\dots,7$ give reasonable possibilities for stress-testing or for evaluation of the potential consequences of lighter or fatter tails;  e.g. by taking the minimum or the maximum over the different shape parameter estimates.  
 After setting $\xi_{mean}=0.6147$ as shape parameter and the original thresholds $u_{j}$ as location parameters in \eqref{gpd}, we re-estimate by (conditional) maximum likelihood estimation the different scale parameters $\beta_j$ as reported in Table \ref{tab:parfit}. 
 Since $\wh\xi_{mean}$ and $\wh\beta_j$ are asymptotically  dependent (see \citet[Section~6.5.1]{EKM1997}) the re-estimation of $\beta_j$ will adjust for possible misspecification of $\wh\alpha$.
To assess the fit of all marginal models, we perform Kolmogorov-Smirnov goodness of fit tests that, as displayed in column 5 of Table~\ref{tab:parfit}, results in not rejecting the null hypothesis of the goodness of fit of the estimated distribution for each $ET_j$ for $j=1,\dots,7$.   
 Moreover, we inspect the QQ-plots in Figure \ref{fig:qqplotET}:
 they provide further evidence of the fit of the gpd to all marginals for  $\wh\xi_{mean}$ and the scale parameters $\wh\beta_j$ reported in Table~\ref{tab:parfit}: only the event type data $ET_2$ and $ET_3$ show some deviations in the right tail. 
 
For data sets with decidedly different parameters $\al_j$, it is adequate to transform all marginal distributions to have the same parameter $\alpha$, which is occasionally done in extreme value statistics, and is indeed a prerequisite to all copula-based estimation. However, this introduces an extra estimation error into the procedure. 
Our method, which is justified for this data set, has the advantage that all quantities of interest are always measured in EUR, as e.g. Value-at-Risk, Conditional-Tail-Expectation, as well as risk contributions, which one looses when transforming the marginals in a non-linear way. Moreover, throughout this analysis all tail risk measures are estimated from the original data, assessing risk exactly from them.
By goodness of fit tests, confidence intervals and visual inspection of QQ-plots, we find that the gpd marginals with common $\wh\alpha$ provide a reasonable fit to the real-world data.

 \begin{figure}[htbp!]
\centering
\begin{tabular}{ccc}
\subfigure{\includegraphics[width=5cm]{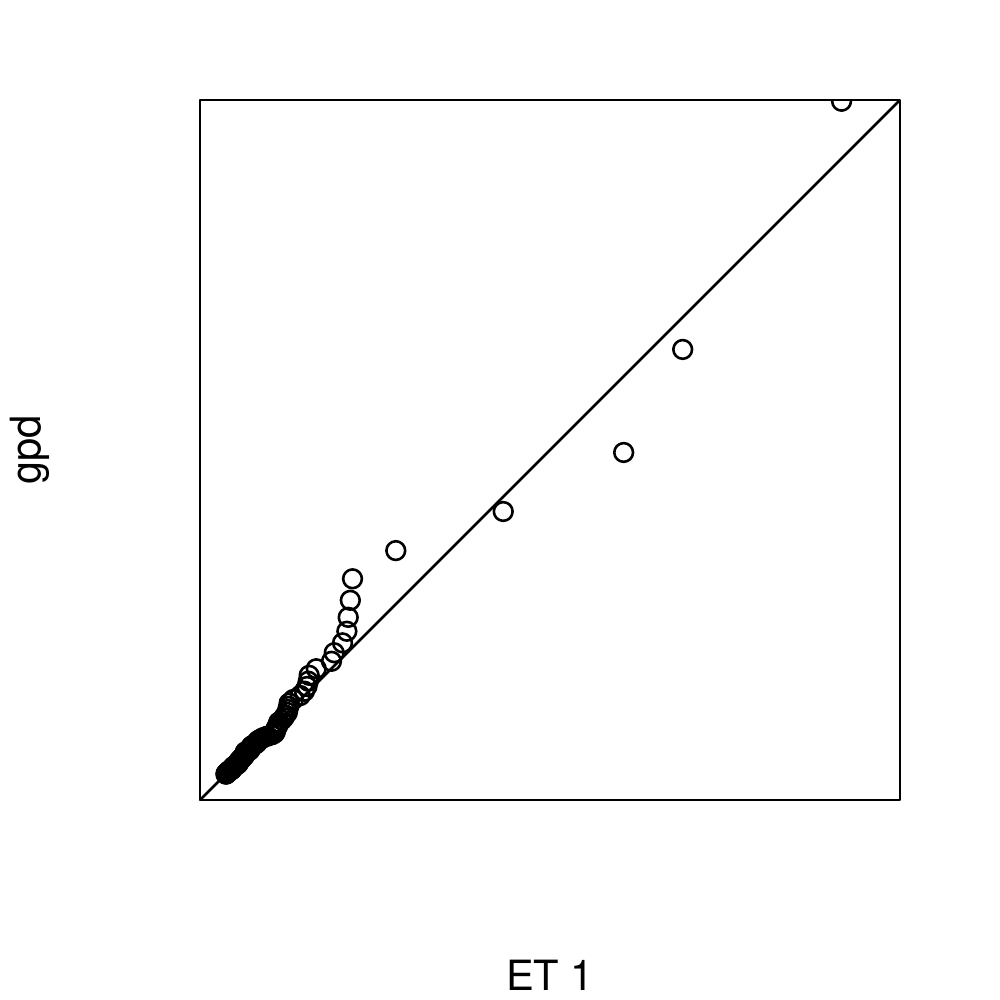}} 
   & \subfigure{\includegraphics[width=5cm]{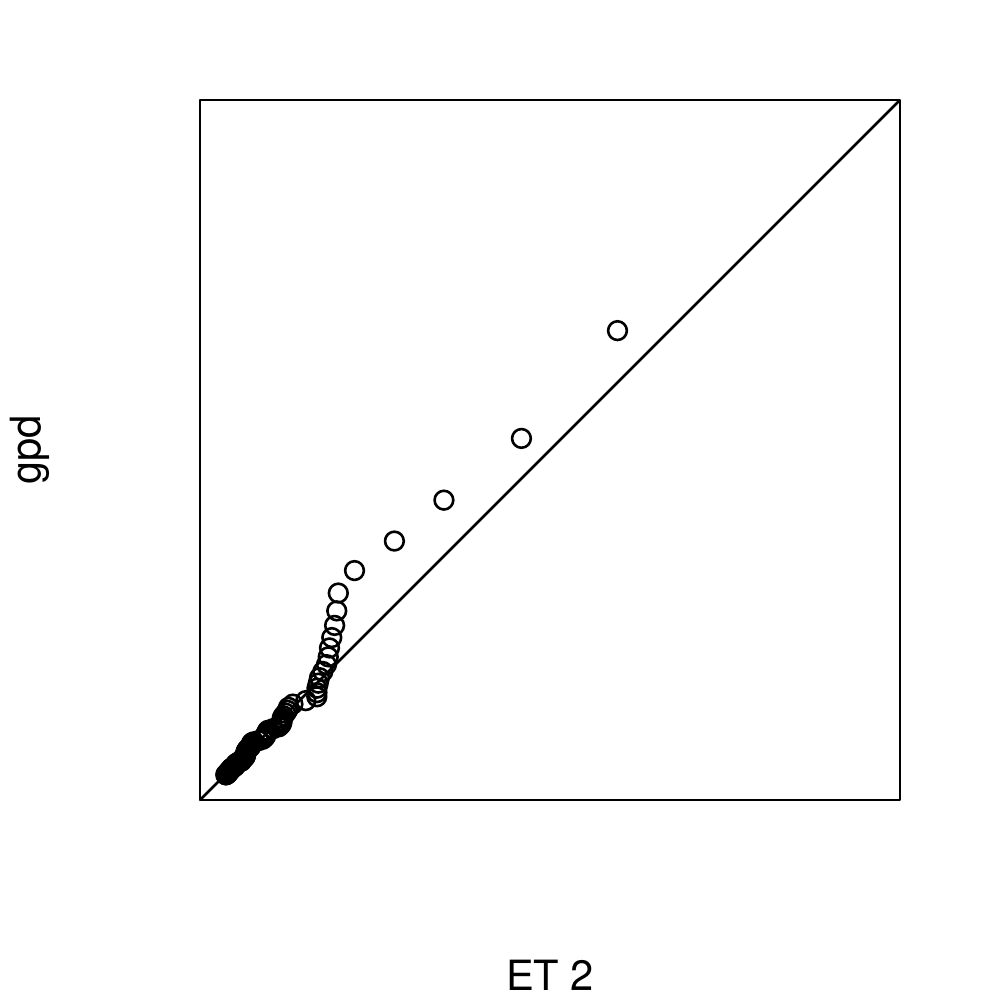}} 
& \subfigure{\includegraphics[width=5cm]{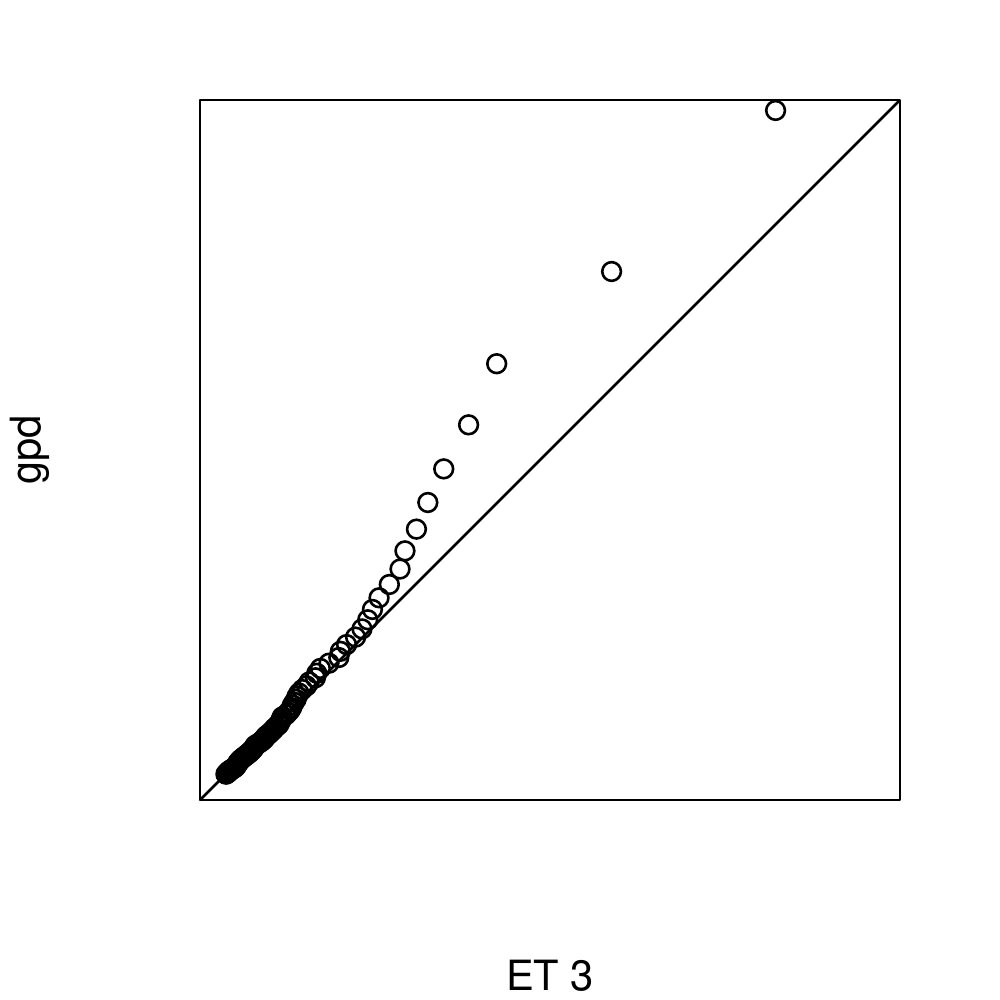}}  \\ 

\subfigure{\includegraphics[width=5cm]{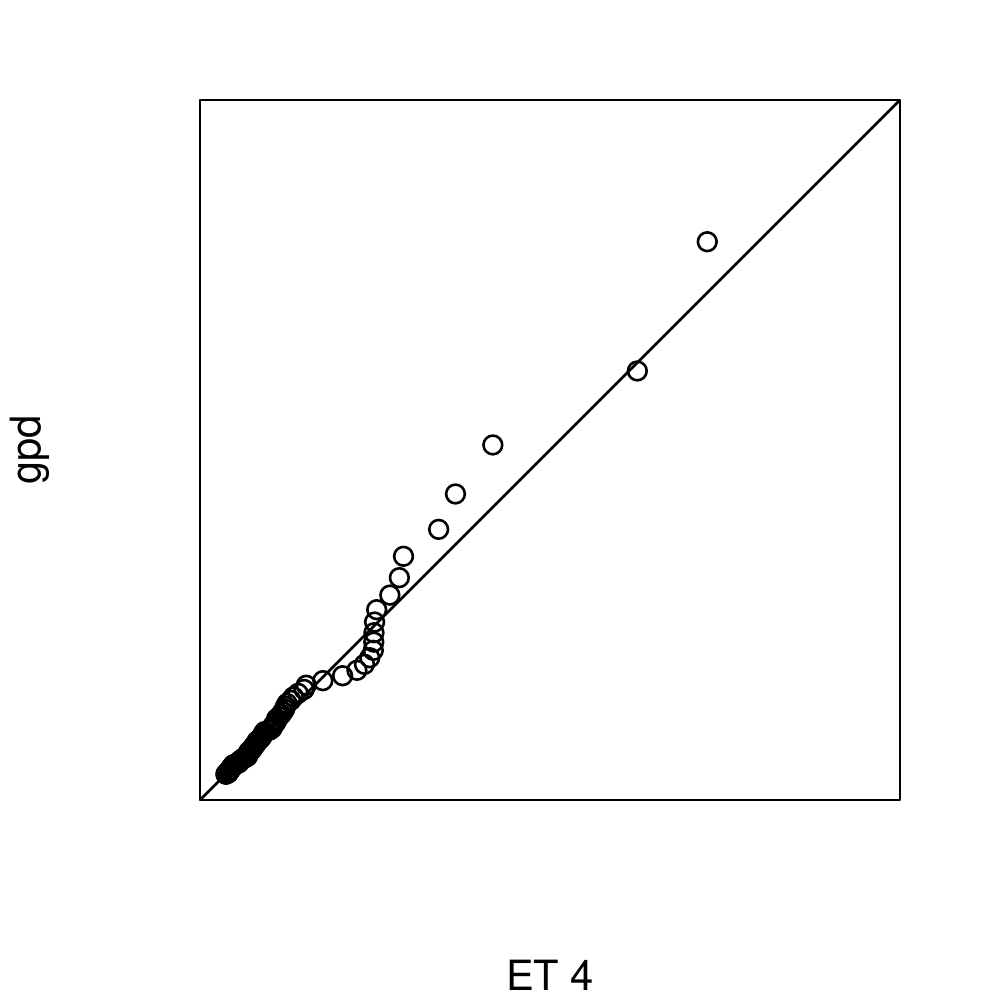}} 
   & \subfigure{\includegraphics[width=5cm]{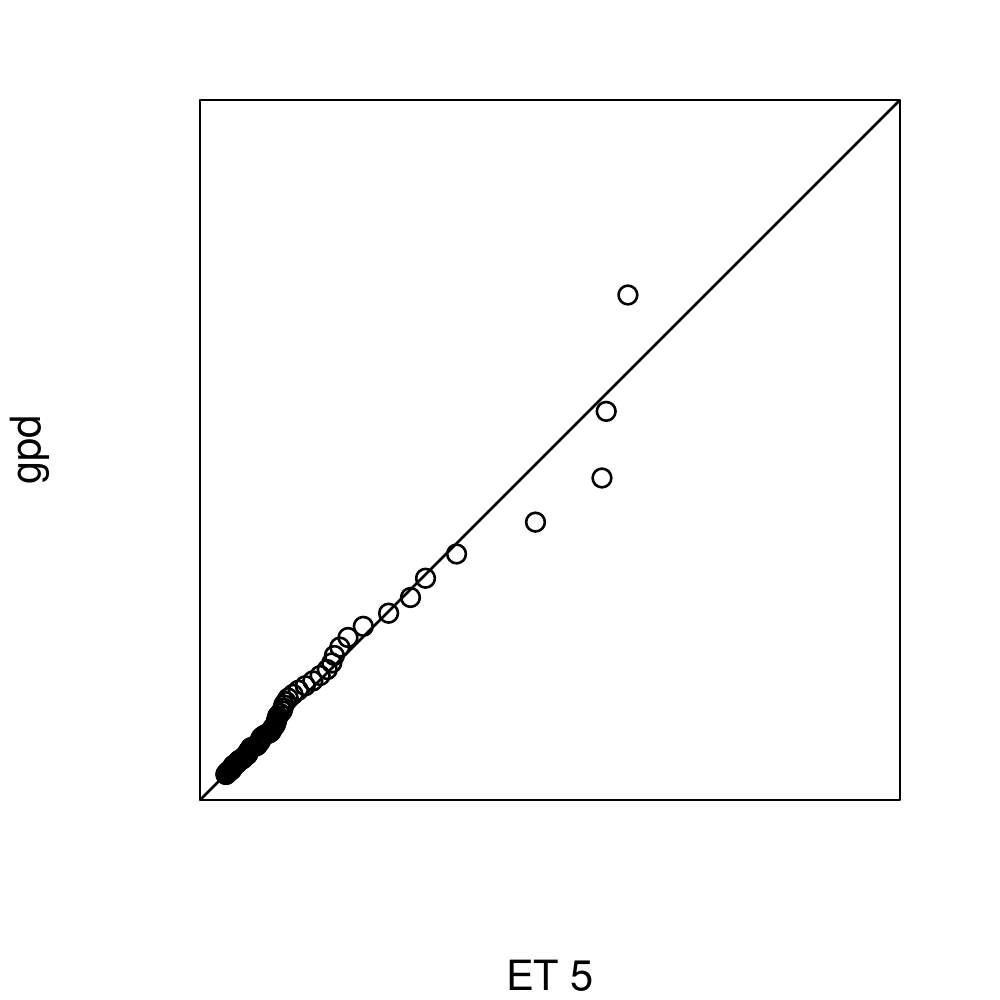}} 
& \subfigure{\includegraphics[width=5cm]{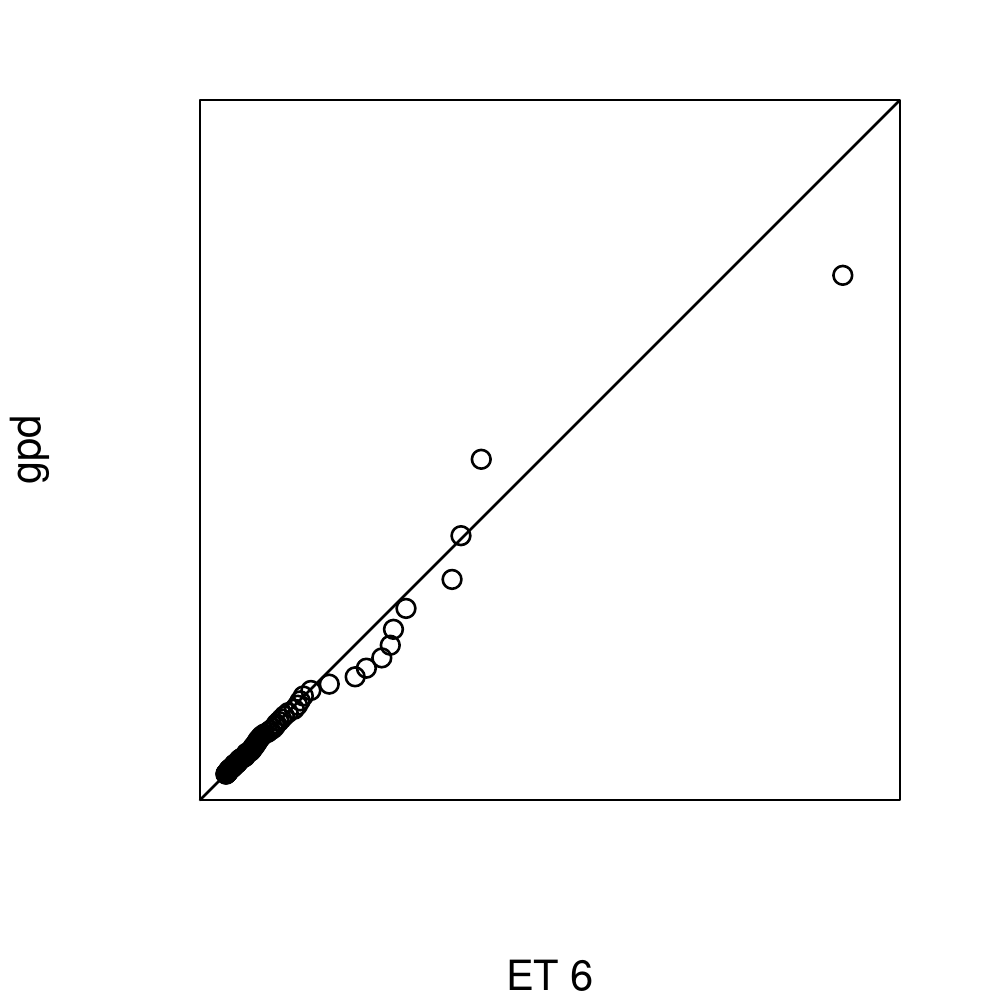}} \\
\ &\subfigure{\includegraphics[width=5cm]{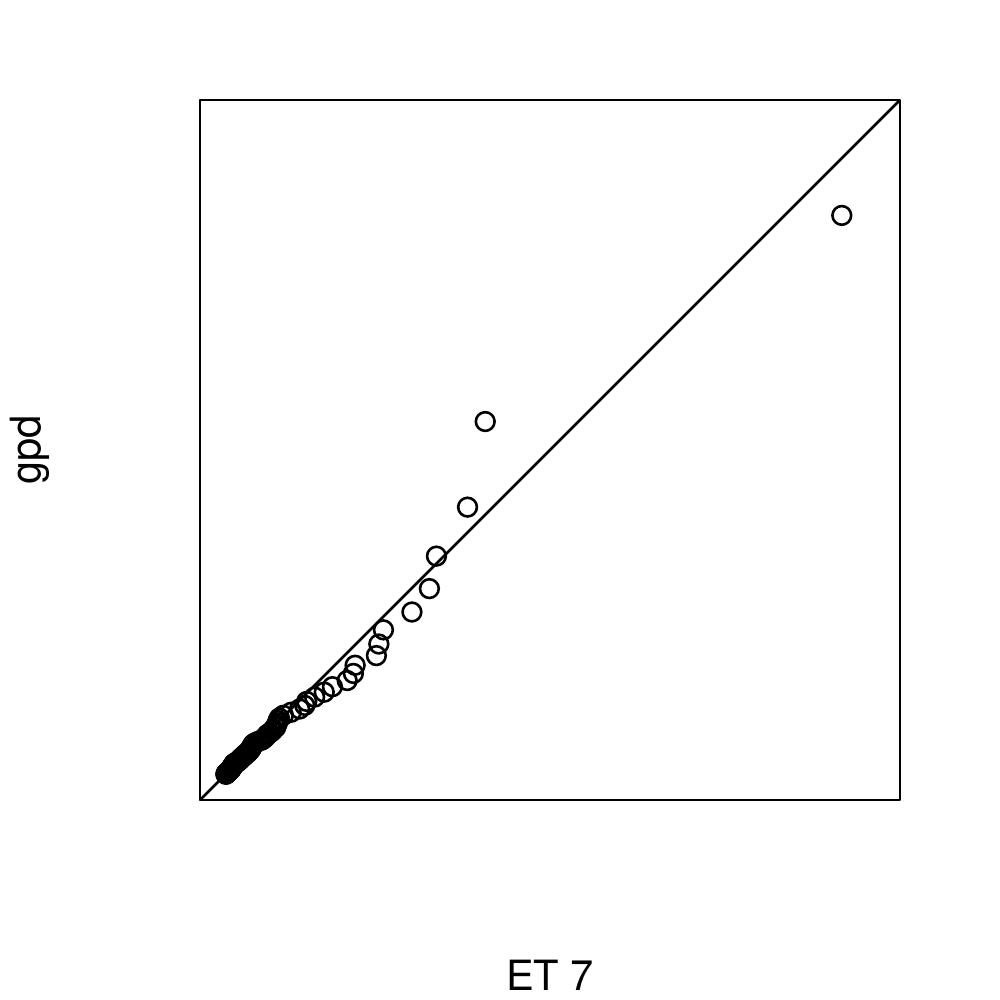}} & \ 
	\end{tabular}
\caption{ QQplot of DIPO data vs. gpd with $\wh\xi_{mean}= 0.61471$
and estimated scale and location parameters from Table~\ref{tab:parfit} }.
\label{fig:qqplotET}
\end{figure}

\begin{table}[]
\centering
\begin{tabular}{|l|c|c|r|r|c|}
\hline
\textbf{Class} & \textbf{$\wh\xi_{mean}$} & \textbf{$\wh\alpha$} &\textbf{$\wh{\beta}_j$} & \textbf{$\wh{u}_j$} & \textbf{p-value  KS} \\ \hline
$ET_1$   & 0.6147    & 1.6268      & 2\,509\,123       & 3\,320\,328    & 0.9601               \\ \hline
$ET_2$  & 0.6147     & 1.6268    & 1\,680\,229       & 4\,150\,559    & 0.3115               \\ \hline
$ET_3$   & 0.6147    & 1.6268      & 513\,882        & 633\,702     & 0.6685               \\ \hline
$ET_4$   & 0.6147    & 1.6268      & 8\,923\,344       & 10\,463\,268   & 0.9361               \\ \hline
$ET_5$  & 0.6147     & 1.6268     & 193\,148        & 187\,932     & 0.9998               \\ \hline
$ET_6$  & 0.6147     & 1.6268     & 115\,836        & 49\,243      & 0.9251               \\ \hline
$ET_7$   & 0.6147    & 1.6268      & 2\,535\,017       & 3\,417\,419    & 0.3557               \\ \hline
\end{tabular}
\caption{Re-estimated scale parameters $\wh{\beta}$, when using the same shape parameter $\wh\xi_{mean}$  for all event types.  Column 2:  shape parameter $\xi_{mean}$, column 3: scale parameter $\wh{\beta}$, column 4: location (threshold) parameter $\wh{u}$, column 5: $p$-value of Kolmogorov-Smirnov test.}
\label{tab:parfit}
\end{table}

\subsection{Independence of  loss variables and networks}\label{s32}

One assumption of the bipartite operational risk model and needed in Theorem~\ref{prop:varasym} is the independence of the random network fraction matrix $A$ and the regularly varying event type vector $ET$. 
Hence, we test the null hypothesis of stochastic independence of $ET$s and $A$ against the alternative of dependence; see \citet{Bakirov, energypackage}. The test is based on the hypothesis that a multivariate distribution function (here of $A\times ET$) decomposes in a product of two distribution functions, here that of $A$ and that of $ET$.
For running the test within the framework of \citet{energypackage}, we used the function \texttt{indep.test()} with method "mvI" and parameter $R=199$. 
The test statistic is then equal to $I=0.02490639$ (with value 0 corresponding to independence) with a corresponding $p$-value of  0.3618.

Instead of estimating the spectral measure of the event type vector \ET\ explicitly, we will in the next section estimate the risk constants, which are in the focus of our investigation.

\subsection{Risk estimation}

Recall that according to Theorem~\ref{tail-risk-measures} the asymtotic behaviour of business line specific risk in terms of $\VaR$ is given by 
\begin{gather*}
\VaR_{1-\gamma}(BL_i) \sim (C^i)^{1/\alpha}\gamma^{-1/\alpha}  \mbox{ and }
\VaR_{1-\gamma}(\|BL\|) \sim (C^S)^{1/\alpha}\gamma^{-1/\alpha},\quad \gamma \rightarrow 0,
\end{gather*}
with $\CoTE_{1-\gamma}(BL_i)$ and $\CoTE_{1-\gamma}(\|BL\|)$ having a similar form with the same risk constants $C^i$ and $C^S$.

The key to find an estimate for the risk is---besides having a valid estimate for $\alpha$ in hands---handling the constants $C^S$ and $C^i$ as given in \eqref{constantsIS} from a statistical perspective. 
In Section~\ref{s31} we have estimated $\alpha$ and $K_1,\dots, K_7$, 
and we take $A$ as the realized network fraction matrix.
What remains to be done for estimating $C^i$ and $C^S$, is an estimate for the dependence structure in terms of the spectral measure $\Gamma$ of \ET.

For that purpose, we consider the general form of these constants given in \eqref{constantsIS} which incorporate the spectral measure in two different integrals. 
The spectral measure is a tool to capture the extremal dependence between the event type losses or, in other words, the dependence between rare events.
We estimate $\Gamma$ not explicitly (which we could), but since our focus is on the risk constants, we estimate those constants immediately by using the empirical version of $\Gamma$ in \eqref{constantsIS}.
As $\Gamma$ is a measure for large values of the data, 
only those multidimensional event type observations are taken into account that are large in some norm: As a natural norm we take the sum-norm of the vector such that only those sums of data enter into the estimation, which exceed a certain threshold. 
To choose this threshold is a meaningful point as it marks the border between the so-to-say usual observations and the area, where the extremes begin. In spirit this is similar to choosing the threshold values $u_j$ for the marginal distributions.
As we have spelt out there, the marginal thresholds are chosen as to allow for a best approximation of a gpd to the conditional marginal distribution. 

In a similar spirit for the spectral measure we choose a threshold, which approximate the product limit in \eqref{multrv} in an optimal way; see \citet{WD,energypackage}. 

Besides the estimates for the tail index $\alpha$ which we obtained in Section~\ref{s31}  and  the spectral measure representing dependence information between the event types, one other ingredient of the model has not yet been mentioned.
The network of the bipartite graph represented by snapshots of the weekly fraction matrices need to be incorporated as well. 
This network first implies a change of dimension from the seven event types to the eight business lines as the network fraction matrix is applied to the event type vector. As a result, we obtain the dependence structure  between the business lines.
All that brought together leads to formulae for the estimators of the risk constants as in \eqref{est_ru}, \eqref{est_agent},  and capital allocation constants \eqref{est_ru}.

As discussed in Section \ref{model}, the multivariate regularly varying model allows us not only to estimate the VaR and CoTE at aggregated level for the entire system, but also at the level of a single business line. 
The two situations are described below for the VaR.

\subsection{Risk estimation at the individual and the system level}

In order to find estimates for the system VaR, we estimate the constant $(C^S)^{1/\alpha}$ based on the marginal estimates and formula \eqref{est_ru}, which takes care of the dependence structure. 
This constant can then be scaled with the factor $\gamma^{-1/\alpha}$ to obtain an approximation of the VaR at confidence level $1-\gamma$ for small $\gamma$.

Similarly to various estimators for the tail index $\al$, there also  exist different estimators for the scaling constants $K_j$ for $j=1,\dots,7$.
In Section~\ref{subsec:estimation_procedure},  we outline two methods to estimate the scaling constants $K_j$ for $j=1,\dots,7$ for the tail behaviour of each event type as in \eqref{Pareto}. 
In our study we have compared both estimates and the resulting estimates for the risk constants $(C^S)^{1/\alpha}$, the business line specific constants $(C^{i})^{1/\alpha},\ i=1,\dots,8,$ from \eqref{constantsIS} and for the capital allocation constants from \eqref{capall}.
Our conclusion is that the semi-parametric estimator \eqref{K-estimatorMVR} yields slightly different values than the parametric estimator \eqref{K_POT}. It is, however, remarkable that both estimates yield the same order of $\wh K_1,\dots,\wh K_7$.
Moreover, \eqref{K_POT} results in more stable estimators than the parametric estimator \eqref{K_POT}.

\begin{figure}
\subfigure{\includegraphics[width=8cm]{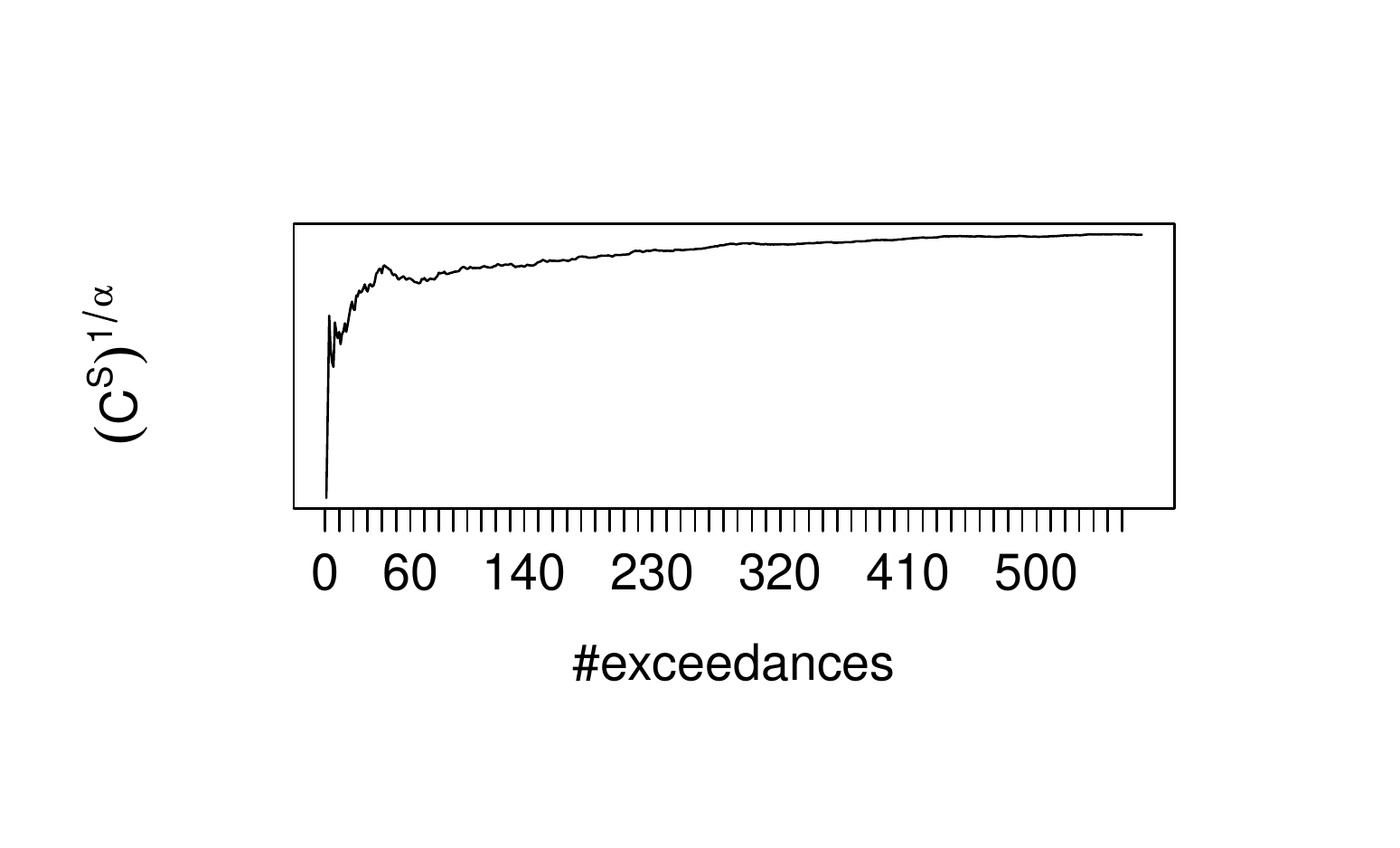}}\hfill
\subfigure{\includegraphics[width=8cm]{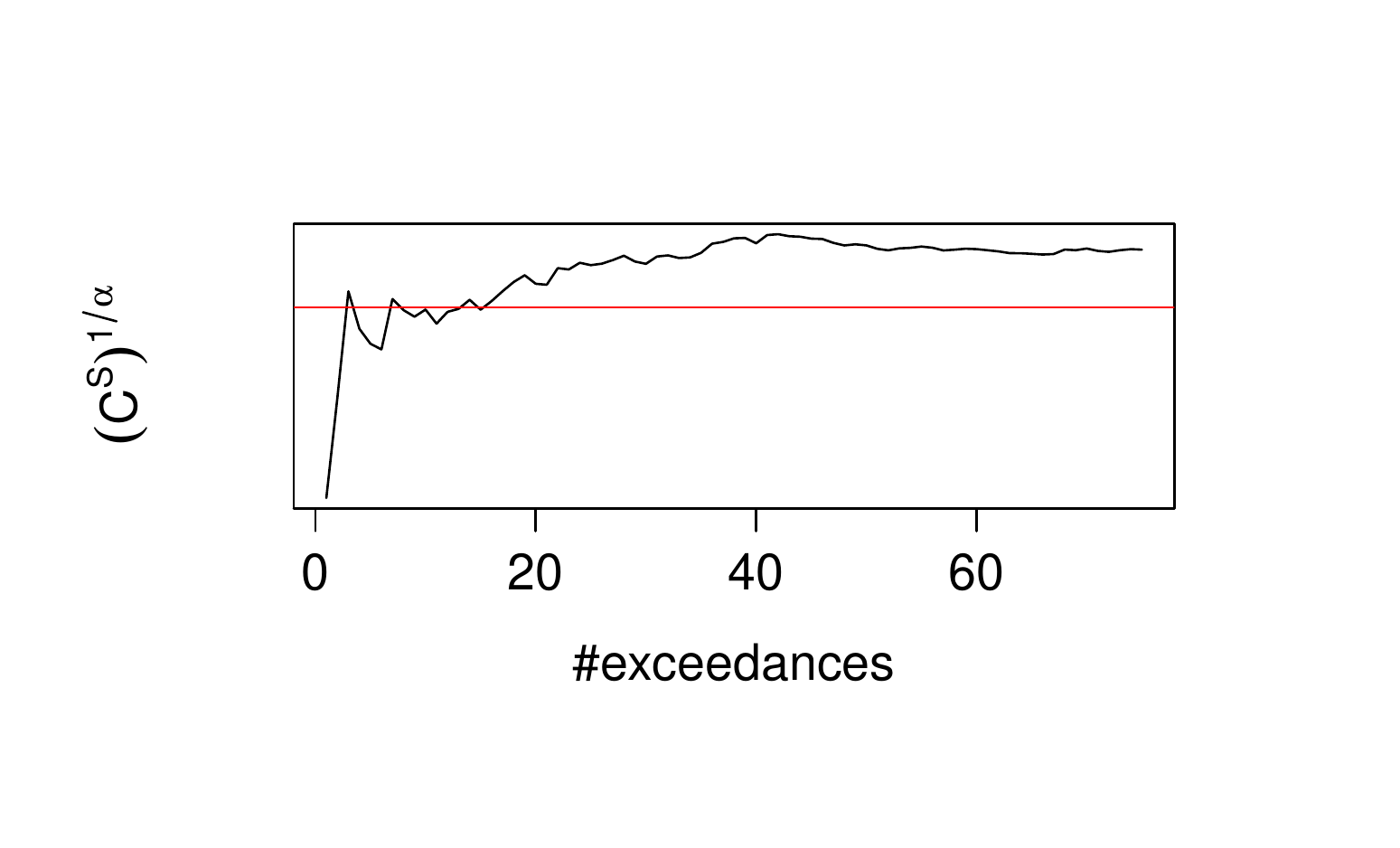}}\\
\vspace{-1cm}
\caption{\label{OPdatawithoutPOT}
VaR constants $(C^S)^{1/\alpha}$ estimated from the DIPO data based on estimates for the scaling constants $K_j$ via \eqref{K-estimatorMVR}. The right hand plot is a zoom into the left hand plot. The red horizontal line corresponds to 
$(\sum_{j=1}^{7} (\wh  K_j)^{1/\alpha}$, which is the value corresponding to asymptotic independence between the event types. }
\end{figure}

We are not allowed to report numerical values, however from \eqref{systembounds}, we find together with the 
 estimated $\wh\al\ge 1$ that $C_{ind}^S\le C_{\Gamma}^S$ indicating the dependence for the business lines.
The risk constant $C_{ind}^S$ can be computed explicitly, since by \eqref{Sum} the VaR of $\|BL\|$ and $\|ET\|$ are the same.

Finally, by considering the bounds for {$\wh\al\ge 1$ for the risk constants from \eqref{individualbounds}} at different confidence levels $1-\gamma$, we illustrate the gap between the VaR estimates based on the two extreme dependence structures of asymptotic independence and asymptotic full dependence for the DIPO data.
In Figure~\ref{depindfigure} we plot  the estimates of the functions $(C^S_{ind})^{1/\alpha}\gamma^{-1/\alpha}$ (dashed) and $(C^S_{dep})^{1/\alpha}\gamma^{-1/\alpha}$  (solid). 
As the solid curve corresponds to  comonotonicity between the different event types, requiring then to sum up the Values-at-Risks of all the event types, we immediately notice that taking into account the multivariate dependence structure leads to considerably lower estimates of the VaR with lower bound given by asymptotic  independence,
Figure~\ref{CiOPBLb} displays
the results for the VaR constants $(C^{i})^{1/\alpha}$ at individual level for both methods. 
A risk constant running below the asymptotic independence bound indicates dependence in the data.

\begin{figure}
\begin{center}{\includegraphics[width=10cm, height=7cm]{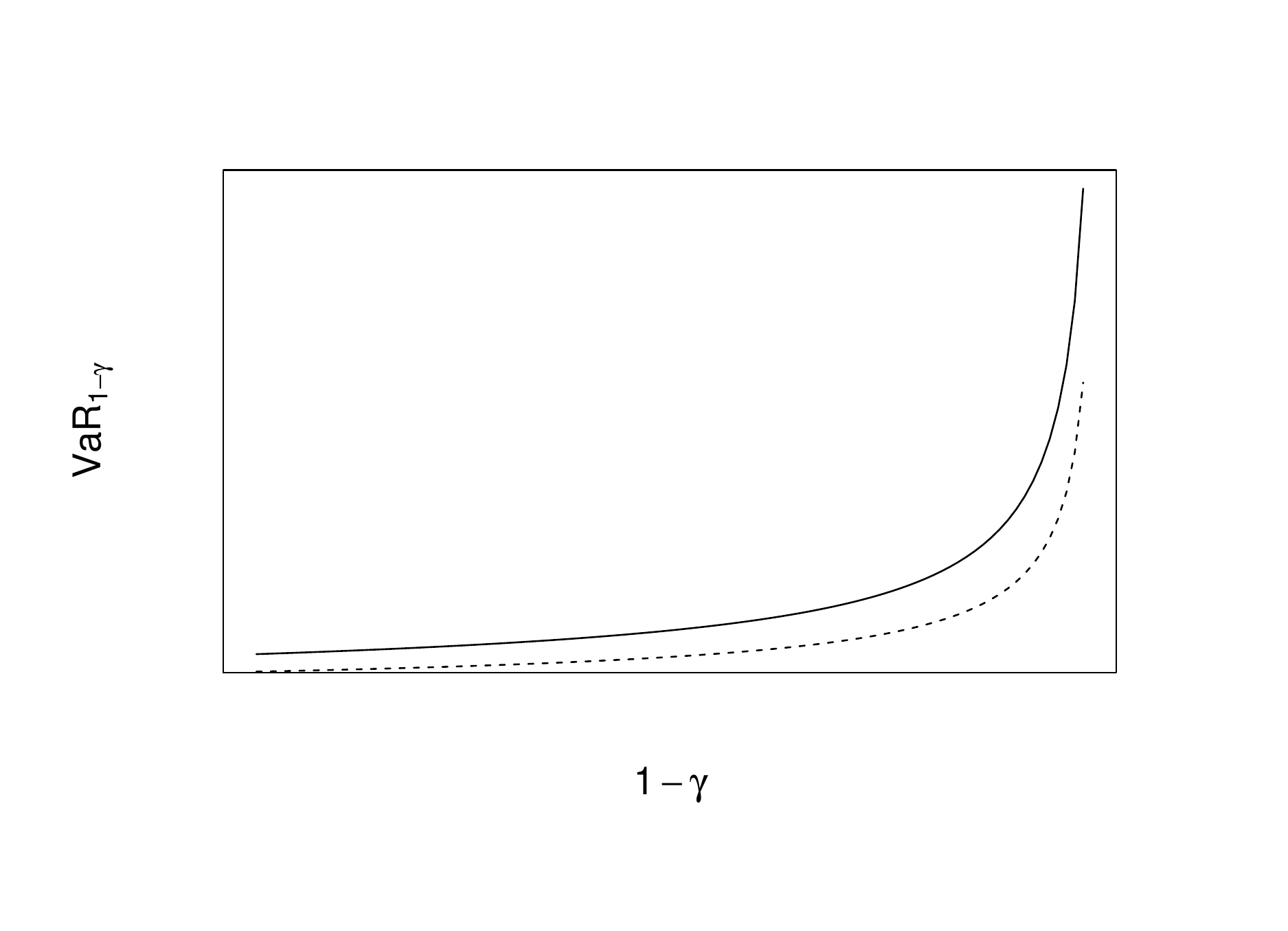}}\end{center}
\vspace{-1.5cm}
\caption{\label{depindfigure}  
VaR bounds estimates of the DIPO data for $(C^S_{ind})^{1/\alpha}\gamma^{-1/\alpha}$ for the asymptotically independent case (dashed) and $(C^S_{dep})^{1/\alpha}\gamma^{-1/\alpha}$ for the asymptotically dependent case (comonotonic) case (solid).  }
\end{figure}

\begin{figure}[htpb!]
\centering
\begin{tabular}{cc}
\subfigure{\includegraphics[width=8cm]{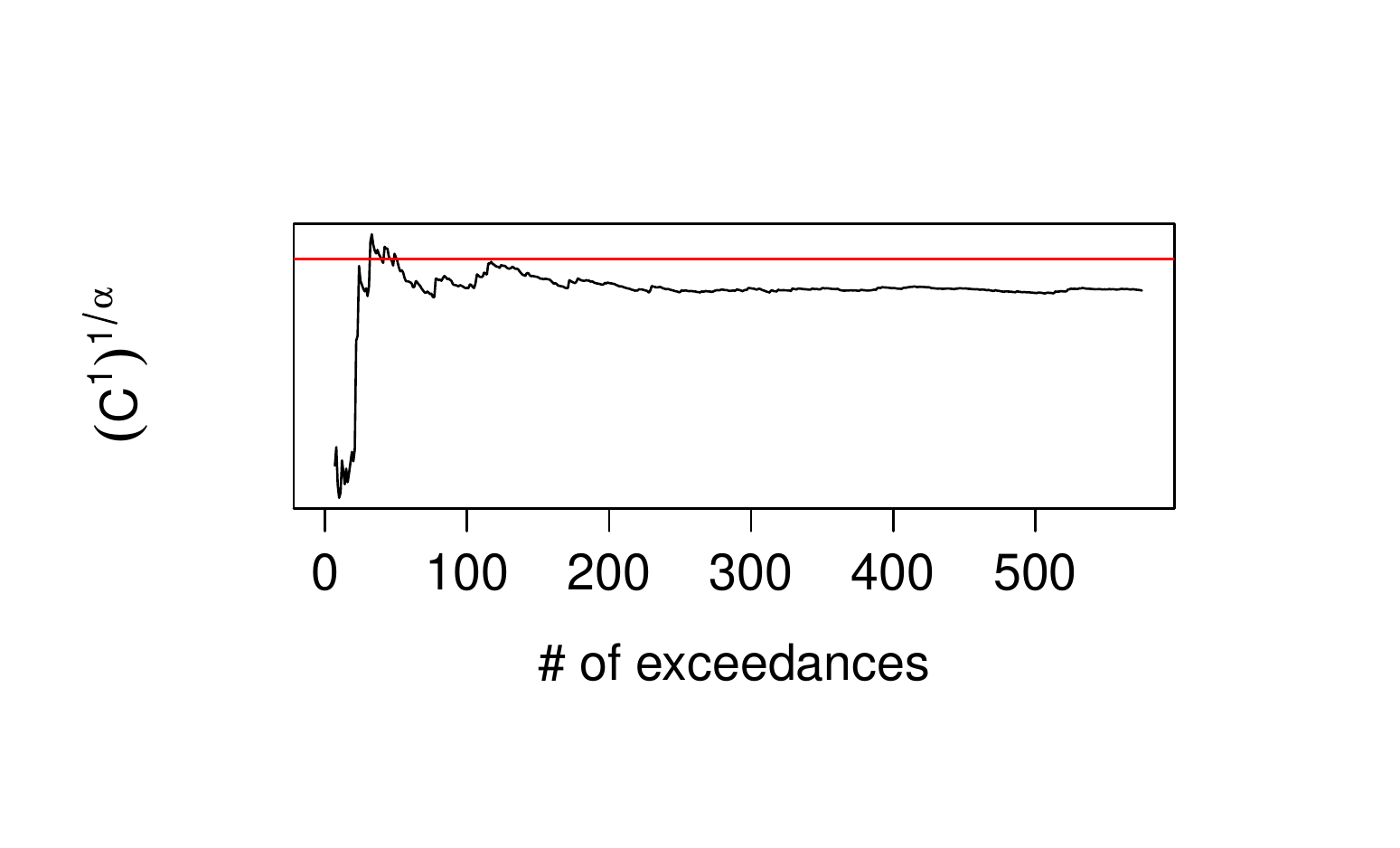}} 
   & \subfigure{\includegraphics[width=8cm]{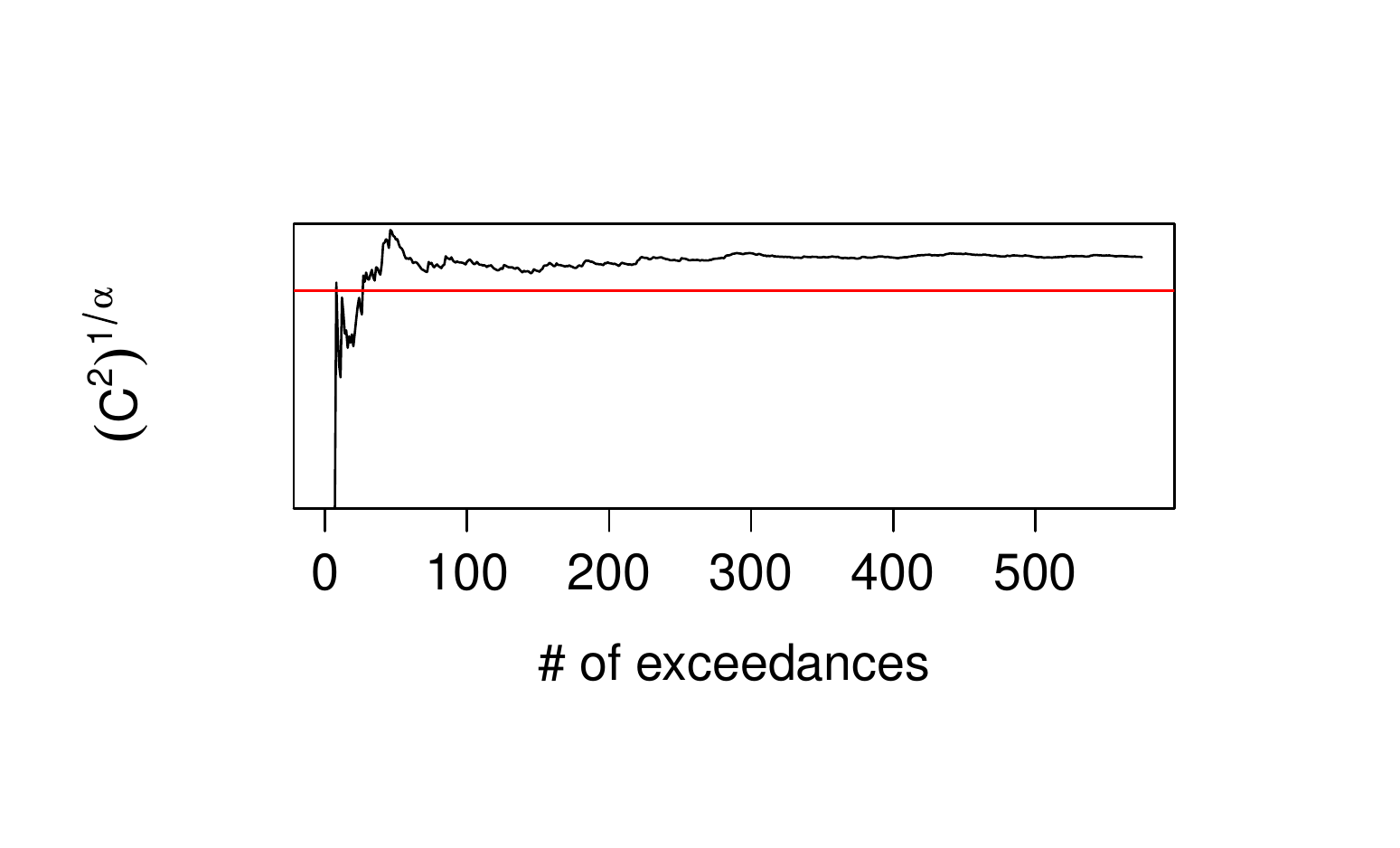}} \\[-1cm]
 \subfigure{\includegraphics[width=8cm]{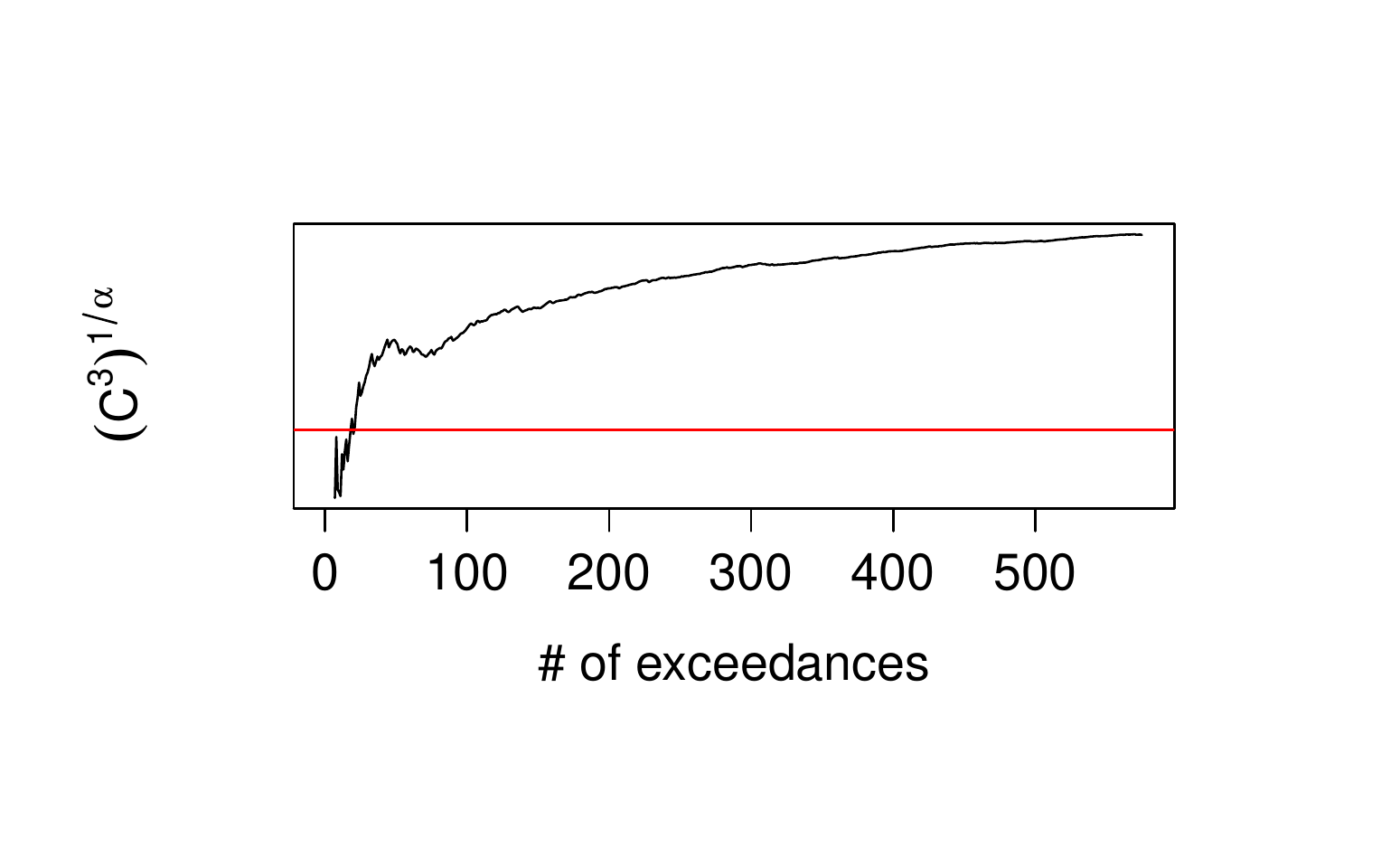}}   
& \subfigure{\includegraphics[width=8cm]{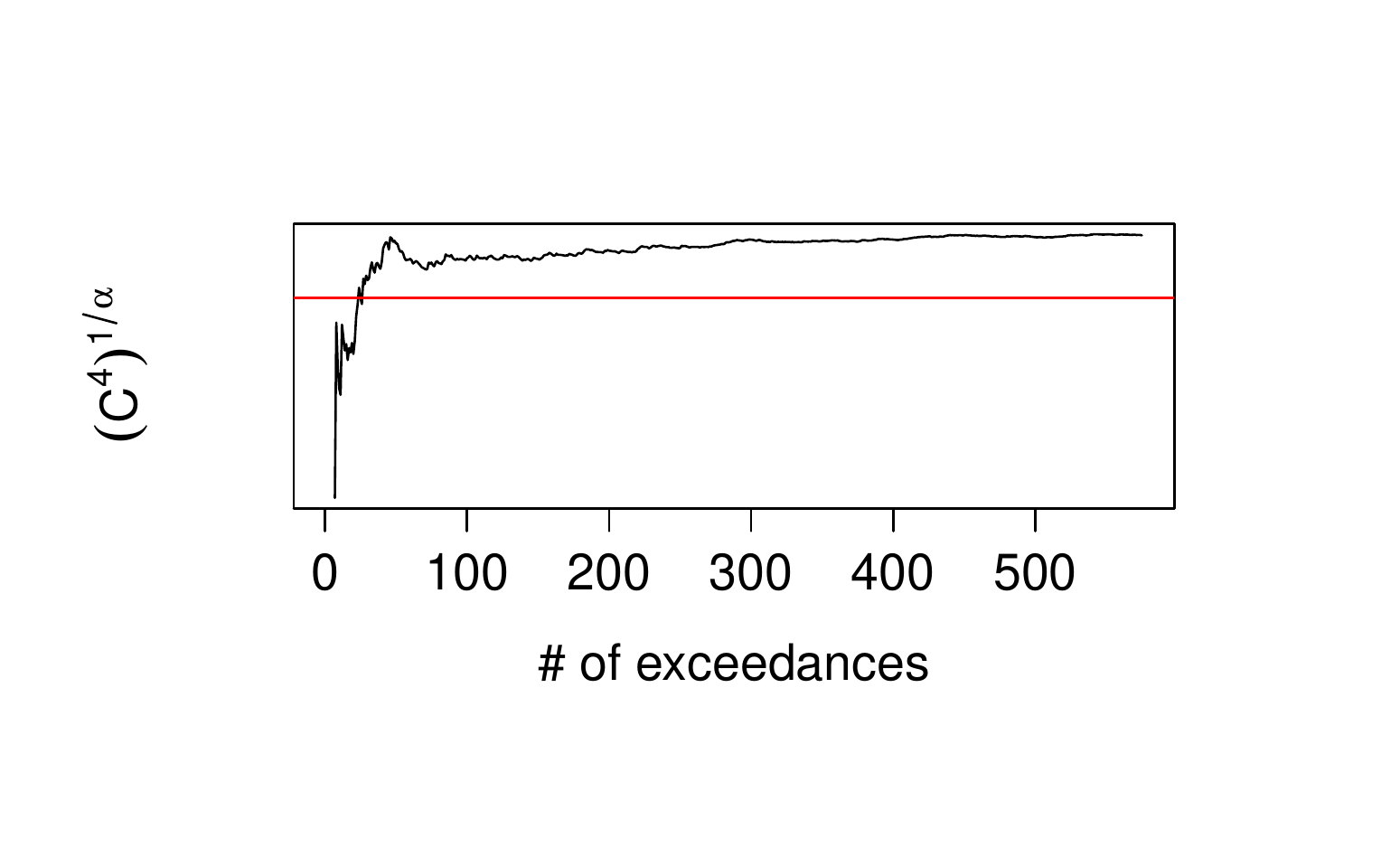}}  \\[-1cm]
\subfigure{\includegraphics[width=8cm]{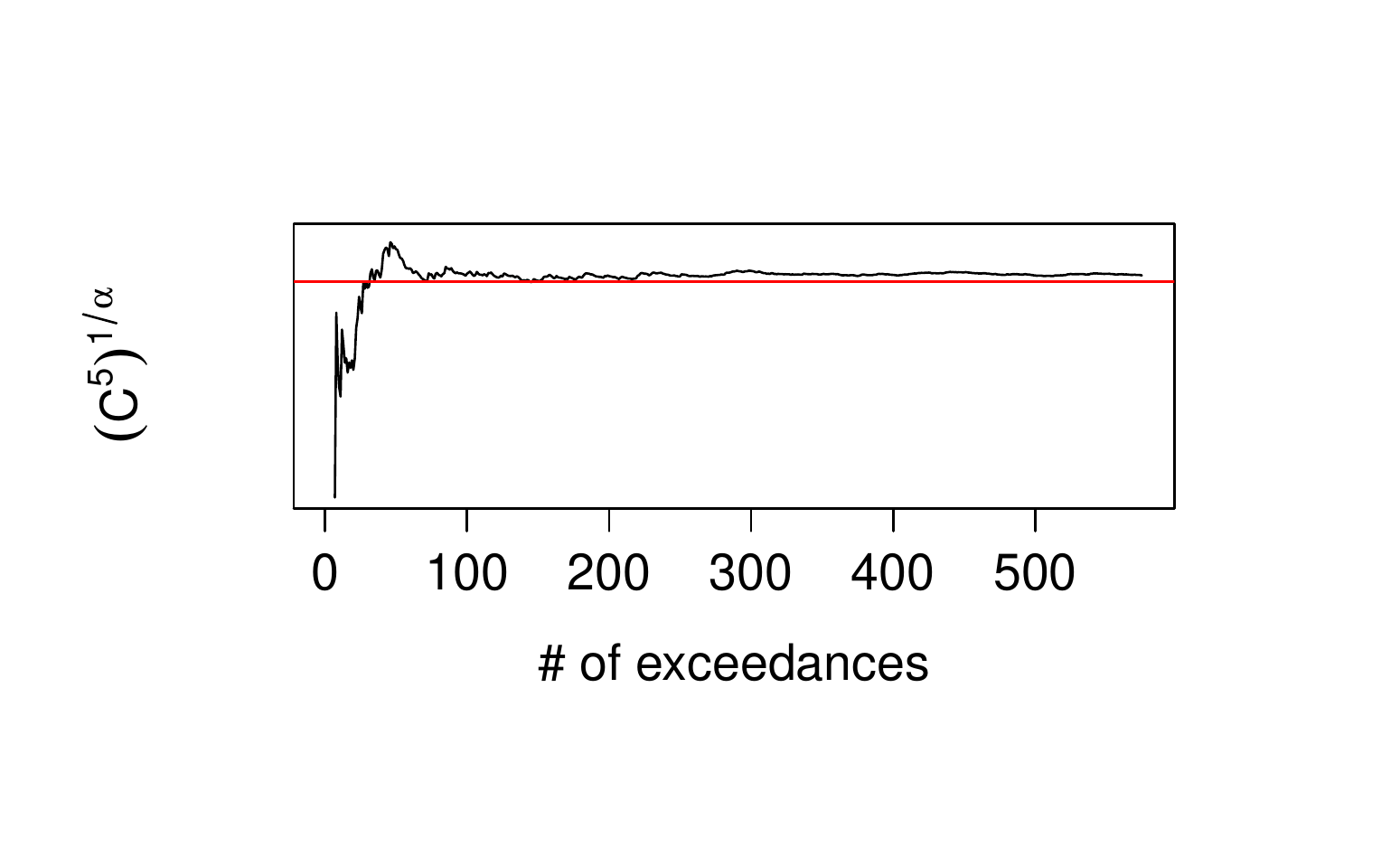}} 
   & \subfigure{\includegraphics[width=8cm]{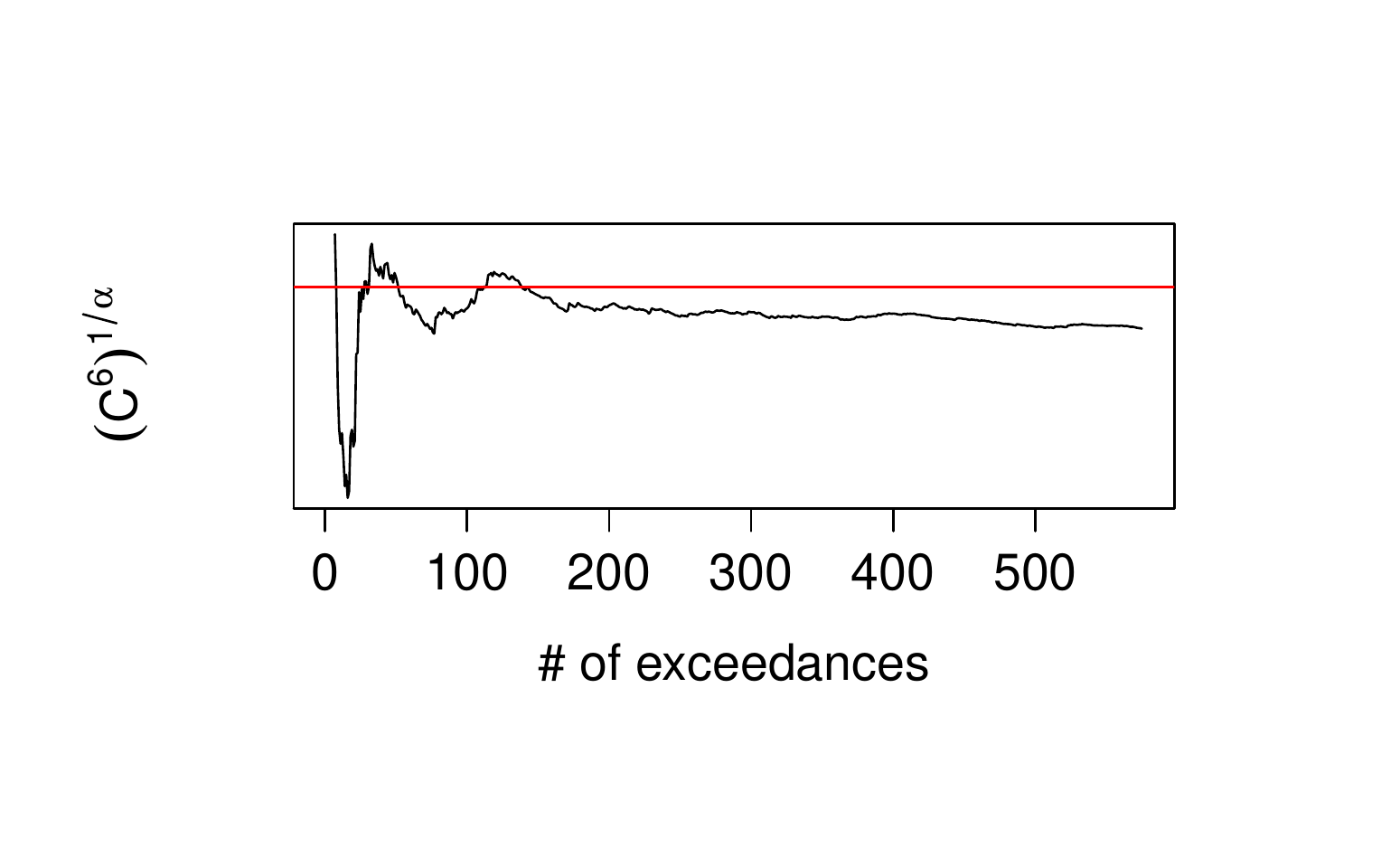}} \\[-1cm]
 \subfigure{\includegraphics[width=8cm]{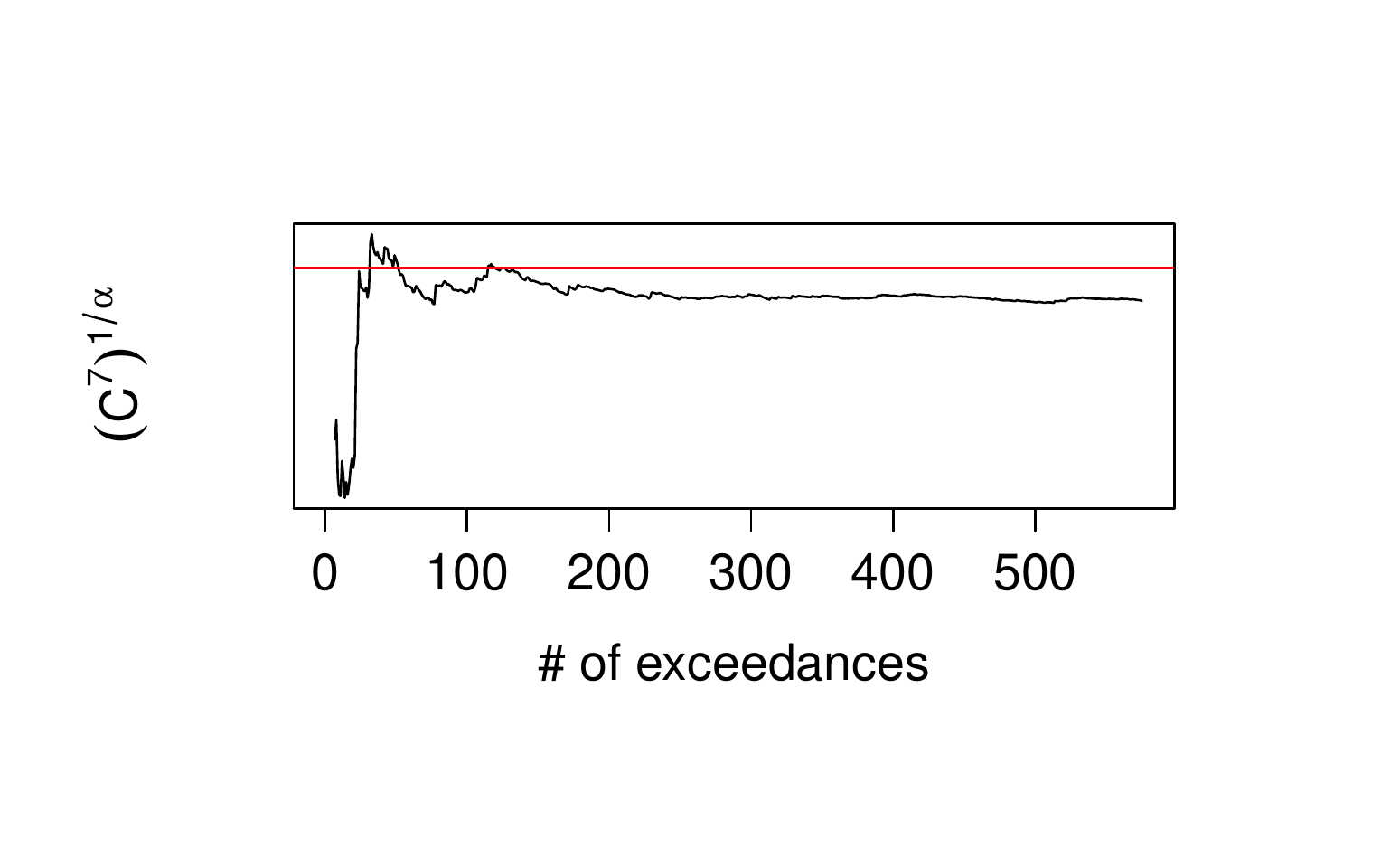}}   
& \subfigure{\includegraphics[width=8cm]{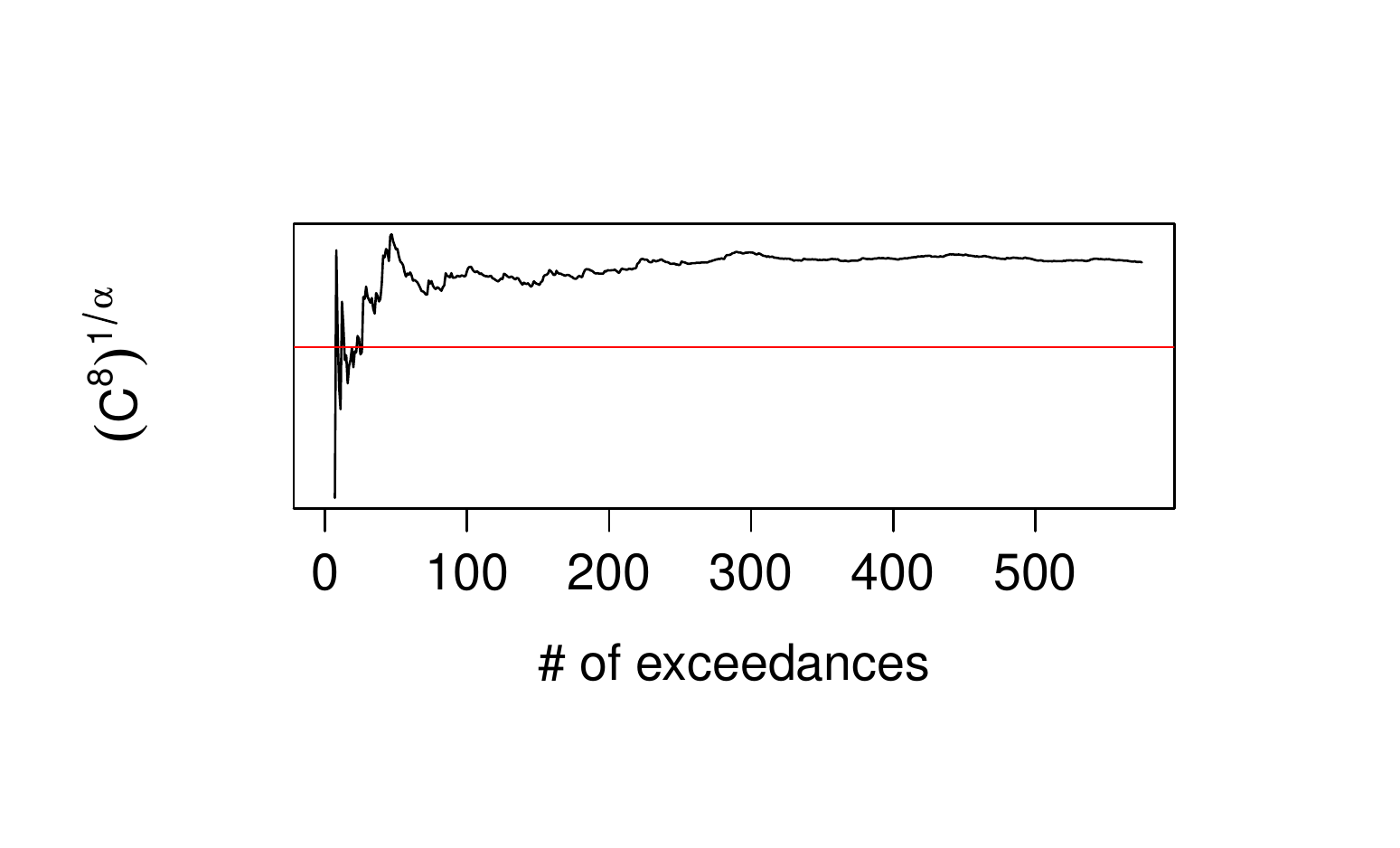}}  
	\end{tabular}
\vspace{-0.5cm}
\caption{ The VaR constant $(C^{i})^{1/\alpha}$ estimated for the business lines $BL_i$ for $i=1,\dots,8$ with the estimates for the scaling matrix $K^{1/\alpha}$ via \eqref{K-estimatorMVR}. 
The red horizontal lines correspond to the values for asymptotic independence between the event types.}
\label{CiOPBLb}
\end{figure}

\subsection{Capital allocation---estimation of the risk contributions}

The estimation of  the risk contributions \eqref{riskcontvari}  and \eqref{riskcontcotei} requires beside an estimate for the tail index $\al$ and  for the risk constant $C^S$, which is given as an integral with respect to the spectral measure $\Gamma$ involving the network fraction matrix $A$, another such integral, which we estimate by its empirical counterpart given in \eqref{est_agent_riskcont}.

The estimation  results for the capital allocation to each business line are shown in Figure~\ref{SeverityDist}, when relying on the POT method. 
We notice immediately, that the capital allocation to each business line are stable already for small $k$ values, between 100 and 200,  and slightly understimating when compared to the theoretical one for $BL_1$, $BL_2$, $BL_5$, $BL_6$, and $BL_7$, overlapping for $BL_4$  and $BL_8$ and strongly underestimating for $BL_3$. 
This is not surprising if we review the descriptive statistics of the DIPO data at the beginning of Section~\ref{opRiskData}. 
In fact, $BL_3$ includes about 43\% of the total empirical severity of the losses with no missing observations. 
Hence, $BL_3$ has---as to be expected---the largest share in the capital allocation. 
Then, $BL_4$ and $BL_8$ have also no missing observation and respectively 18.4\% and 24.3\% of the total severity of the losses, which then allows to reach estimates close to the theoretical ones, even for small $k$. Figure~\ref{VaRFinal} allows to also point out the variability in the VaR estimates for $k=100$ when considering different business lines and confidence levels. As expected, the availability of more observations and of largest total severity results then in a larger variability of estimates for $BL_3$, $BL_4$ and $BL_5$, especially for very high confidence levels.

\begin{figure}[htpb!]
\centering
\begin{tabular}{cc}
\subfigure{\includegraphics[width=8cm]{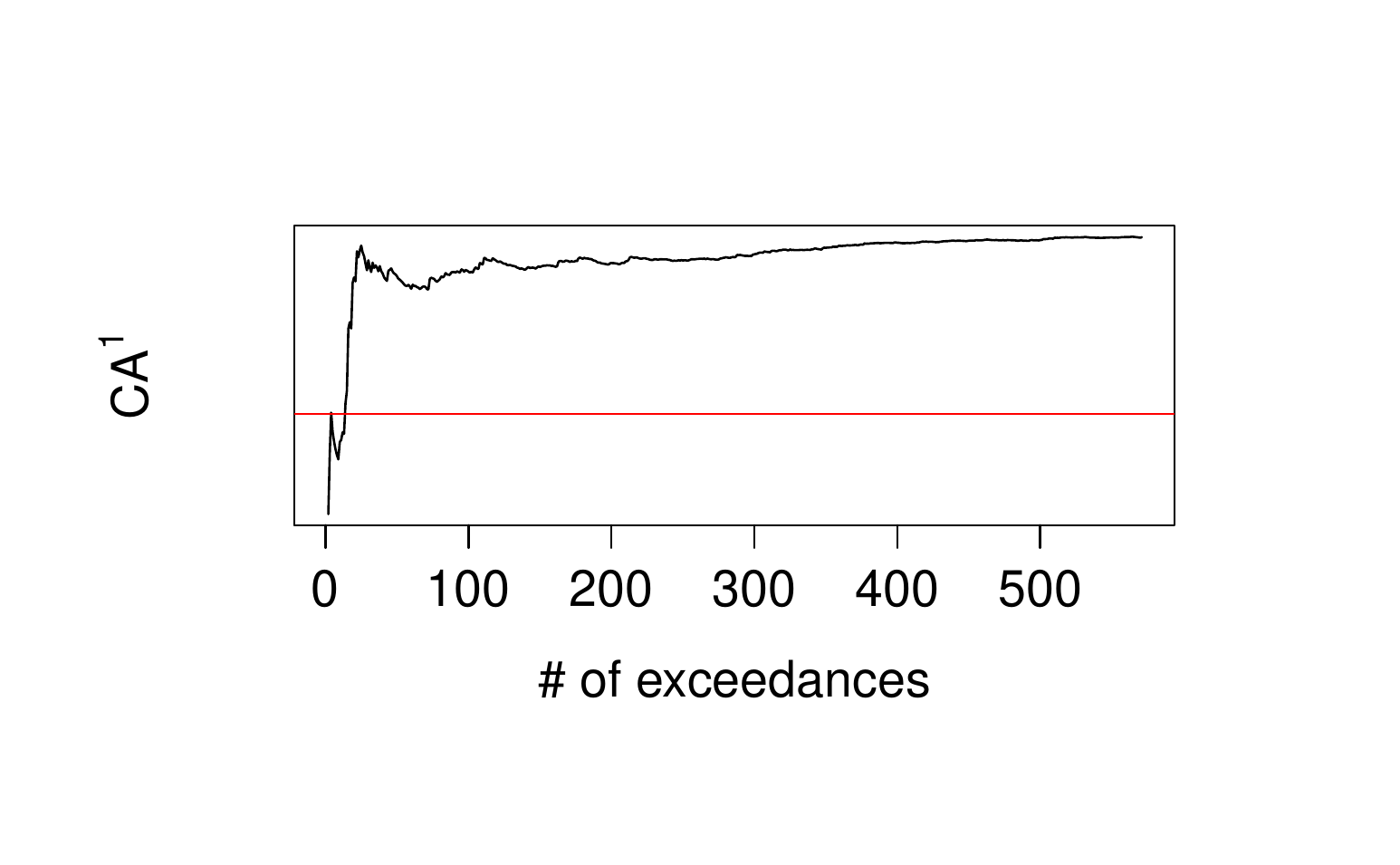}} 
   & \subfigure{\includegraphics[width=8cm]{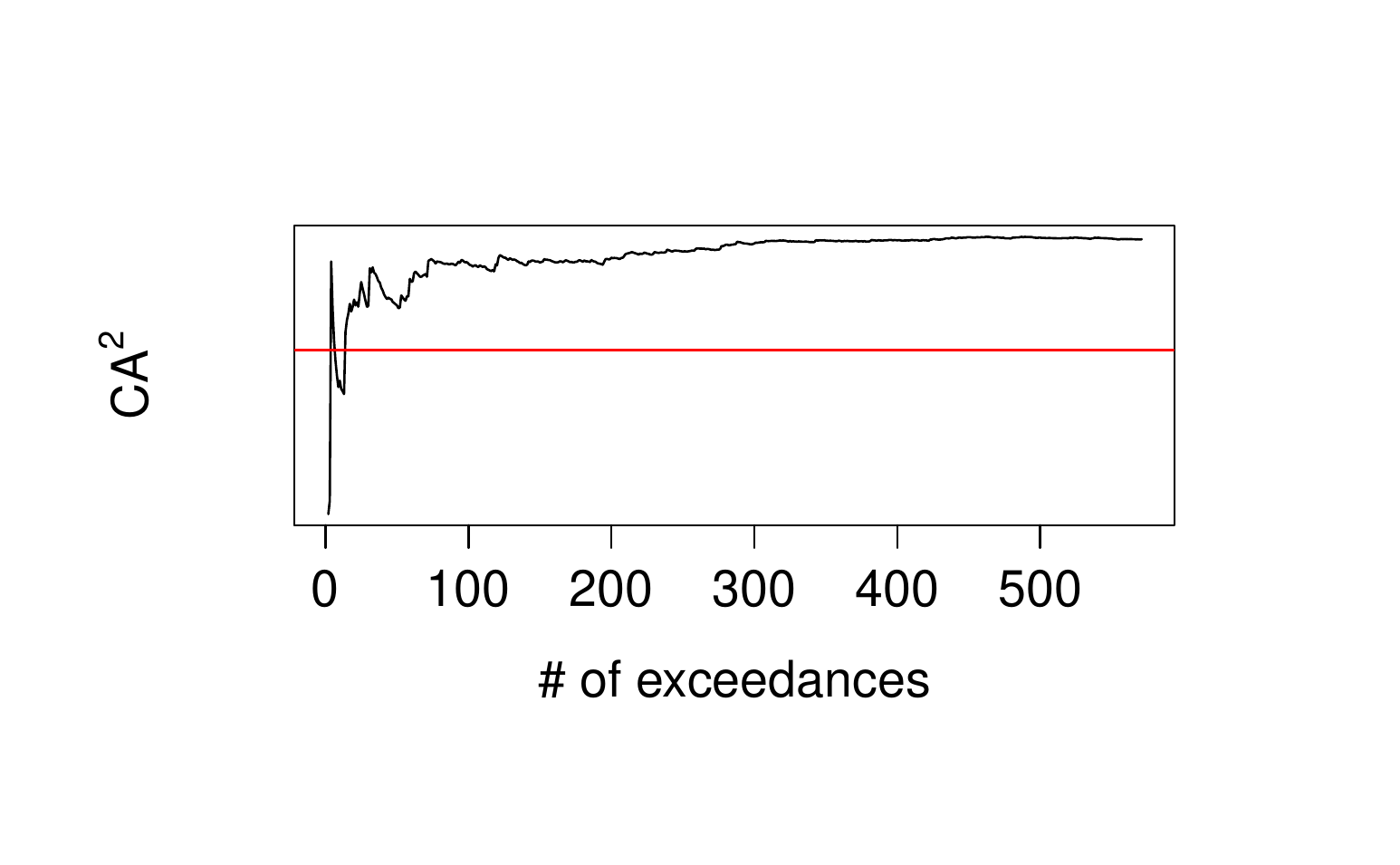}} \\[-1cm]
 \subfigure{\includegraphics[width=8cm]{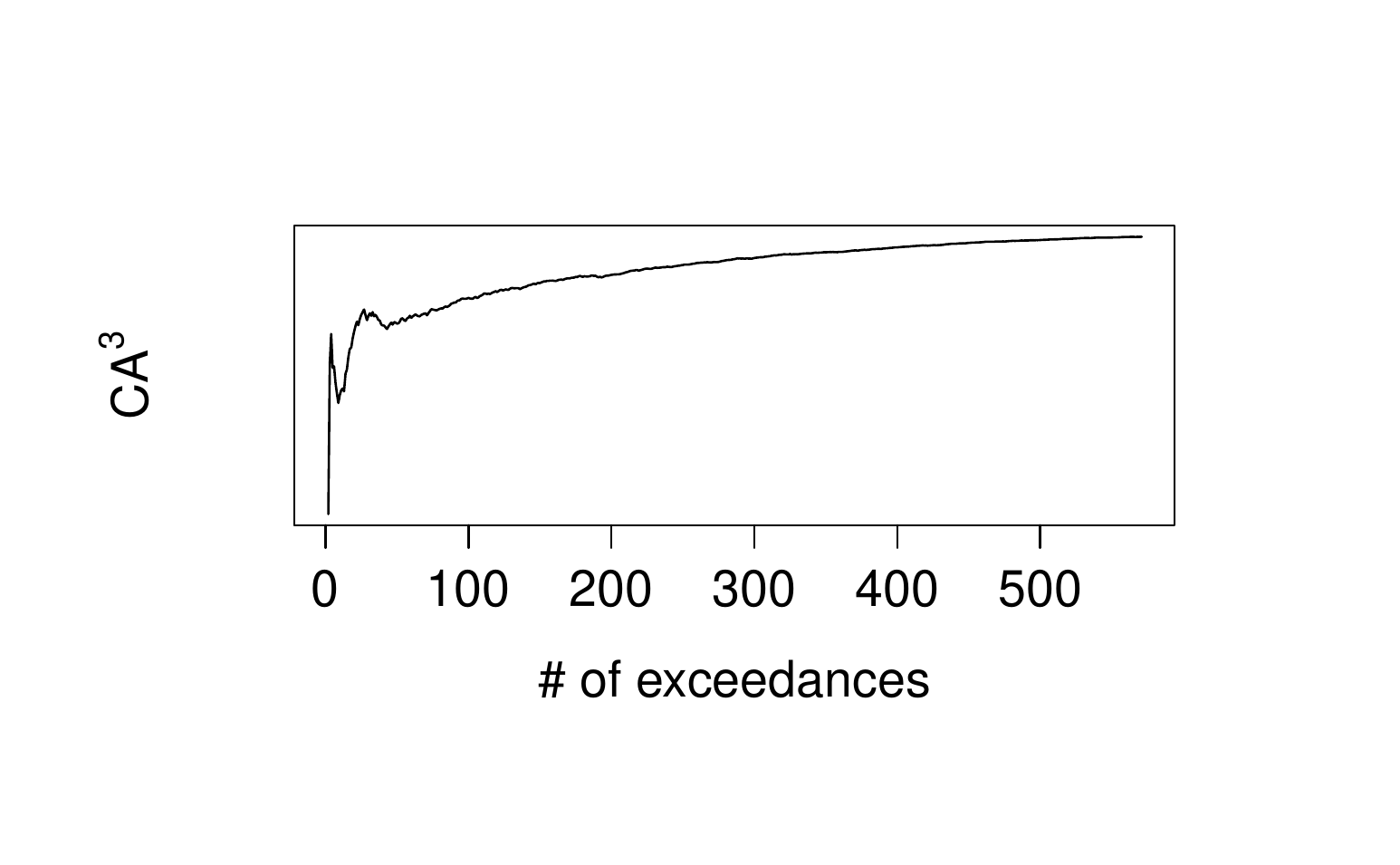}}   
& \subfigure{\includegraphics[width=8cm]{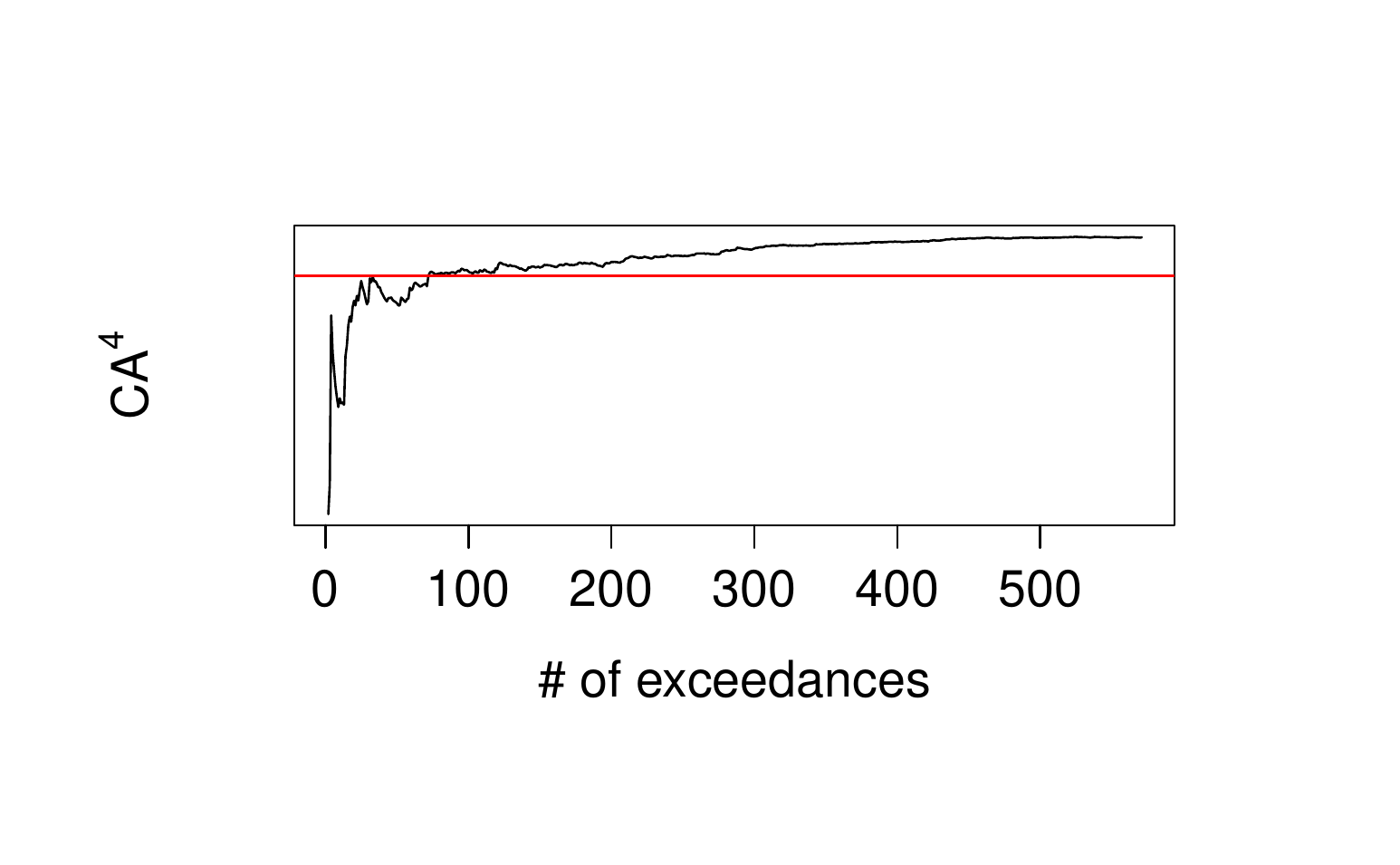}}  \\[-1cm]
\subfigure{\includegraphics[width=8cm]{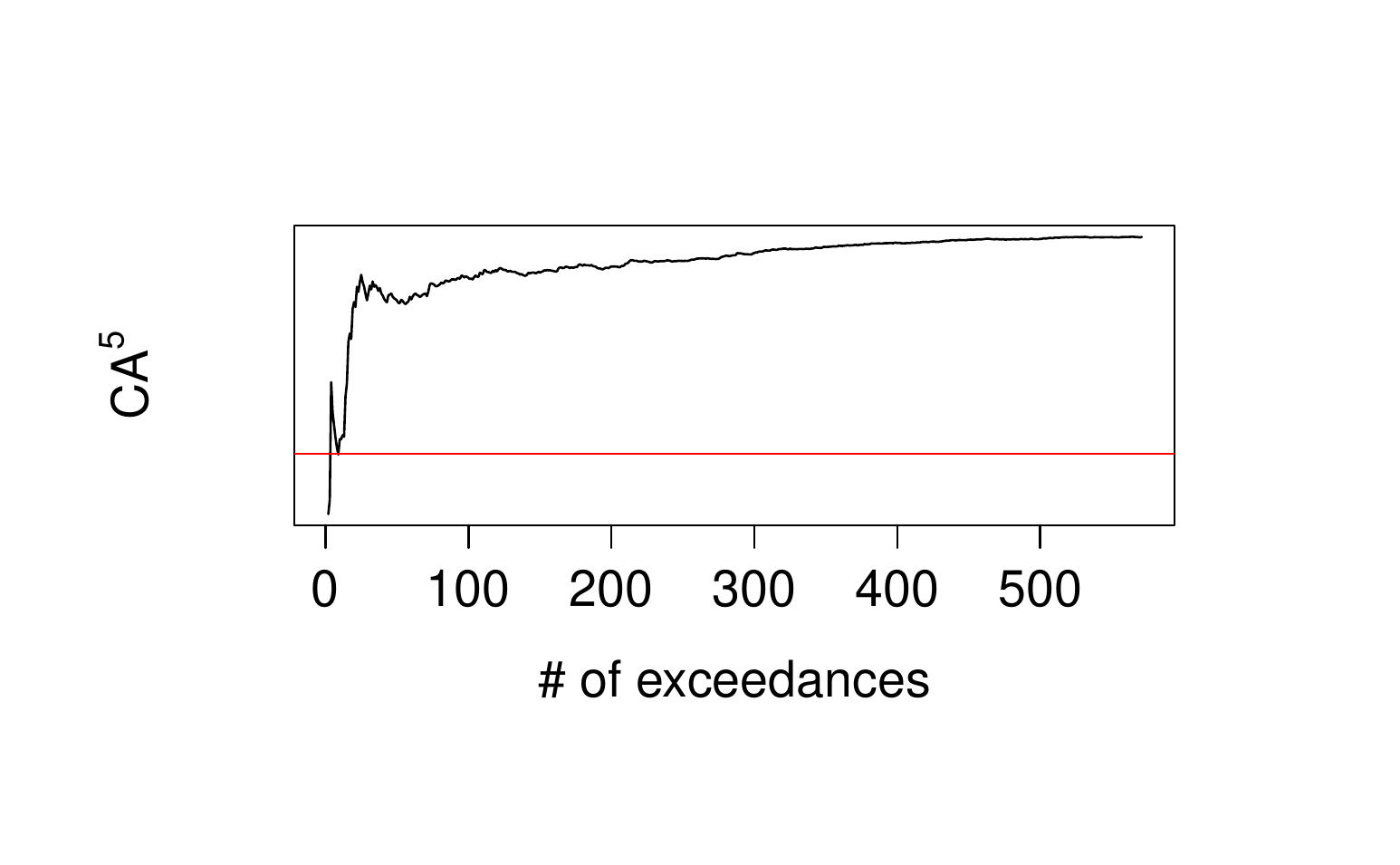}} 
   & \subfigure{\includegraphics[width=8cm]{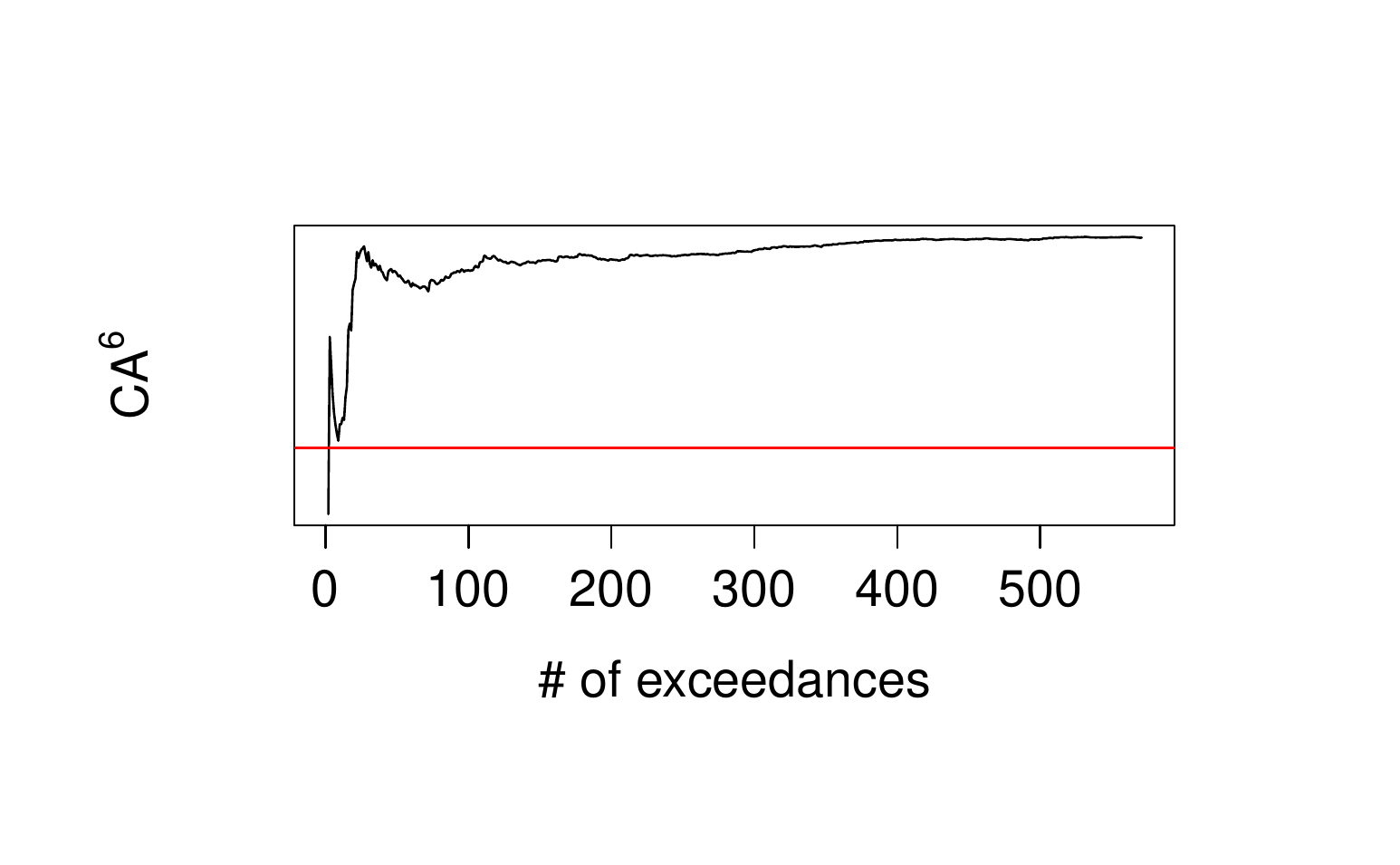}} \\[-1cm]
 \subfigure{\includegraphics[width=8cm]{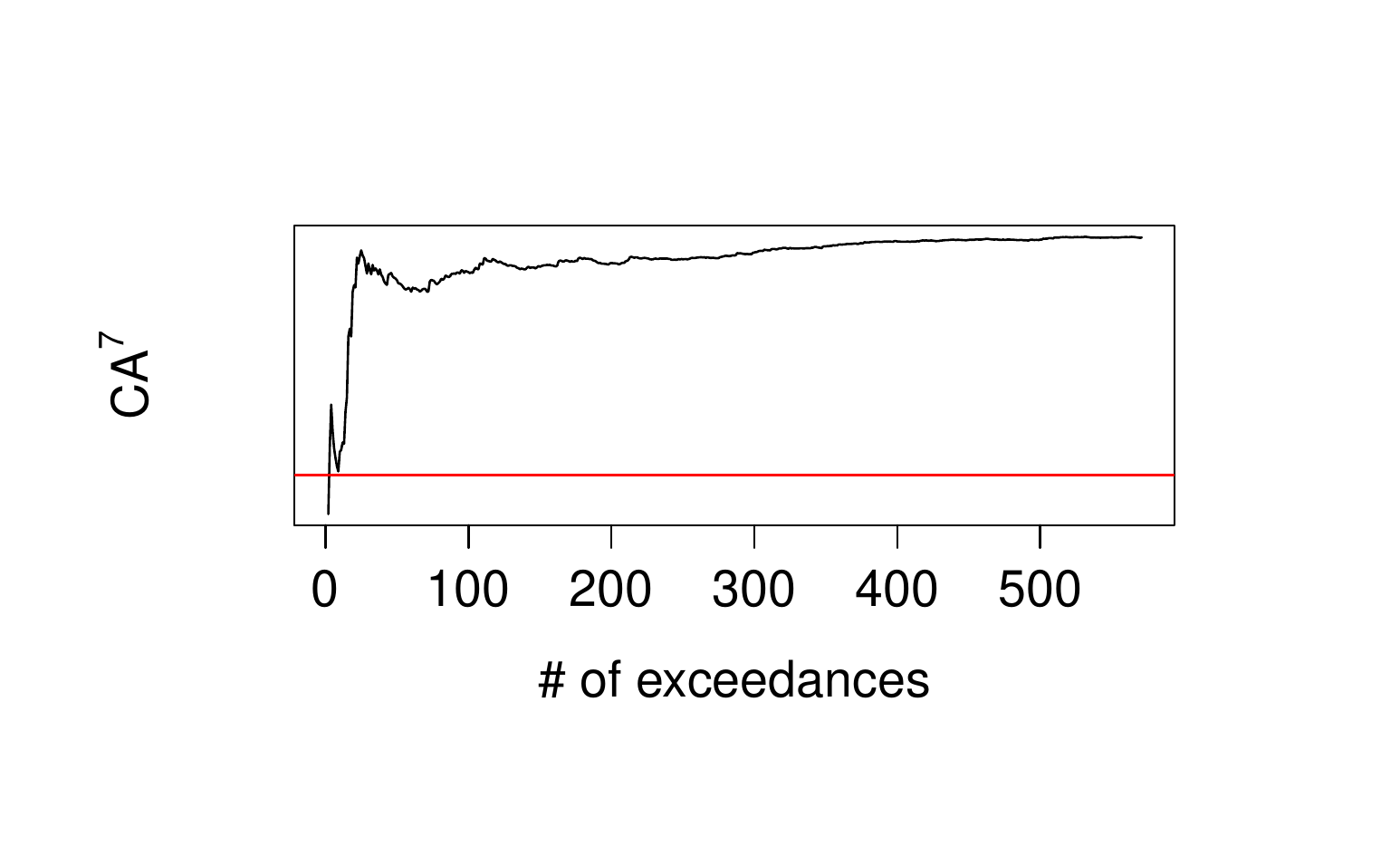}}   
& \subfigure{\includegraphics[width=8cm]{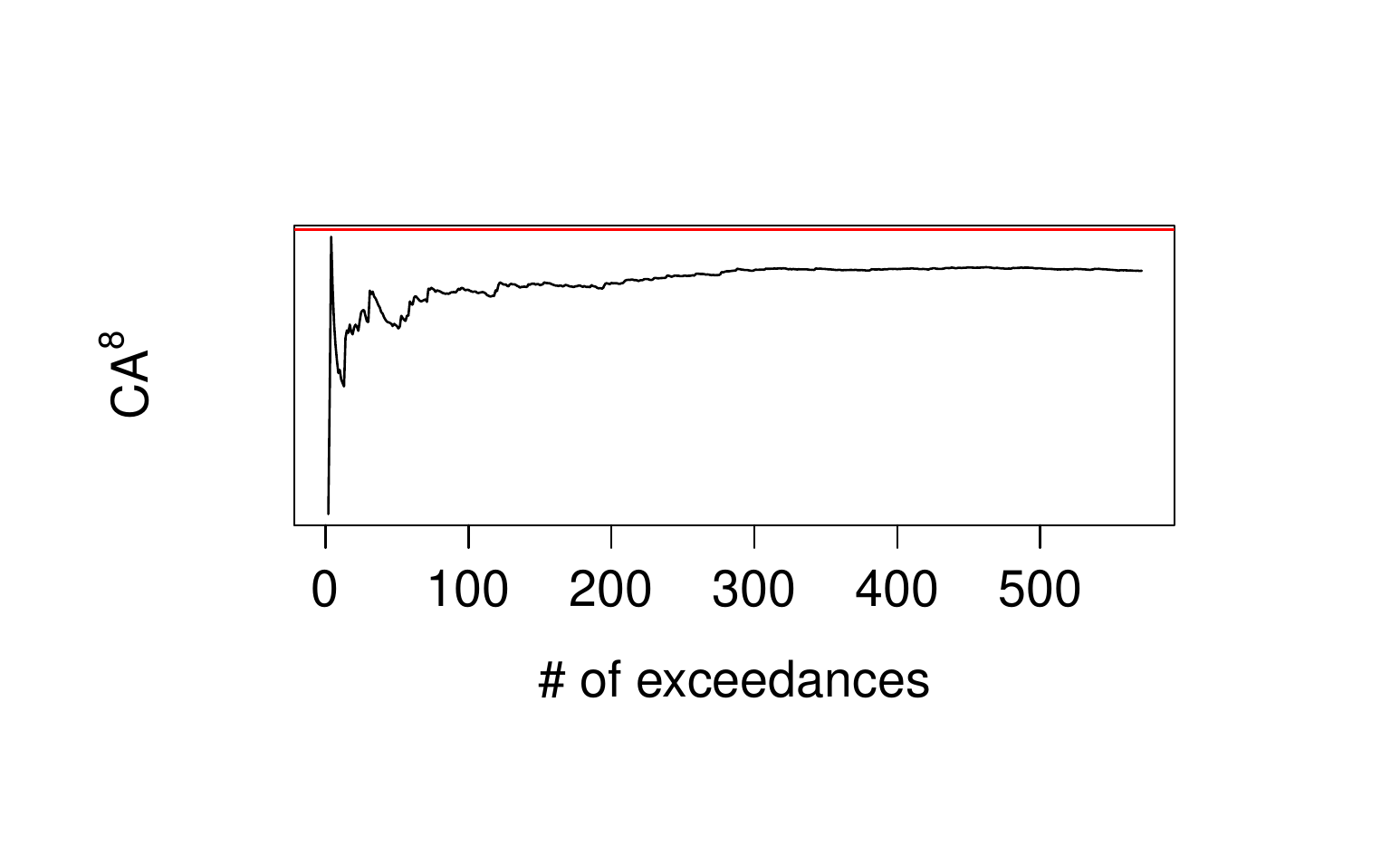}}   
	\end{tabular}
\vspace{-0.5cm}
\caption{Capital allocation constants ${CA}^i$ estimated for the business lines $BL_i$ for $i=1,\dots,8$.
We present the estimated risk contributions based on \VaR\ from \eqref{riskcontvari}. 
The red horizontal lines correspond to the values for asymptotic independence between the event types. 
For  $CA_3$ it is outside the range.}
\label{SeverityDist}
\end{figure}

\begin{figure}
\begin{center}
\includegraphics[width=7cm, height=6cm]{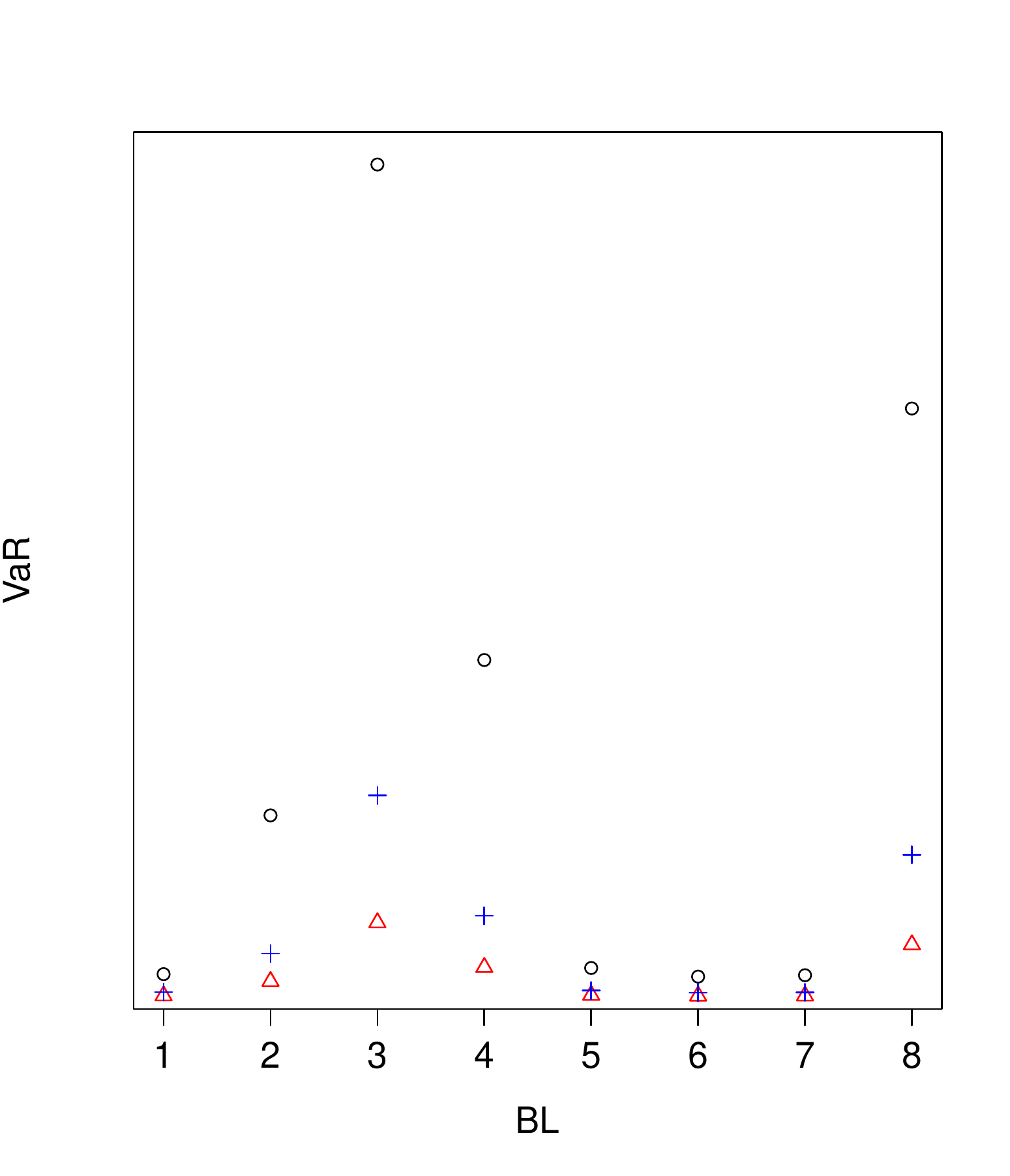} 
\end{center}
\vspace{-0.7cm}
\caption{\label{VaRFinal}  VaR estimates for the business lines at confidence levels   95\% (red triangle), 99\% (blue cross), and 99.9 \% (black circle) }
\end{figure}

\section{Simulation study}\label{Sec:SimulationStudy}

To evaluate and illustrate the performance of the estimators for the risk constants and the risk contributions  in Section~\ref{subsec:estimation_procedure}, we study their behavior in a simulation set-up, where closed-form asymptotic solutions are available.
Also, we shed some light on the issue of choosing the appropriate number of exceedances $k <n$ in \eqref{final_est_general}.
The marginals $X_j$ of  the loss vector $X$ are simulated independently, corresponding to independent event types. 
This enables us to  compare the outcome of the simulation with a value easily computable from the marginals, as the spectral measure is concentrated on the axes.

The marginal distribution of the loss vector $X$ is for every component $j=1,\dots,d$ simulated according to a mixture distribution  
\begin{gather}\label{mixture}
F_j(t)=\mathds{1} \{t<u_j \} H_j(t) +  \mathds{1} \{ t\geq u_j \}( H_j(u_j) + (1-H_j(u_j))G_j(t) ),\quad t\ge 0,
\end{gather}
 with $H_j$ lognormal, $G_j$ gpd with $u_j$ being the threshold dividing the lognormal and the gpd parts. 
Such marginal distribution is a typical choice for operational risk losses.     

We simulate $m$ independent blocks of data (i.e. $m$ Monte Carlo runs), each consisting of a number $n$ of $d-$dimensional observation vectors. 
For each of these blocks we then apply the estimation procedure as detailed in Appendix~\ref{subsec:estimation_procedure}; i.e., first estimating the the marginal distributions for $j=1,\dots,d$. 
After applying an automated procedure to detect the threshold value, the GPD shape and scale parameters 
are estimated. In a next step, aiming at a common tail index, the mean $\hat\xi_{mean}$ over the estimated shape parameters $\wh\xi_j$ for $j=1,\dots,7$ is taken and the scale parameters are re-estimated given $\hat\xi_{mean}$. 
This is senseful as the estimators for shape and scale are dependent, and the re-estimation takes care of the deviances from having taken the mean for the shape parameters. The common tail index $\wh{\al}$ is then the reciprocal of the mean shape parameter $\hat\xi_{mean}$. In the consequent step, the constants $K_j^{1/\wh\alpha}$ are estimated 
by \eqref{K-estimatorMVR}.  
We use these estimated marginal parameters for scaling the data as in \eqref{scaled_observations} and compute the risk constants \eqref{est_ru} and \eqref{est_agent} as well as the capital allocations via \eqref{est_agent_riskcont}.

The following simulations do have in common the choice of $n=1\,000$ and $m=500$. 
Furthermore and unless stated otherwise, we show boxplots for a range of different values of $k$ denoting here the number of exceedances of the {radial parts}  as in \eqref{radial_parts} used for the estimation,  ranging from  200 to 10  with stepsize equal to 10. 
We examine, in particular, the estimation  of the VaR constant $(C^S)^{1/\alpha}$ at system level.

Then, we report here the two following {representative} marginal scenarios \footnote{More results are available upon request.}
\begin{itemize}
\item {Parameter Scenario 1)}: 
We set the dimension $d=3$, for the gpd with  $\xi=1/\alpha=0.5$,  $\beta=(1,5,50)$, $u=(1,10,50)$ and for the lognormal, expected value $\mu=(0,1,2)$ and standard deviation  $\sigma=(1,2,4)$ (on the log-scale). 
\item {Parameter Scenario 2)}:   
We use the parameters estimated from all event types of the real-world DIPO data by fitting a log-normal distribution below  and a gpd above the thresholds (see Table~\ref{tab:parest} in Section~\ref{s32} for the full parameter specification). 
\end{itemize}

\subsection{Estimation for aggregated risk}\label{subsection:simulation:aggregated}

The results for Parameter Scenario 1) {are} plotted in Figure~\ref{MIXsimple},   which   displays  the boxplots for the constant $(C^S)^{1/\alpha}$ for different values of exceedances.

\begin{figure}
\begin{center}{\includegraphics[scale=0.8]{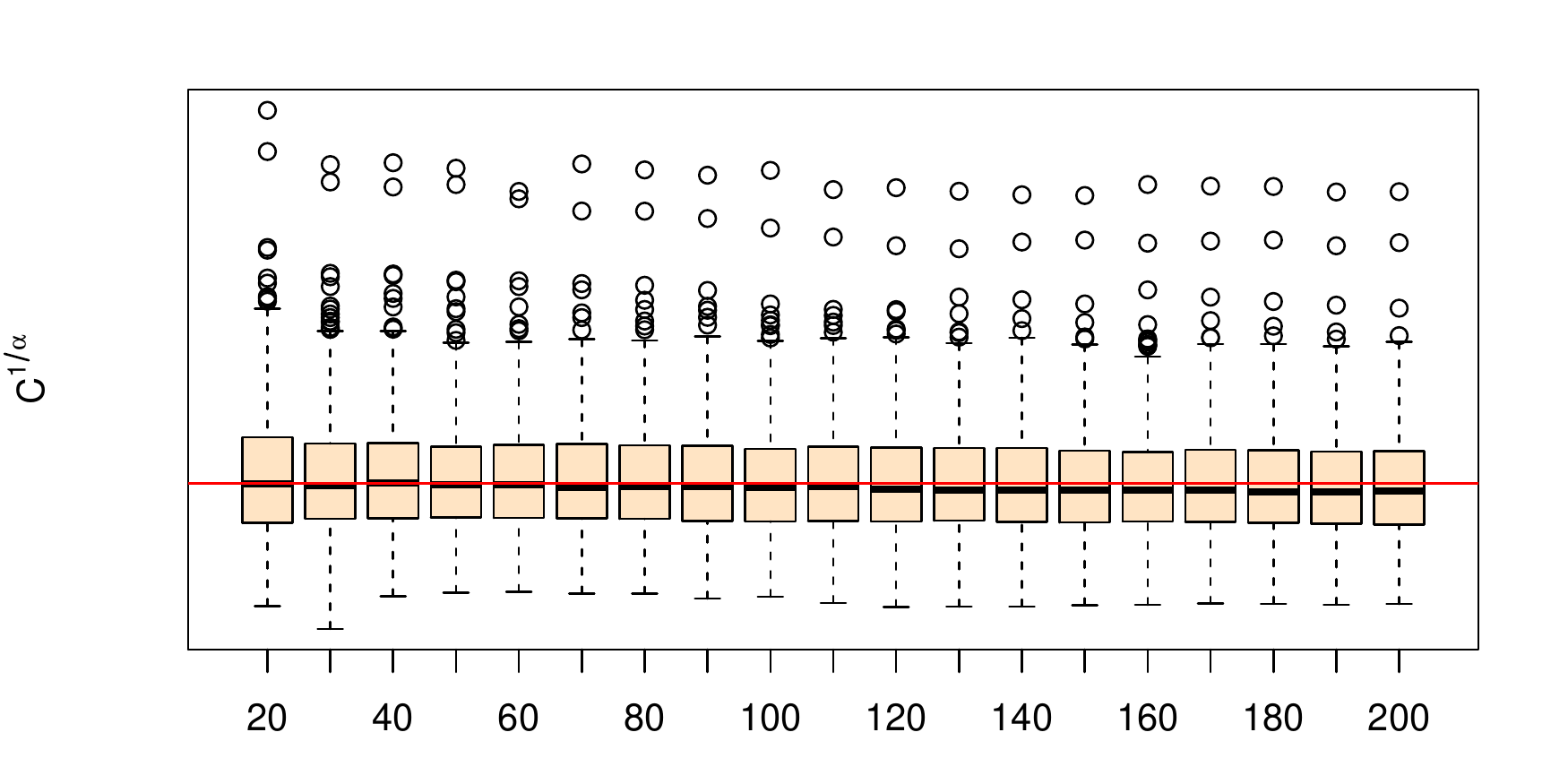}}
\end{center}
\vspace*{-1.0cm}
\caption{\label{MIXsimple}Boxplots  for the estimated VaR constant $(C^S)^{1/\alpha}$  from \eqref{est_ru} with parameters $n=1000, m=500, d=3$, for the gpd $\xi=1/\alpha=0.5$,  $\beta=(1,5,50)$, $u=(1,10,50)$ and for the lognormal  $\mu=(0,1,2)$ and  $\sigma=(1,2,4)$ (on the log-scale). 
On the horizontal axis the number of radial part exceedances $k \le n$ used for dependence estimation.
{The red line is the theoretical asymptotic value.}}
\end{figure}

We find that the median of the estimates is very stable over the entire range of the considered radial part exceedances. 
The proposed estimation method is robust at median level also in case only a small number of exceedances is available, which is the typical case for operational risk. In fact, the red horizontal line represents the theoretically computed value for the parameters used in the simulation, which is always very close to the {medians} of the different boxplots.   
The same interpretation applies to the situation displayed in Figure~\ref{MIXrealParameters}, where the set-up is defined in Parameter Scenario 2) and parameters are estimated by using the real world data.  
{We observe that  the values between $k=90$ and $k=120$ for the exceedances provide a good fit to the theoretically computed value represented by the red horizontal line in Figure~\ref{MIXsimple}. 
The slight decrease in the medians for smaller values of exceedances $k$ appears to be in correspondence also to the curve for $(C^S)^{1/\al}$ in Figure~\ref{OPdatawithoutPOT} based on the real world data.
Moreover, there is a bias for larger values of $k$. 
This is in agreement with common knowledge in different situations in extreme value statistics, for more details see e.g. \citet[Section~4.2.2]{Beirlant}. }

\begin{figure}
\begin{center}{\includegraphics[scale=0.8]{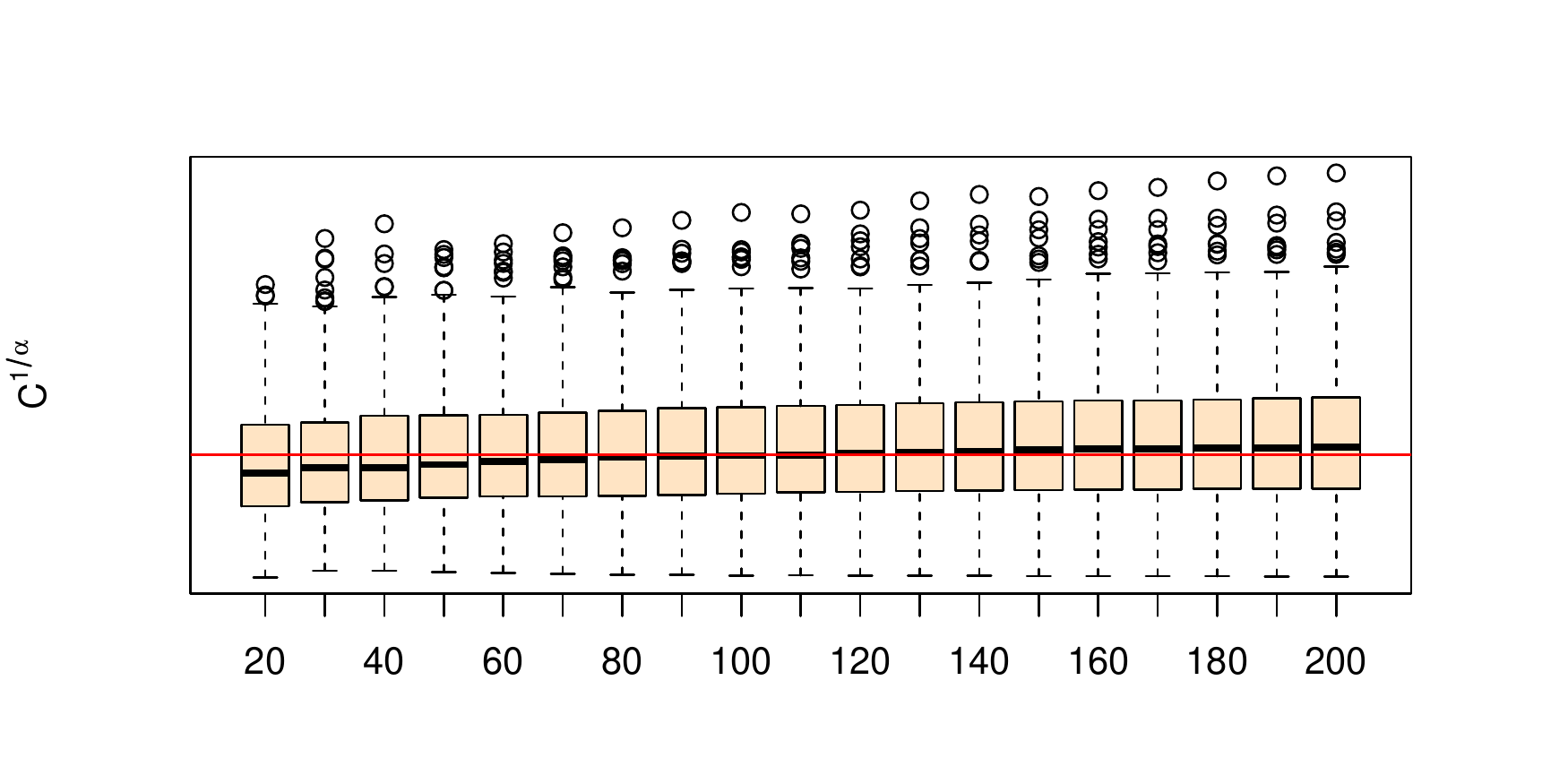}}\end{center}
\vspace{-1.0cm}
\caption{\label{MIXrealParameters} 
Boxplots  for the estimated VaR constant $(C^S)^{1/\alpha}$ with parameters $n=1000, m=500, d=7$, and parameters for the gpd  and lognormal estimated from DIPO real-world data for the marginals reported in Table~\ref{tab:parest} in Section~\ref{s32}. 
On the horizontal axis the number of radial part exceedances  $k \le n$  used for dependence estimation.
 The horizontal values cannot be reported for confidentiality reasons. 
 The red line is the theoretical value {for asymptotically independent marginals.}}
\end{figure}

\subsection{{Individual Risk and Capital Allocation for Different Network Scenarios}}

Besides estimating the VaR and consequently the risk capital  for the entire system,  it is important to disentangle the different components and  study how much risk is concentrated in some business lines, and what is the risk contribution of a business line to the overall risk, after taking into account the network/dependence structure by means of the network fraction matrix $A$. Hence, we apply the following rationale: whenever a loss of a particular event type appears, it is distributed to the different business lines by having each one taking over a certain fraction of that loss. Since a priori there is hardly any {determinism on which fractio}n of a certain event type loss is assigned to some business line, we model this procedure of distributing the event type losses as a stochastic phenomenon.

 Intuitively, we can  think of an underlying weighted bipartite graph where the event types are one type of nodes and the business lines form the other type of nodes. As described in Section \ref{intuitive}
the fractions of the matrix $A$ are given by $A_{ij} = L_{ij}/(L_{1j} + \dots + L_{8j}).$  
They are computed as the fractions of losses from the $j$-th event type that occured in the $i$-th business line as in Figure~\ref{bipartite}. 

In the remainder of this section, we report results from  the simulation set-up in  Parameter Scenario 2), which refers to the real-world estimates in Section~\ref{subsection:simulation:aggregated} (Results for Parameter Scenario 1) are available upon request). Then, we distinguish between two network scenarios: 
we consider the case where the losses are distributed according to a \textit{homogeneous} law for the bipartite graph (see \citet{OKCKLUGR}, Example~4.2.2) and constrast this model with the empirically observed network structure (i.e. \textit{empirical}), where  $A_{ij}$ is the actual  fraction of the loss of event type $j$,  which is allocated to business line $i$. 
We summarize the details as follows.
\begin{itemize}
	\item {Network Scenario 1):}  
	For the \textit{homogeneous} network model, we denote by $\1\{i \sim j \}, i=1,\dots,8,j=1,\dots, 7, $ independent $Bernoulli(p)$ random variables for some fixed connectivity parameter $p \in (0,1)$ independent of either $i$ or $j$. If $\1\{i \sim j \}=0$, nothing of the loss occuring in event type $j$ is distributed to business line $i$. In contrast, if $\1\{i \sim j \}=1$, the loss fraction of $\frac{1}{\deg(j)}$ is attributed to business line $i$, where $\deg(j)=\sum_{i=1}^{8}\1\{i \sim j \}$ is the number of degrees or existing links. Hence, the loss emerging in event type $j$ is distributed in equal parts to all business lines receiving a non-zero loss. As a consequence, the random matrix $A$ with components 
\begin{gather}
A_{ij}=\frac{\1\{i \sim j \}  }{\deg(j)}\quad \mbox{with } \frac{0}{0}:=0,
\end{gather} describes the assignment of the event type losses to the business lines in the homogeneous model. {In this simulation we use a connectivity  parameter $p=0.8$. }
\item {Network Scenario  2):} From the weekly aggregated DIPO data we have 575 observations of $ET-BL$ loss matrices, from which we calculated the fraction matrices $A$ for each week. We consider these fraction matrices as independent realizations $A^{(l)},l=1,\dots, 575,$ of some underlying network law. 
From Section~\ref{s32} we know that we can assume independence of $A$ and the vector \ET.
\end{itemize}

\begin{figure}
\subfigure{\includegraphics[width=8.2cm]{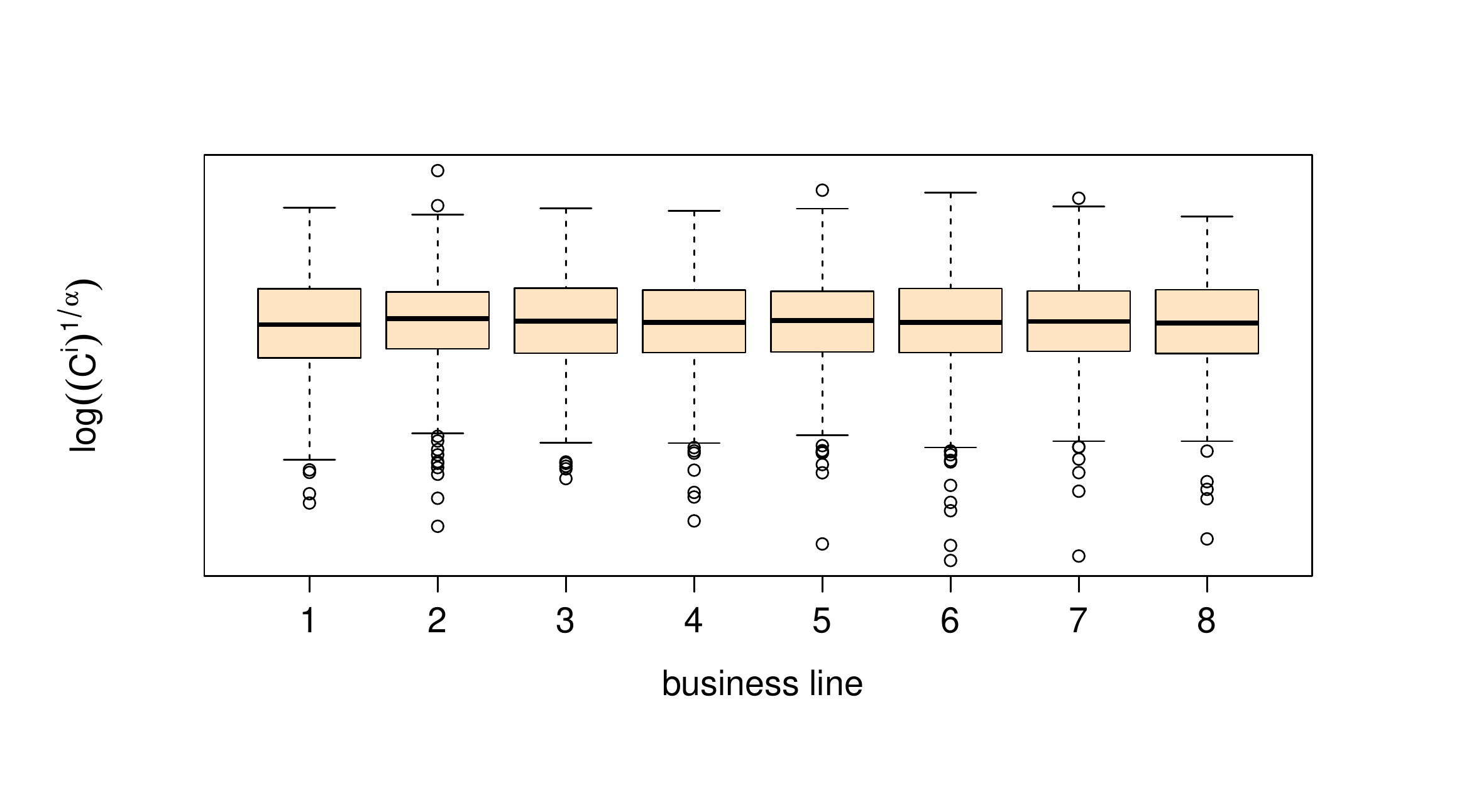}}\hfill
\subfigure{\includegraphics[width=8.2cm]{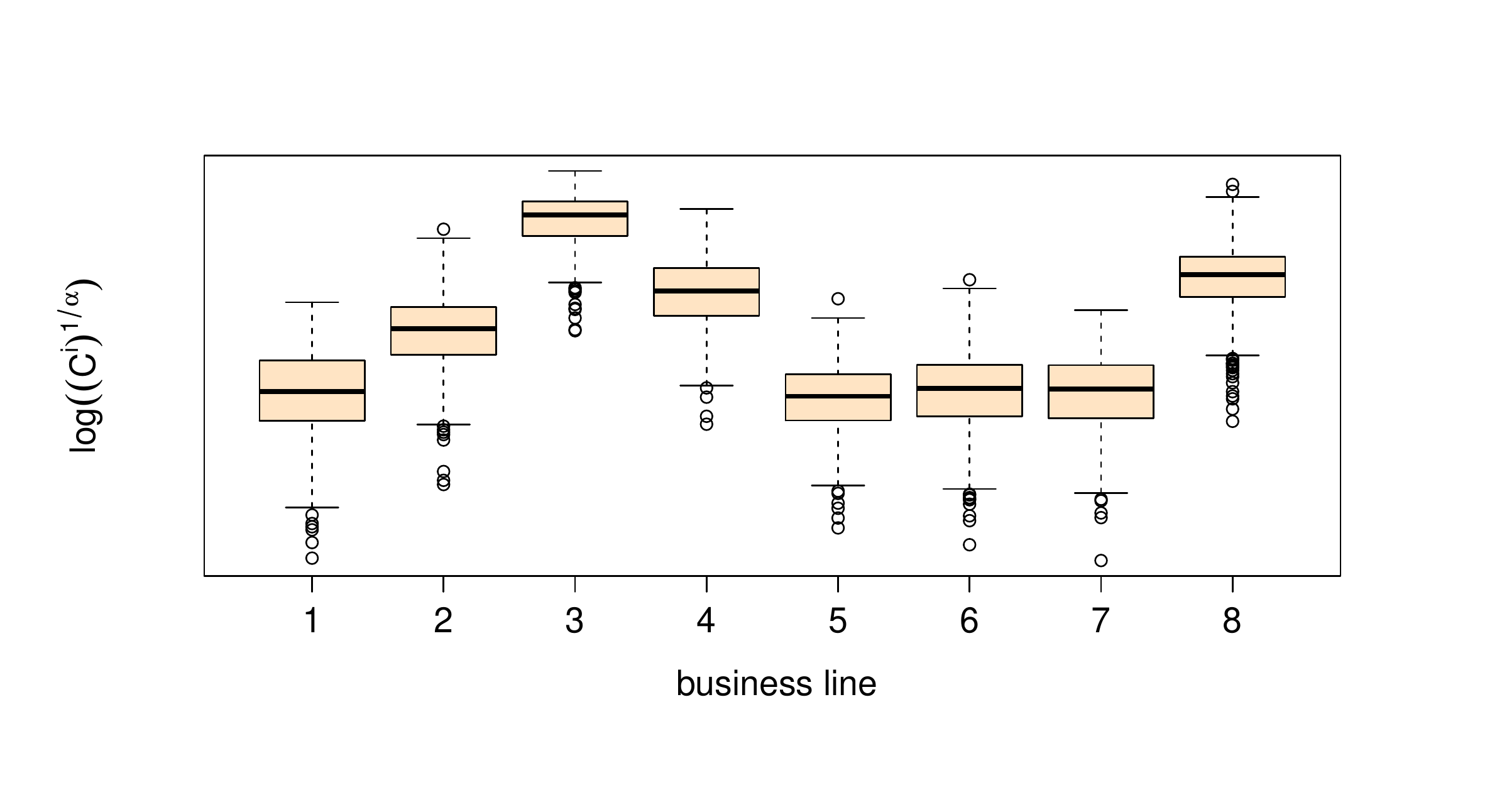}}
\vspace{-1.3cm}
\caption{ \label{CiBL} 
Boxplots of individual VaR constants $(C^i)^{1/\alpha}$  (on the log scale)  
(left for homogeneous networks, right for empirical networks estimated from the DIPO data) when considering $k=100$ as the number of exceedances used for dependence estimation.  
On the horizontal axis are the eight business lines. The horizontal values cannot be reported for confidentiality reasons.} 
\end{figure}

\begin{figure}
\hspace*{-5mm}
\subfigure{\includegraphics[width=9cm]{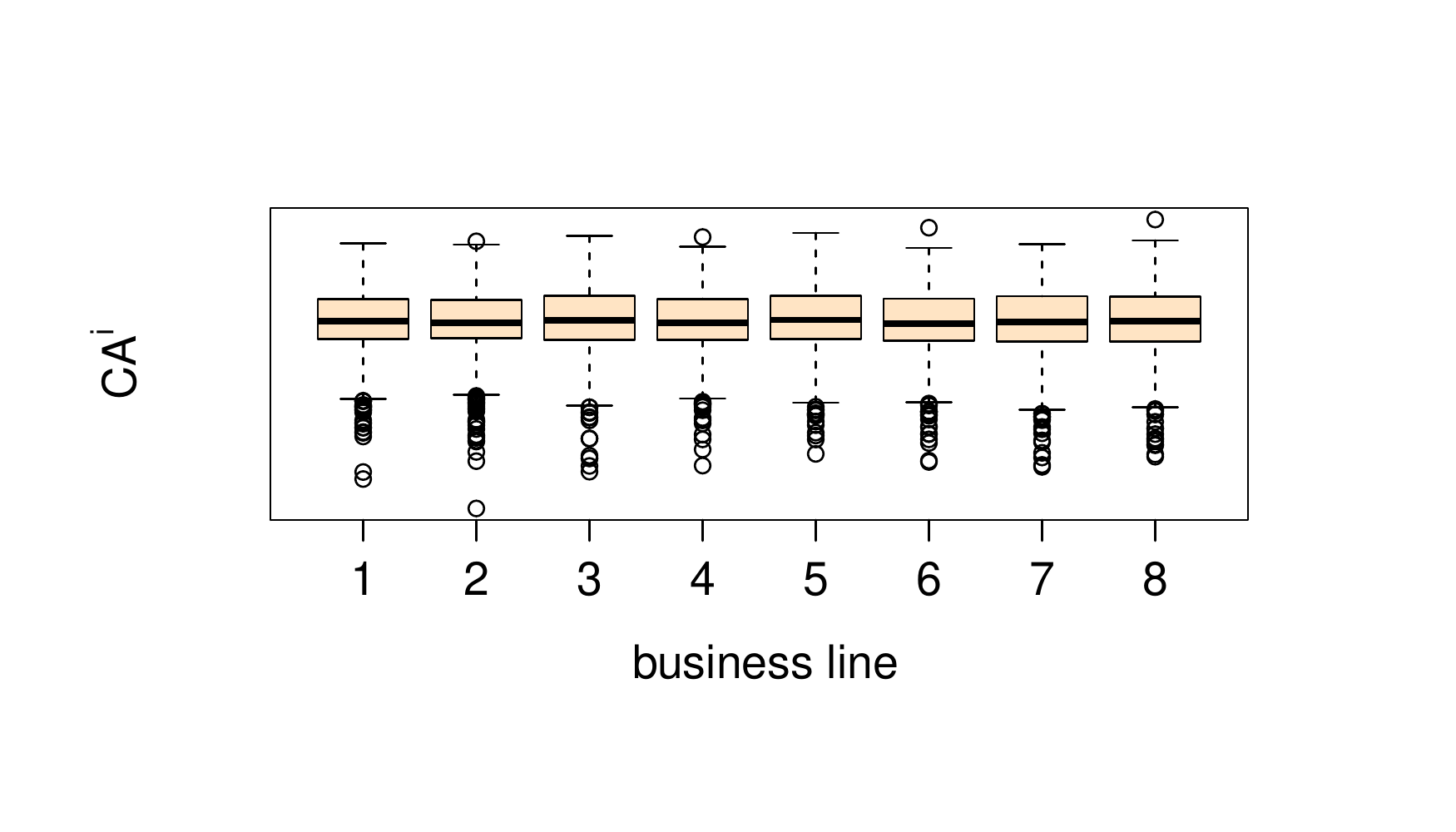}}\hspace{-8mm}
\subfigure{\includegraphics[width=9cm]{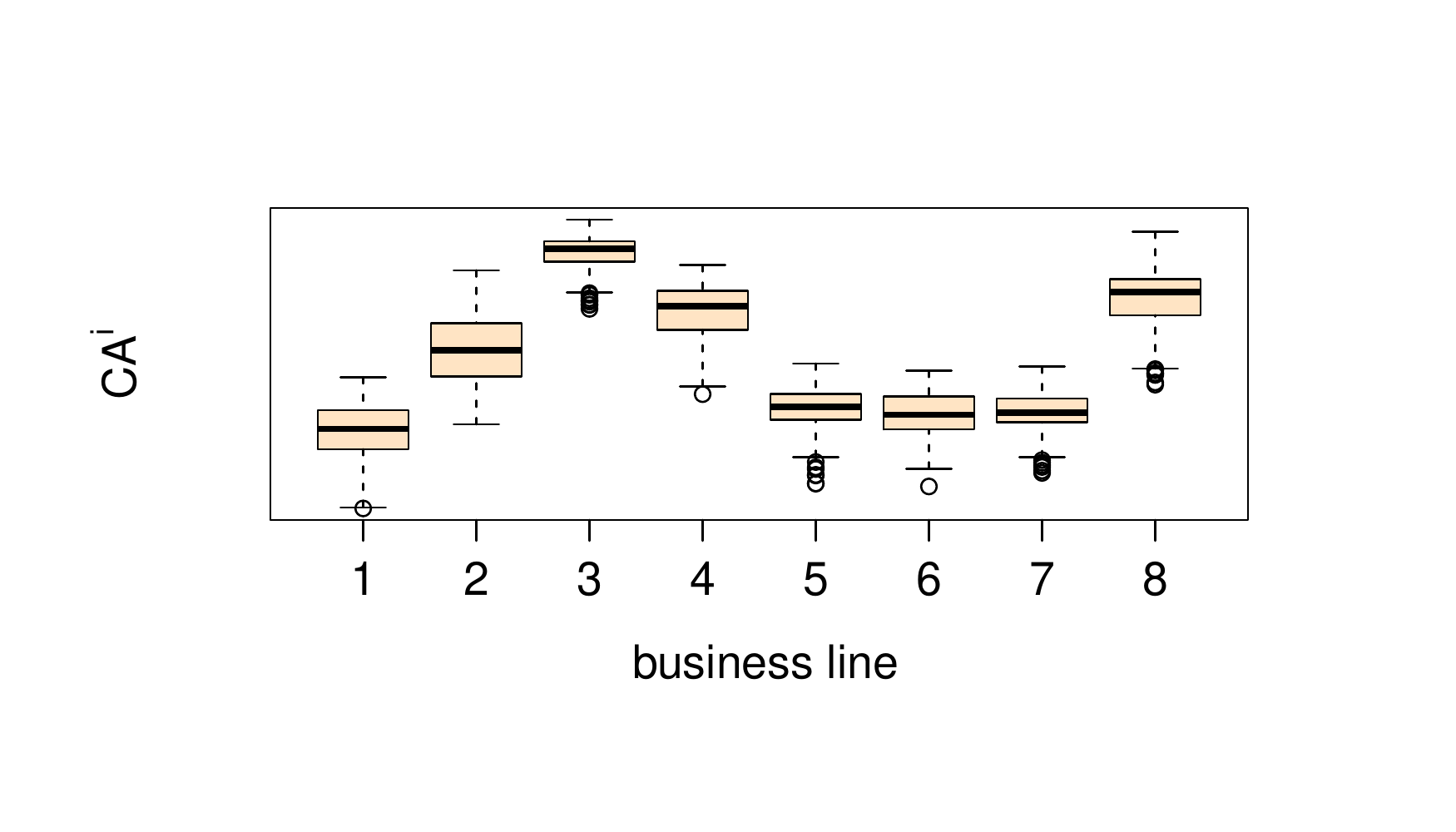}}
\vspace{-1.3cm}
\caption{ \label{CABL} 
Boxplot of the {risk contributions $CA^i$ as in \eqref{riskcontvari}}  estimated with POT method (left for homogeneous networks, right for empirical networks computed from the DIPO data) when considering $k=100$ as the number of exceedances.  On the x-axis the number of the business line,  while on the y-axis the value of  VaR constant  at individual level { $CA^i$}. 
y-scaling cannot be reported for confidentiality reasons. }
\end{figure}

In Figure~\ref{CiBL} we compare the estimates of the individual risk constants $(C^i)^{1/\alpha}$ for the homogeneous model in the left-hand plot (Network Scenario 1), and for the DIPO networks in the right-hand plot (Network Scenario 2). 
As to be expected, in the case of the homogeneous model, the estimators do not substantially differ between different business lines, since every business line has theoretically the same probability to be exposed to losses from a particular event type. 
In contrast to this, we recognize much more variability in the estimates based on the empirical networks estimated from the DIPO data, where the dependence structure is explicitly included, taking into account the unequal distribution of the frequency and the severity of losses in the $ET-BL-$matrix (see Section~\ref{opRiskData} for a description of the data).

In Figure~\ref{CABL} we plot the results for the capital allocations to the different business lines. 
Again, we see on the left-hand side the estimates for the homogeneous model, where capital allocations to business lines are as expected very much the same, whereas in the right-hand plot  for the empirical DIPO networks the estimates differ substantially among the business lines. 

Comparing the right-hand plots of Figures~\ref{CiBL} and~\ref{CABL}, we find that they also look very similar, though the one in Figure~\ref{CiBL} is on log-scale, while the other in Figure~\ref{CABL} is not. {The differences between the $C^i$  are larger than those between the capital allocation constants $CA$. 
This is due to the fact that the capital allocations should sum up to the total VaR (where the event types are simulated independently). In contrast, summing up the VaRs for the single business lines would lead to the total VaR for fully dependent event types, which typically is known to result in a huge overestimation} whenever the event type variables have finite mean; i.e,  we have a tail index $\alpha>1$.

\section{Conclusion}\label{sect:conclusion}

Being capable of properly estimating risk-capital numbers is high on the agenda of regulators and practitioners, especially after the recent crises, as it has become even more evident that adequate reserves are needed to avoid spillovers effects to the entire system with potentially dramatic consequences.

Operational risk management aims to set provisions aside for very heterogeneous event types, from earthquakes to fraud, which might affect all the different business lines of an institution. 
Despite explicit regulation was introduced after market and credit risk, the modelling of operational risk has seen an increased interest due to the magnitude of losses within such a category as well as the need to find appropriate tools to capture the salient features of such 56 dimensional $BL-ET$ loss distribution,  characterized by heavy tails and heterogeneous dependence, as well as the requirement to estimate reliable risk measures at very high confidence levels in presence of missing data and very short times series of operational losses.

In our paper we introduce a theoretical model that naturally incorporates most of the relevant characteristics of operational risk: by relying on heavy-tailed distributions and bipartite graphs we address the presence of extreme losses by heavy-tailed distributions and heterogeneous dependence by a bipartite graph, which also naturally allows us to deal with missing data. 
Using multivariate regular variation we can derive asymptotic formulas for tail-risk measures at high confidence levels  and also provide bounds, while also being capable of suggesting two hands-on and reliable estimation methods. 
Results for simulated data and real-world data support the validity of the proposed approach, while also pointing out suggestions for further research. 

{In this paper, we rely on the assumption of a common tail index $\alpha$ for all event types, which we have proven also to be a statistically reliable choice with regards to the DIPO data set. 
We have also indicated how to deal with decidedly different tail indices, namely to run the very same analysis after transforming the marginal data to the same tail index as is standard in copula models.

There are two suggestions to conduct stress tests within the bipartite framework. 
First, one could move away from asymptotic independence between the event types, which results in higher VaR's in view of the provided bounds in Section~\ref{model}, on both an individual business line level as well as on an aggregated bank level. 
Second, one could vary the network fraction matrix $A$ to stress particular business lines.}
As a priority of future research we want to use additional information about certain operational loss events that are classified by DIPO to be systemically relevant and hence could enrich a model of interbank credit contagion as a particular shock component to the asset side of the balance sheet. 
We also consider developing a model for pricing financial instruments related to the operational tail risks.

\appendix

\section{Appendix}\label{appendix}

\subsection{Mathematical background} \label{Appendixmodel}

Equivalent formulations of multivariate regular variation as defined in \eqref{multrv} are provided in \citet[Theorem~6.1]{Resnick2007}.
For ease of notation,
for a multivariate regularly varying vector $X$ with tail index $\alpha>0$ we write $X\in \MRV(-\alpha)$.

One of these equivalent formulation is the following.
Given a norm $\|\cdot\|$---we take the sum norm if not indicated differently---$X \in \MRV(-\alpha)$ is equivalent to vague convergence 
\begin{gather}
\frac{\pr{X \in t\cdot }}{\pr{\|X\|>t}} \stv \mu(\cdot),\ t \to \infty, 
\end{gather}
(where $\stv$ denotes vague convergence)
to a non-null Radon measure $\mu$, which is called {\em intensity measure}.
The measure $\mu$ is homogeneous of order $-\alpha$; i.e., $\mu(uB)=u^{-\alpha }\mu(B)$ for all Borel sets $B\in \mathcal{B}(\R_+^{d}\setminus\{0\})$.

Another equivalent formulation requires a scaling
function $b: \R_+\rightarrow \R_+$ such that
\begin{gather}\label{exponentmeasure}
t\pr{ \frac{X}{b(t)} \in \cdot } \stv \nu(\cdot),\ t\to \infty,
\end{gather}
to a non-null Radon measure $\nu$ holds. 
Assuming \eqref{Pareto} holds true we can and do choose $b(t)=t^{1/\alpha}$ as a scaling function for $X$ in \eqref{exponentmeasure}. 
In this case, we call the limit measure $\nu$ the \textit{exponent measure} of $X$.

Note that we can then recover the constants $K_j, j=1,\dots, d,$ from \eqref{Pareto} by inserting a set 
\begin{gather}B_j=(0,\infty)\times \dots \times(0,\infty)\times \underbrace{(1,\infty)}_{\mbox{$j$-th\ comp.}}\times (0,\infty)\times \dots \times (0,\infty)  \end{gather} 
in \eqref{exponentmeasure}.
Then, we have  
\begin{gather}\label{defKjs}
t\pr{X_j>t^{1/\alpha}}=\pr{\frac{X}{b(t)} \in B_j} \rightarrow \nu(B_j) =:K_j,\  t\to \infty.
\end{gather}
The measures $\mu$ and $\nu$ are connected via $\mu(\cdot)=\frac{\nu(\cdot)}{\nu(\{x:\|x\|>1\})},$ see Lemma~2.3 in \citet{OKCKLUGR}.

There is also a connection between $\mu$ and $\nu$ and the spectral measure $\Gamma$ in \eqref{multrv}.
For the positive unit sphere $\mathbb{S}_+^{d-1}$ in $\R^d_+$ define a mapping $\pi: \R_+^{d}\setminus \{0\}\rightarrow (0,\infty)\times \mathbb{S}_+^{d-1}$ as the projection to polar coordinates, $\pi(x):= (\|x\|, \frac{x}{\|x\|})$, and fix a measure $\lambda_\alpha$ on the one-dimensional Borel sets $\mathbb{B}(\R_+)$
by $\lambda_\alpha((x,\infty])=x^{-\alpha}.$ 
We define $\circ$ as the usual composition of mappings and $\otimes$ the product of measures (corresponding to independence).
Then
\begin{gather} \label{spectral}
\nu\circ \pi^{-1} = \lambda_\alpha\otimes \ov\Gamma\quad\mbox{and}\quad
\mu\circ \pi^{-1} = \lambda_\alpha\otimes \Gamma.
\end{gather}
with spectral measures $\ov\Gamma$ and $\Gamma$ on the  sphere, where $\Gamma$ is standardized to be the probability measure as defined in \eqref{multrv}.

As stated before, our interest is in the asymptotic behaviour of tail-related risk measures not only of $X$ itself but also of  transformations by random matrices and homogeneous aggregation functions. For this purpose, a random transformation matrix $A: \R^{d} \rightarrow \R^{q}$ should be (stochastically) independent of the random vector $X$  and satisfy the moment condition 
\begin{gather}\label{momentA}
\E \sup_{\|x\|=1} \|Ax\|^{\alpha + \delta} < \infty
\end{gather}
for some $ \delta > 0 $ in order to apply a Breiman-type limit result. The matrix $A$ alters the original dependence structure between the components of the random vector $X$. 

Such a matrix can be directly read off from the operational risk data by considering  the network fraction matrix that distributes the event type losses to the different business lines.  
Considering also alternative (random) matrices instead of this given matrix $A$ can be a way to stress test the system and evaluate the impact of dependence on the tail-risk measures estimates. 
We also want to emphasize that for every network fraction matrix condition \eqref{momentA} is obviously satisfied.
For more details regarding real-world applications we refer to Section~\ref{opRiskData}.  

\begin{theorem}\label{prop:varasym}
Let $X\in\R_+^d$ be $\MRV(-\alpha)$  and $\nu$ its exponent measure as in \eqref{exponentmeasure}.
Let $A$ be a random transformation matrix of dimension $q\times d$ as in \eqref{BPhiT}, independent of $X$ and satisfying \eqref{momentA}.
Furthermore, let $h:\R_+^{q}\rightarrow \R_+$ be a $1-$homogeneous function. We define
the constant
\begin{gather}\label{VaR.const}
C_\Gamma^{h}(A):= 
\E\Big( \int_{\mathbb{S}_{+}^{d-1}}  h(As)^\alpha d\ov\Gamma(s)\Big).
\end{gather}
Then the asymptotic behaviour of the VaR of $h(AX)$ is given by
\begin{gather}\label{VaRasym}
\VaR_{1-\gamma}(h(AX))\sim \gamma^{-1/\alpha} (C_\Gamma^{h}(A))^{1/\alpha},\quad\ga\to 0,
\end{gather} 
and the asymptotic behaviour of the CoTE of $h(AX)$--- provided $\alpha>1$---is given by
\begin{gather}\label{CoTEasym}
\CoTE_{1-\gamma}(h(AX))\sim \frac{\alpha}{\alpha-1}\gamma^{-1/\alpha} (C_\Gamma^{h}(A))^{1/\alpha},\quad\ga\to 0.
\end{gather}
\end{theorem}

\begin{proof}[\textit{Proof of Theorem~\ref{prop:varasym}}] 
We apply Proposition~A.1 in \citet{Basrak200295}, for which we need independence of $A$ and $X$ as well as \eqref{momentA}.
From there we know the first equality below, the second follows from \eqref{spectral}, giving
\begin{gather*}
\lim_{t\to\infty}  t^{-\alpha}\pr{h(AX) >t } = \E \nu(\{x:\ h(Ax) >1   \}) =  \E \int_{\mathbb{S}_{+}^{d-1}}  h(As)^\alpha d\ov\Gamma(s).\end{gather*}
Now \eqref{VaRasym}  may be easily derived using standard arguments, in particular, that inverses of
regularly varying functions are again regularly varying (cf. \citet[Prop.~1.5.15]{BGT1987}). \\
The asymptotic behaviour of $\CoTE$ in \eqref{CoTEasym} is 
then a consequence of a Karamata-type argument (cf. \citet[Theorem~1.6.5]{BGT1987}).
\end{proof}

The function $h(\cdot)$  plays a role in determining at what level {to compute the risk-measures}. In particular,  
two  choices for $h$ in Theorem~\eqref{prop:varasym}  are important: first,  the sum norm $h(x)=\|x\|=\sum_{i=1}^q|x_i|$, which reflects the aggregated risk or risk at the entire \textit{system} level (i.e. $S$). 
For this sum norm, we will use the notation $C^S$ instead of $C^h$. 
Second, for the projection on the $i-$th coordinate, we set $h(x)=x_i$, and we write $C^i$ instead of $C^h$, reflecting the risk at an \textit{individual} level, which for operational risk then refers to the single business lines (or event types). 
Both measures are of interest, as the constant $C^S$ should provide a single metric of systemic risk that encompasses all institutions under consideration, but also, it should be possible to apply them to any subset of institutions in the system down to the level of a single {business line or agent}, depending on the case study and the available data. Moreover, both regulators and risk managers need to have methods for the allocation of risk to business lines/event types according to their firm-wide importance to better quantify their risk appetite and tolerance, and then set-up adequate monitoring, insurance and hedging strategies.

In the following, we  reformulate the risk constants from \eqref{VaR.const} and the risk allocation constants for general aggregation function $h$ in such a way that we can introduce reasonable estimators, which perform well in practical situations. 
In particular, we take into account that the constants $K_j$ from \eqref{Pareto} can differ considerably among the components of the risk vector $X$.
Marginal scaling resulting in $K_j=1$ for all $j=1,\dots,d$ 
affects the integral \eqref{VaR.const} as a simple transformation of variables, so that it affects the dependence structure in terms of the spectral measure $\Gamma$ in an obvious way. 
Consequently, we consider the scaled data components $K_j^{-1/\alpha}X_j$ for which we have $\P[K_j^{-1/\alpha}X_j>t ]\sim t^{-\alpha}$ for $j=1,\dots,d.$

%

\begin{theorem}\label{prop.estimator.basic}
Let $X\in\R_+^d$ be $\MRV(-\alpha)$ and let the assumptions of Proposition~\eqref{prop:varasym} hold with two 1-homogeneous functions $h,g:\R^{q}\rightarrow \R$.   
We further write  $K^{1/\alpha}: =diag(K_1^{1/\alpha},\dots, K_d^{1/\alpha})$ with $K_j$ for $ j=1,\dots,d$ as in \eqref{defKjs}. 
Furthermore, let $\Gamma_{K}$ be the spectral measure  as  in \eqref{multrv} of the scaled vector $K^{-1/\alpha}X.$ Then we obtain 
\begin{gather}\label{A.10}
\E  \int_{\mathbb{S}_+^{d-1}}  h(As)^{\alpha-1}g(As)  \ov\Gamma(ds)=\frac{\E  \int_{\mathbb{S}_+^{d-1}}  h(AK^{1/\alpha}s)^{\alpha-1}g(AK^{1/\alpha}s)\Gamma_{K}(ds) }{\int_{\mathbb{S}_+^{d-1}}  s_1^{\alpha} d\Gamma_{K}(s)}
\end{gather}
For the special case $h=g$ we find
\begin{gather}
 C_{\Gamma}^{h}(A)=\frac{\E \int_{\mathbb{S}_+^{d-1}}  h(As)^{\alpha} d\Gamma_{K}(s) }{\int_{\mathbb{S}_+^{d-1}}  s_1^{\alpha} d\Gamma_{K}(s)}.
\end{gather} 
For $h$ equal to the sum norm this applies to the right-hand factor in  \eqref{riskcontvari}.
\end{theorem}

\begin{proof}[\textit{Proof of Theorem~\ref{prop.estimator.basic}}]
Using the notation for the image measure $\nu_K:=\nu\circ K^{1/\alpha}$ and  $\ov\Gamma_K$ gives \eqref{A.10}.
The relation between $\ov\Gamma_K$ and $\Gamma_{K}$, see \eqref{spectral}, is given by 
$\Gamma_{K}=\frac{\ov\Gamma_K}{(\nu_K(\|x\|>1))^{-1}}$. 
Since
\begin{gather}\label{norming}
\frac{1}{\nu_K(\|x\|>1)}=\lim_{t\to \infty}\frac{ \pr{(K^{-1/\alpha}X)_1>t  }}{\pr{ \|K^{-1/\alpha}X\|>t}}= \mu_K(x_1>1)= \int_{\mathbb{S}_+^{d-1}}  s_1^{\alpha} d\Gamma_{K}(s),
\end{gather}
 the assertion follows immediately. 
 {Note that the choice of the first component in \eqref{norming} is arbitrary, since $\nu_K(x_j>1)=1$ for all $j=1,\dots,d$.}
\end{proof}

As a last step in this section, we present the proof of Theroem~\ref{pr:riskcontvar}.

\begin{proof}[\textit{Proof of Theorem~\ref{pr:riskcontvar}}]
Similarly as in \citet{Tasche2008, kalkbrener}, the risk contributions in our setting are based on Euler's Theorem for 1-homogeneous functions.
Referring to the constant $C^S$ from \eqref{constantsIS}, we define a function
\begin{gather}
c: \R^{q}\rightarrow \R, x=(x_1,\dots,x_q)^{\top}\mapsto \Big( \E  \int_{\mathbb{S}_+^{d-1}}  \| \operatorname{diag}(x_1,\dots, x_q) As\|^{\alpha}\ov\Gamma(ds) \Big)^{1/\al}
\end{gather}
which is obviously 1-homogeneous, hence, Euler's Theorem applies and  we can write
\begin{gather}
c(x)= \sum_{i=1}^{q} x_i \frac{\partial}{\partial x_i}c(x),\quad x\in\R^q.
\end{gather}
We compute the partial derivatives for $i=1,\dots,q$ and the sum norm  by
\begin{align*}
&\frac{\partial}{\partial x_i}c(x)=  \frac{\partial}{\partial x_i}\Big( \E \int_{\mathbb{S}_+^{d-1}} \| \operatorname{diag}(x_1,\dots, x_q) As\| ^{\alpha} \ov\Gamma(ds)\Big)^{1/\al}\\
&= \frac{1}{\al} \Big( \E \int_{\mathbb{S}_+^{d-1}} \| \operatorname{diag}(x_1,\dots, x_q) As\| ^{\alpha} \ov\Gamma(ds)\Big)^{1/\al-1}\\ 
&\quad \times \Big( \E \int_{\mathbb{S}_+^{d-1}} {\al} \big( \| \operatorname{diag}(x_1,\dots, x_q) As\|   \big)^{{\alpha}-1}  (As)_{i} \, \ov\Gamma(ds) \Big),
\end{align*}
and get in particular
\begin{gather}
\frac{\partial}{\partial x_i}c(1,\dots,1)= (C^{S})^{1/\al-1} \Big( \E \int_{\mathbb{S}_+^{d-1}}
\| A\| ^{\alpha-1}   (As)_{i} \, \ov\Gamma(ds)  \Big).
\end{gather}
All in all, we have for $\VaR_{1-\gamma}(Y_i|\|Y\|)= \gamma^{-1/\al}(C^{S})^{1/\al-1} \Big( \E \int_{\mathbb{S}_+^{d-1}}
\| A\| ^{\alpha-1}   (As)_{i}  \ov\Gamma(ds)  \Big) $
the full asymptotic risk allocation property as stated in \eqref{fullAllocation}. 
This translates to \eqref{alloc} for the business lines with risk contributions \eqref{riskcontvari}.\\
The same arguments apply to find the risk contributions of $\CoTE$. 
\end{proof}

\subsection{Statistical background}\label{subsec:estimation_procedure}

Theorem~\ref{prop.estimator.basic} suggests a semi-parametric estimation procedure for the VaR constant $C_{\Gamma}^{h}(A)$ from \eqref{VaR.const} in two steps. Before we explain the two steps in details, we want to
recall also that the constant  $C_{\Gamma}^{h}(A)$ is a key quantity to be estimated as then the VaR from \eqref{varconstants} and the CoTE from \eqref{coteconstants} can be directly computed by scaling the risk measures by the confidence level $\gamma$ and the tail index $\alpha$.

For the first step, we estimate the marginal parameters; i.e., the tail index $\alpha$ as well as the constants {$K_j,\ j=1,\dots,d$, from \eqref{Pareto}}. 
Given the marginal estimators, in a second step we estimate the dependence structure determined by both, the random matrix $A$ as well as the spectral measure $\Gamma_K$ of the scaled observations given by $K^{-1/\al}X$.

Throughout we assume independent multivariate observations $x^{(1)},\dots,x^{(n)}$  of a vector 
$X\in\MRV(-\al)$, representing here the weekly event types $ET=(ET_1,\dots,ET_7)$. 

\subsubsection*{Marginal estimation}

We use the well--known POT--method (\textit{peaks--over--threshold}) as the basis for the marginal estimation; for details see e.g.  \citet[Section~6.5.1]{EKM1997}. 
	The tail of the  generalized Pareto distribution (gpd) is denoted by $\ov G_{\xi,\beta, u}$.
	In the heavy-tailed case it has positive shape parameter $\xi$, which is connected to the tail index $\alpha$ by  $\xi=1/\alpha$, scale parameter $\beta>0$ and location parameter (threshold) $u$.
	It has the following form:
\begin{gather}\label{gpd}
\ov G_{\xi,\beta, u}(t)= \left(  1+ \frac{\xi(t-u)}{\beta}  \right)^{-1/\xi} = \left(  1+ \frac{t-u}{\alpha\beta}  \right)^{-\alpha}
=\ov G_{1/\alpha,\beta, u}(t),\quad  t > u.
\end{gather}
To find the location parameters $u_j$ for $j=1,\dots, d$; i.e., the threshold value above which a good approximation by a generalized Pareto distribution is granted.
We use an automatic approach taking a bias--variance trade-off into account as presented in \citet{mager_paper}, based on a robust method suggested in \citet{Dupuis}. 

Once having found the optimal thresholds for each event type component, only those observations are used for ML-estimation of $\xi$ and $\beta$ that are above this threshold. 
This certainly leads to different estimates for $\xi_1,\dots,\xi_d$ as well as for $\beta_1,\dots \beta_d$. 
Since for multivariate regular variation one main assumption is that all event types have the same tail indices, we have to go one step further. 
One way to achieve this requirement would be to take the different estimates $\wh\xi_1,\dots,\wh\xi_d $ and transform the marginals to, for instance, standard Fr\'echet. 
This is often done, but it has the disadvantage that it gives all marginals the same form, making small components large and large components small, overestimating the influence of small components for a tail risk measure and underestimating large components.
Based on the fact that the estimates $\wh\xi_1,\dots,\wh\xi_d$ are not so far apart (all components have finite first moment, and only one has finite variance), 
we take the mean $\wh\xi_{mean}=\frac1d \sum_{j=1}^d\wh\xi_j $ of all those estimated shape parameters $\wh\xi_j$. 
Notably, we do have the chance to correct this step of obvious wrong-doing thanks to the dependence between the estimators for the scale and shape parameters; see \citet[Section~6.5.1]{EKM1997}. 
Consequently, given the estimated mean $\wh\xi_{mean}$, we re-estimate the scale paramters $\beta_1,\dots, \beta_d $ again by a (conditional) ML-procedure. 
This results in a shape estimate $\wh\xi_{mean}$ and a scale estimate $\wh\beta_j$ all based on a thresold value $u_j$ for each marginal $j$.  

The transformation of the gpd estimates into estimates for the tail index $\al$ and $K_j$ for $j=1,\dots,d$ is a simple transformation.
The tail index estimator  $\wh\al=1/\wh\xi_{mean}$ and for the $\wh K_j$ we 
use \citet[Theorem~3.4.13(b)]{EKM1997}, which gives the following approximation:
\begin{gather}\label{gpd_app1}
\pr{X_j>t} \approx  \pr{ X_j>u_j}\overline{G}_{\xi,\beta_j,u_j}(t), \quad  t> u_j,
\end{gather}for $j=1,\dots, d$. 
Then, as $\overline{G}_{1/\alpha,\beta_j,u_j}(t) \sim  (\alpha\beta_j)^{\alpha} t^{-\alpha}$, we have
\begin{gather}\label{gpd_app2}
\pr{X_j>t} \approx  \pr{ X_j>u_j} (\alpha\beta_j)^{\alpha} t^{-\alpha}, \quad  t> u_j.
\end{gather} 
Estimating the probability on the right-hand side by its empirical counterpart, we obtain for $j=1,\dots,d$,
\begin{gather}\label{K_POT}
\wh K_j= ( \wh{\alpha} \wh\beta_j)^{\wh{\alpha}} \frac{1}{n}\sum_{i=1}^{n}\mathds{1}\{ x_{j}^{(i)} > u_j\}.
\end{gather}

Furthermore, after having estimated the tail index $\al$, we can also use \eqref{defKjs} to estimate the $K_j$ for $j=1,\dots,d$, which works as follows.
Choose $b(t)=t^{1/\al}$ as the natural scaling function in \eqref{exponentmeasure}. 
Then define the estimators
\begin{gather}\label{K-estimatorMVR}
\wh K_j:= \left(\frac{n}{k_j}\right)^{\wh{\al}} \frac{1}{n}\sum_{i=1}^n \mathds{1}\{ x_j^{(i)} \geq  \frac{n}{k_j} \}
\end{gather} 
Here the number $k_j$ of upper order statistics in the $j$th component of the vector ET used for the estimation in \eqref{K-estimatorMVR} and the threshold $u_j$ used for the POT-estimation are related by the fact that $u_j$ is taken as the smallest order statistic $k_j$, which allows for a good approximation in \eqref{gpd_app1}.
To obtain a more robust estimator for $K_j$, we take the mean value over all estimates corresponding to the order statistics in a stable range.

\subsubsection*{Estimation of the risk constants}

The risk constants involve besides the tail index $\al$ and the $K_j$ for $j=1,\dots,d$ also the spectral measure $\Gamma$.
As in Theorem~\ref{prop.estimator.basic} we scale the observations by defining $\wh{K}^{1/\wh{\alpha}} =diag( \wh K_1^{1/\wh{\alpha}},\dots, \wh K_1 ^{1/\wh{\alpha}}  )$ and taking the scaled observations 
\begin{gather}\label{scaled_observations}
y^{(i)}=\wh K^{-1/\wh{\alpha}}x^{(i)}, \quad i=1,\dots,n,
\end{gather} 
such that the sample $y^{(1)},\dots,y^{(n)})$ originate from a multivariate regularly varying distribution with spectral measure $\Gamma_K$ as in Theorem~\ref{prop.estimator.basic}.

Denote the order statistics of the norms of the vectors by 
\begin{gather}\label{radial_parts}
\|y^{(\cdot)}\|_{(1)}\le \dots\le  \|y^{(\cdot)}\|_{(n)}; 
\end{gather}
Since the spectral measure exists only as a limit measure, among all $\|y^{(1)}\|,\dots, \|y^{(n)}\|$, we choose a sufficiently high order statistics, say  $\|y^{(\cdot)}\|_{(k)}$, to ensure a good approximation to the spectral measure.
Then 
\begin{gather}\label{estimatorRho}
\wh\Gamma_{K}(\cdot)= \frac{\sum_{i=1}^{n} \mathds{1}\{(\|y^{(i)}\|, y^{(i)}/\|y^{(i)}\|) \in  [\|y^{(\cdot)}\|_{(k)},\infty] \times \cdot    \} }{ \sum_{i=1}^{n} \mathds{1}\{\|y^{(i)}\| \in  [\|y^{(\cdot)}\|_{(k)},\infty]   \} },
\end{gather}
see \citet[eq. (9.23)]{Resnick2007} and subsequent comments.  
Given marginal estimates $\wh{\alpha}_{mean}$ and $\wh K_j$ for $j=1,\dots,d,$ and given realizations of the random network $A^{(1)},\dots,A^{(n)}$, we  estimate $C_\Gamma^h(A)$ by
\begin{gather}\label{estimatorC}
\wh C_{\Gamma}^{h}(A)= \frac{1}{n}\sum_{i=1}^n \frac{\int_{\mathbb{S}_+^{d-1}}h(A^{(i)}s)^{\wh{\alpha}} d\wh\Gamma_{K}(s)} {\int_{\mathbb{S}_+^{d-1}} s_1^{\wh{\alpha}} \, d\wh\Gamma_{K}(s) }.
\end{gather}
Combining \eqref{estimatorRho} and \eqref{estimatorC} yields 
\begin{gather}\label{final_est_general}
 \wh C_{\Gamma}^{h}(A)= \frac{1}{n}\sum_{i=1}^n \frac{  \sum_{j=1}^n h(A^{(i)}y^{(j)}/\|y^{(j)}\|)^{\wh{\alpha}}\mathds{1}\{\|y^{(j)}\| \in  [\|y^{(\cdot)}\|_{(k)},\infty]  \} }{ \sum_{j=1}^n  (y^{(j)}/\|y^{(j)}\| )_1^{\wh{\alpha}}\mathds{1}\{ \|y^{(j)}\| \in  [ \|y^{(\cdot)}\|_{(k)},\infty]    \} }.
\end{gather}

For the particular case of $A$ being the identity matrix in $\R^d$, and aggregation function $h(x)=\|x\|_1=\sum_{j=1}^{d}|x_j|$ is the sum norm on $\R^d$, the estimator for the systemic constant $C^S$ as in \eqref{constantsIS} is given by 
\begin{gather}\label{est_ru}
\wh C^{S}=   \frac{  \sum_{j=1}^n \mathds{1}\{(\|y^{(j)}\| \in  [\|y^{(\cdot)}\|_{(k)},\infty]  \} }{ \sum_{j=1}^n  (y^{(j)}/\|y^{(j)}\| )_1^{\wh{\alpha}}\mathds{1}\{ \|y^{(j)}\| \in  [ \|y^{(\cdot)}\|_{(k)},\infty]    \} }
\end{gather}
This estimator has already been introduced in \citet{Mainik2010}, where for continuous spectral distribution  also properties like consistency and asymptotic normality have been studied.

If the function $h$ is chosen as the projection on the $i$-th coordinate, we have 
\begin{gather}\label{est_agent}
 \wh C^{i}(A)= \frac{1}{n}\sum_{l=1}^n \frac{  \sum_{j=1}^n (A^{(l)}y^{(j)}/\|y^{(j)}\|)_i^{\wh{\alpha}}\mathds{1}\{\|y^{(j)}\| \in  [\|y^{(\cdot)}\|_{(k)},\infty]  \} }{ \sum_{j=1}^n  (y^{(j)}/\|y^{(j)}\| )_1^{\wh{\alpha}}\mathds{1}\{ \|y^{(j)}\| \in  [ \|y^{(\cdot)}\|_{(k)},\infty]    \} }.
\end{gather}

In order to estimate the risk allocation constants $CA^i$ for the risk contributions from \eqref{riskcontvari}, we estimate the factor $$  \E  \int_{\mathbb{S}_+^{d-1}} \| As\| ^{\alpha-1}  (As)_{i}  \Gamma(ds)  $$ by
\begin{gather}\label{est_agent_riskcont}
 \frac{1}{n}\sum_{l=1}^n \frac{  \sum_{j=1}^n (A^{(l)}y^{(j)}/\|y^{(j)}\|)_i   \|A^{(l)}y^{(j)}/\|y^{(j)}\| \|^{\wh{\alpha}}\mathds{1}\{(\|y^{(j)}\| \in  [\|y^{(j)}\|_{(k)},\infty]  \} }{ \sum_{j=1}^n  (y^{(j)}/\|y^{(j)}\| )_1^{\wh{\alpha}}\mathds{1}\{ \|y^{(j)}\| \in  [ \|y^{(j)}\|_{(k)},\infty]    \} }.
\end{gather}
and combine it with \eqref{est_ru}.

\subsection*{Acknowledgement}

The authors thank Claudia Pasquini, Claudia Capobianco, and Vincenzo Bugge\`e from the Italian Database of Operational Losses (DIPO) and its Statistical Committee for their support. The views expressed in this paper are those of the authors and do not necessarily reflect the viewpoints of DIPO or the DIPO Statistical Committee. 
Sandra Paterlini acknowledges financial support from ICT COST Action IC1408.


\end{document}